\documentclass[11pt]{article}

\usepackage[margin=1in]{geometry}
\usepackage{fullpage}
\usepackage{tgtermes}
\usepackage[T1]{fontenc}
\usepackage[colorlinks,citecolor=blue,linkcolor=blue,urlcolor=red]{hyperref}
\usepackage{graphicx}
\usepackage{amsfonts,amsmath,amsthm,amssymb,dsfont,mathtools,multirow}

\mathtoolsset{centercolon}
\usepackage{xfrac,nicefrac}
\usepackage{mathdots}
\usepackage{bm,bbm}
\usepackage{url}
\usepackage{caption}
\usepackage{paralist}
\usepackage{enumerate}
\usepackage[normalem]{ulem}
\usepackage{enumitem}
\usepackage{xspace}
\xspaceaddexceptions{]\}}
\usepackage[capitalise]{cleveref}
\crefname{algocf}{Algorithm}{Algorithms}
\Crefname{algocf}{Algorithm}{Algorithms}
\crefname{claim}{Claim}{Claims}
\usepackage{tabu}
\usepackage{framed}
\usepackage{float,wrapfig}
\usepackage{subfigure}
\usepackage{soul}
\usepackage[usenames,dvipsnames]{pstricks}
\usepackage{thmtools, thm-restate}
\usepackage[linesnumbered,boxed,ruled,vlined]{algorithm2e}
\usepackage{algpseudocode}
\usepackage{tikz}
\usetikzlibrary{decorations.pathreplacing}
\usetikzlibrary{calc}
\usepackage{regexpatch}
\usetikzlibrary{positioning}
\usetikzlibrary{arrows.meta}

\theoremstyle{plain}

\newtheorem{theorem}{Theorem}[section]
\newtheorem{lemma}[theorem]{Lemma}

\newtheorem{prop}[theorem]{Proposition}

\newtheorem{claim}[theorem]{Claim}
\theoremstyle{definition}
\newtheorem{definition}[theorem]{Definition}
\newtheorem{remark}[theorem]{Remark}

\DeclarePairedDelimiter{\tvdbasic}{\lVert}{\rVert}
\makeatletter
\newcommand{\@tvdstar}[2]{\tvdbasic*{#1 - #2}_{\mathrm{TV}}}
\newcommand{\@tvdnostar}[3][]{\tvdbasic[#1]{#2 - #3}_{\mathrm{TV}}}
\newcommand{\tvd}{\@ifstar\@tvdstar\@tvdnostar}
\makeatother

\usepackage[breakable]{tcolorbox}

\renewcommand{\epsilon}{\varepsilon}

\newcommand{\norm}[1]{\left\lVert #1 \right\rVert}
\newcommand{\midnorm}[1]{\lVert #1 \rVert}

\newcommand{\bk}[1]{\left( #1 \right)}

\newcommand{\bigbk}[1]{\bigl( #1 \bigr)}
\newcommand{\Bigbk}[1]{\Bigl( #1 \Bigr)}

\newcommand{\Biggbk}[1]{\Biggl( #1 \Biggr)}

\newcommand{\Bk}[1]{\left[ #1 \right]}

\newcommand{\bigBk}[1]{\bigl[ #1 \bigr]}

\newcommand{\BK}[1]{\left\{ #1 \right\}}
\newcommand{\midBK}[1]{\{ #1 \}}
\newcommand{\bigBK}[1]{\bigl\{ #1 \bigr\}}

\newcommand{\angbk}[1]{\left\langle #1 \right\rangle}

\DeclareMathOperator*{\E}{\mathbb{E}}

\let\Pr\myPr

\newcommand{\F}{\mathbb{F}}

\newcommand{\defeq}{\coloneqq}
\newcommand{\eps}{\varepsilon}
\newcommand{\T}{\mathcal{T}}
\newcommand{\N}{\mathbb{N}}
\newcommand{\R}{\mathbb{R}}
\newcommand{\Z}{\mathbb{Z}}
\renewcommand{\l}{\ell}

\newcommand{\numberthis}{\addtocounter{equation}{1}\tag{\theequation}}

\newcommand{\tall}{\vphantom{\sum}}

\usepackage{mleftright}
\let\left\mleft
\let\right\mright

\newcommand{\CW}{\mathrm{CW}}

\renewcommand{\split}{\textsf{\textup{split}}}
\newcommand{\Split}{\textsf{\textup{split}}}

\makeatletter
\let\save@mathaccent\mathaccent
\newcommand*\if@single[3]{%
  \setbox0\hbox{${\mathaccent"0362{#1}}^H$}%
  \setbox2\hbox{${\mathaccent"0362{\kern0pt#1}}^H$}%
  \ifdim\ht0=\ht2 #3\else #2\fi
  }
\newcommand*\rel@kern[1]{\kern#1\dimexpr\macc@kerna}
\newcommand*\widebar[1]{\@ifnextchar^{{\wide@bar{#1}{0}}}{\wide@bar{#1}{1}}}
\newcommand*\wide@bar[2]{\if@single{#1}{\wide@bar@{#1}{#2}{1}}{\wide@bar@{#1}{#2}{2}}}
\newcommand*\wide@bar@[3]{%
  \begingroup
  \def\mathaccent##1##2{%
    \let\mathaccent\save@mathaccent
    \if#32 \let\macc@nucleus\first@char \fi
    \setbox\z@\hbox{$\macc@style{\macc@nucleus}_{}$}%
    \setbox\tw@\hbox{$\macc@style{\macc@nucleus}{}_{}$}%
    \dimen@\wd\tw@
    \advance\dimen@-\wd\z@
    \divide\dimen@ 3
    \@tempdima\wd\tw@
    \advance\@tempdima-\scriptspace
    \divide\@tempdima 10
    \advance\dimen@-\@tempdima
    \ifdim\dimen@>\z@ \dimen@0pt\fi
    \rel@kern{0.6}\kern-\dimen@
    \if#31
      \overline{\rel@kern{-0.6}\kern\dimen@\macc@nucleus\rel@kern{0.4}\kern\dimen@}%
      \advance\dimen@0.4\dimexpr\macc@kerna
      \let\final@kern#2%
      \ifdim\dimen@<\z@ \let\final@kern1\fi
      \if\final@kern1 \kern-\dimen@\fi
    \else
      \overline{\rel@kern{-0.6}\kern\dimen@#1}%
    \fi
  }%
  \macc@depth\@ne
  \let\math@bgroup\@empty \let\math@egroup\macc@set@skewchar
  \mathsurround\z@ \frozen@everymath{\mathgroup\macc@group\relax}%
  \macc@set@skewchar\relax
  \let\mathaccentV\macc@nested@a
  \if#31
    \macc@nested@a\relax111{#1}%
  \else
    \def\gobble@till@marker##1\endmarker{}%
    \futurelet\first@char\gobble@till@marker#1\endmarker
    \ifcat\noexpand\first@char A\else
      \def\first@char{}%
    \fi
    \macc@nested@a\relax111{\first@char}%
  \fi
  \endgroup
}
\makeatother

\makeatletter
\xpatchcmd\thmt@restatable{%
\csname #2\@xa\endcsname\ifx\@nx#1\@nx\else[{#1}]\fi
}{%
\ifthmt@thisistheone
\csname #2\@xa\endcsname\ifx\@nx#1\@nx\else[{#1}]\fi
\else
\csname #2\@xa\endcsname[{Restated}]
\fi}{}{}
\makeatother

\renewcommand{\bar}{\widebar}

\newcommand{\Span}{\textup{span}}

\newcommand{\ind}{\mathbbm{1}}

\newcommand{\labs}[1]{\left\lvert #1 \right\rvert}
\newcommand{\lpr}[1]{\left( #1 \right)}
\newcommand{\lcr}[1]{\left\{ #1 \right\}}

\newcommand{\itX}{\textit{X}}
\newcommand{\itY}{\textit{Y}}
\newcommand{\itZ}{\textit{Z}}
\newcommand{\itW}{\textit{W}}

\newcommand{\hashx}{h_{\textit{X}}}
\newcommand{\hashy}{h_{\textit{Y}}}
\newcommand{\hashz}{h_{\textit{Z}}}

\newcommand{\alphx}[1][\alpha]{{#1}_{\textit{X}}}
\newcommand{\alphy}[1][\alpha]{{#1}_{\textit{Y}}}
\newcommand{\alphz}[1][\alpha]{{#1}_{\textit{Z}}}
\let\alphax\alphx
\let\alphay\alphy
\let\alphaz\alphz

\newcommand{\lvl}{\ensuremath{\l}}  %

\newcommand{\numxblock}{N_{\textup{BX}}}
\newcommand{\numyblock}{N_{\textup{BY}}}
\newcommand{\numzblock}{N_{\textup{BZ}}}

\newcommand{\numtriple}{N_{\alphax, \alphay, \alphaz}}

\newcommand{\numalpha}[1][\alpha]{N_{#1}}

\newcommand{\pcomp}{p_{\textup{comp}}}
\newcommand{\pcompY}{p_{\itY, \textup{comp}}}
\newcommand{\pcompZ}{p_{\itZ, \textup{comp}}}

\newcommand{\splres}{\beta}  %

\newcommand{\splresX}{\splres_{\textit{X}}}
\newcommand{\splresY}{\splres_{\textit{Y}}}
\newcommand{\splresZ}{\splres_{\textit{Z}}}

\newcommand{\splresavg}{\bar{\beta}}

\newcommand{\splresXt}[1][t]{\splres_{\textit{X}, #1}}
\newcommand{\splresYt}[1][t]{\splres_{\textit{Y}, #1}}
\newcommand{\splresZt}[1][t]{\splres_{\textit{Z}, #1}}

\newcommand{\splonelevelXt}[1][t]{\gamma_{\textit{X}, #1}}
\newcommand{\splonelevelYt}[1][t]{\gamma_{\textit{Y}, #1}}
\newcommand{\splonelevelZt}[1][t]{\gamma_{\textit{Z}, #1}}

\newcommand{\xiYt}[1][t]{\xi_{\textit{Y}, #1}}
\newcommand{\xiZt}[1][t]{\xi_{\textit{Z}, #1}}

\newcommand{\oeps}{o_{1/\varepsilon}}

\newcommand{\ang}[1]{\langle #1 \rangle}

\newcommand{\splavg}[1]{\splresavg_{#1, *, *, *}}

\newcommand{\+}{\textup{+}}
\newcommand{\<}{\textup{<}}

\newcommand{\THash}{\T_{\textup{hash}}}
\newcommand{\TYComp}{\T_{\textup{comp}}^Y}
\newcommand{\TYUseful}{\T_{\textup{useful}}^Y}
\newcommand{\TZComp}{\T_{\textup{comp}}^Z}
\newcommand{\TZUseful}{\T_{\textup{useful}}^Z}

\newcommand{\mmid}{\;\middle|\;}

\newcommand{\say}[1]{``#1''}

\newcommand{\EquationOnSameLine}[1]{\hspace*{\fill}$\displaystyle #1$\hspace*{\fill}\mbox{}}

\begin{document}

\newcommand{\authorspace}{1cm}

\author{
  \hspace{\authorspace}
  Josh Alman\thanks{Columbia University. \texttt{josh@cs.columbia.edu}. Supported in part by NSF Grant CCF-2238221 and a grant from the Simons Foundation (Grant Number 825870 JA).}
  \and
  Ran Duan\thanks{Institute for Interdisciplinary Information Sciences, Tsinghua University. \texttt{duanran@mail.tsinghua.edu.cn}.}
  \and
  Virginia Vassilevska Williams\thanks{Massachusetts Institute of Technology. \texttt{virgi@mit.edu}. Supported by NSF Grant CCF-2330048, BSF Grant 2020356 and a Simons Investigator Award. This work was done in part while the author was visiting the Simons Institute for the Theory of Computing.}
  \hspace{\authorspace}
  \and
  \hspace{\authorspace}
  Yinzhan Xu\thanks{Massachusetts Institute of Technology. \texttt{xyzhan@mit.edu}. Partially supported by NSF Grant CCF-2330048, BSF Grant 2020356, a Simons Investigator Award and HDR TRIPODS Phase II grant 2217058. This work was done in part while the author was visiting the Simons Institute for the Theory of Computing.}
  \and
  Zixuan Xu\thanks{Massachusetts Institute of Technology. \texttt{zixuanxu@mit.edu}.}
  \and
  Renfei Zhou\thanks{Carnegie Mellon University. \texttt{renfeiz@andrew.cmu.edu}.}
  \hspace{\authorspace}
}
\title{More Asymmetry Yields Faster Matrix Multiplication}
\date{}
\pagenumbering{gobble} 
\maketitle

\begin{abstract}
We present a new improvement on the laser method for designing fast matrix multiplication algorithms. The new method further develops the recent advances by [Duan, Wu, Zhou FOCS 2023] and [Vassilevska Williams, Xu, Xu, Zhou SODA 2024]. Surprisingly the new improvement is achieved by incorporating more asymmetry in the analysis, circumventing a fundamental tool of prior work that requires two of the three dimensions to be treated identically. Together with recent advances in numerical optimization [Dupont et al.\ arXiv:2608.16884], the method yields a new bound on the square matrix multiplication exponent
$$\omega<2.371177,$$
improved from the previous bound of $\omega<2.371552$. We also improve the bounds of the exponents for multiplying rectangular matrices of various shapes.
\end{abstract}

\newpage
\pagenumbering{arabic}

\section{Introduction}
Multiplication of matrices is a fundamental algebraic primitive with applications throughout computer science and beyond. The study of its algorithmic complexity has been a vibrant area in theoretical computer science and mathematics ever since Strassen's \cite{strassen} 1969 discovery that the {\em rank} of 2 by 2 matrix multiplication is $7$ (and not $8$), leading to the first truly subcubic, $O(n^{2.81})$-time algorithm for multiplying $n\times n$ matrices.
Fifty-five years later, researchers are still attempting to lower the exponent $\omega$, defined as the smallest real number for which $n\times n$ matrices can be multiplied in $O(n^{\omega+\eps})$ time for all $\eps>0$.

After many decades of work (e.g.~\cite{strassen,Pan78,BCRL79,Schonhage81,Romani82,cw81as,laser,cw90,stothers,virgi12,LeGall32power,AlmanW21,duan2023,VXXZ24}), 
the current best bound $\omega<2.371552$ was given by \cite{VXXZ24} optimizing a recent approach by Duan, Wu, and Zhou \cite{duan2023}.

There is a straightforward lower bound of $\omega\geq 2$ since the output size is $n^2$. 
No larger lower bound is known, leading many to optimistically conjecture that $\omega=2$.
Unfortunately, several papers \cite{ambainis,blasiak2017groups,blasiak2017cap,almanitcs,Alman21,aw2,ChristandlVZ21} prove significant limitations to all known approaches for designing matrix multiplication algorithms. The most general limitations \cite{Alman21,aw2,ChristandlVZ21} say that even major generalizations of the known approaches cannot prove $\omega=2$. The most restricted limitation~\cite{ambainis} focuses on the {\em laser method} defined by Strassen \cite{laser} and applied to the powers of a very particular tensor $\CW_5$ defined by Coppersmith and Winograd \cite{cw90}. This is the method that {\em all} the best results\footnote{Coppersmith and Winograd's paper \cite{cw90} technically used $\CW_6$ instead of $\CW_5$, but the limitation of \cite{ambainis} on $\CW_6$ is actually worse than those for $\CW_5$.} from the last 38 years have used. The limitation says that the laser method on $\CW_5$ cannot prove that $\omega<2.3078$. All papers on matrix multiplication algorithms from the last 10 years or so have been focused on bringing the $\omega$ upper bound closer to this $2.3078$ lower bound.

The main contribution of this paper is a new improvement over the laser method when applied to the powers of $\CW_5$. The new method builds upon work of \cite{duan2023} and \cite{VXXZ24}  and yields new improved bounds on $\omega$ and several rectangular matrix multiplication exponents\footnote{
  The constraint programs that our new method leads to are significantly larger and more complex than in prior work. The nonlinear solver we are using struggles, and it can take many days for it to get a solution for any fixed $\omega(1,k,1)$. Unfortunately, %
  we were not able to solve the constraint program for the value $\alpha$ studied by Coppersmith~\cite{Coppersmith82,coppersmith1997rectangular} defined as the largest number such that $n$ by $n^\alpha$ by $n$ matrix multiplication can be done in $O(n^{2+\eps})$ time for all $\eps>0$. %
} $\omega(1,k,1)$ defined as the smallest real value for which $n\times n^k$ by $n^k\times n$ matrix multiplication can be done in $O(n^{\omega(1,k,1)+\eps})$ time for all $\eps>0$; the previous best bounds for rectangular matrix multiplication were by \cite{LeGall24,VXXZ24}.

Our results are summarized in \cref{table:result}.
Our new bound on $\omega$ is 
\[\omega<2.371339,\]
improved from the previous bound by \cite{VXXZ24} of $\omega<2.371552$, inching towards the lower bound of $2.3078$.
A recent work \cite{alphaevolve-note} further improves the bound derived from our method to
\[\omega<2.371177.\]
As in prior works, our method reduces the search for the best bound on $\omega$ to a nonlinear optimization problem (see \cref{sec:numerical}); \cite{alphaevolve-note} solve this optimization problem at a larger scale using modern optimization and machine learning techniques, leading to the new bound based on the same algorithm and analysis as in this paper.

As a specific example for a rectangular matrix multiplication exponent, we obtain a new bound for the exponent $\mu$ satisfying the equation $\omega(1,\mu,1)=1+2\mu$. The previous bound from \cite{VXXZ24} was $\mu<0.527661$ and we improve it to $\mu<0.5275$. The value $\mu$ is a key part of the best known running times of several important problems, including All-Pairs Shortest Paths (APSP) in unweighted directed graphs~\cite{zwickbridge}, computing minimum witnesses of Boolean Matrix Multiplication \cite{CzumajKL07}, and All-Pairs Bottleneck Paths in node-weighted graphs~\cite{ShapiraYZ11}. Our new bound implies %
that all the aforementioned problems can be solved in $O(n^{2.5275})$ time, improving on the previous known running time of $O(n^{2.527661})$.

\newcommand{\colwidth}{2.5cm}
\newcolumntype{M}[1]{>{\centering\arraybackslash}m{#1}}

\begin{table}[ht]
  \caption{Our bounds on $\omega(1, k, 1)$ compared to the previous bounds from \cite{VXXZ24}. The stated bound for $k = 1$ ($\dagger$) is obtained in \cite{alphaevolve-note} by finding parameters for our analysis on the eighth power of the CW tensor (the parameters found in this paper for the fourth power give $\omega < 2.371339$ instead). All other stated bounds in the table are obtained by analyzing the fourth power.}\label{table:result}
  \centering
  \begin{tabular}{|c|c|c|}
    \hline
    $k$ & upper bound on $\omega(1, k, 1)$ & previous bound on $\omega(1, k, 1)$ \\
    \hline
    0.33 & 2.000092 & 2.000100 \\
    0.34 & 2.000520 & 2.000600 \\
    0.35 & 2.001243 & 2.001363 \\
    0.40 & 2.009280 & 2.009541 \\
    0.50 & 2.042776 & 2.042994 \\
    \textbf{0.527500} & 2.054999 & N/A \\
    0.60 & 2.092351 & 2.092631 \\
    0.70 & 2.152770 & 2.153048 \\
    0.80 & 2.220639 & 2.220929 \\
    0.90 & 2.293941 & 2.294209 \\
    1.00 & \textbf{2.371177}\textsuperscript{$\dagger$} & 2.371552 \\
    1.50 & 2.794633 & 2.794941 \\
    2.00 & 3.250035 & 3.250385 \\
    \hline
  \end{tabular}
\end{table}

\section{Technical Overview}

In this section, we give a high-level overview of recent algorithms for matrix multiplication and the new ideas we introduce in this paper for further improvement. We assume familiarity with notions related to tensors and matrix multiplication; the unfamiliar reader may want to read the preliminaries in Section~\ref{sec:prelim} below first. Afterward, in Section~\ref{sec:outline}, we give a more detailed overview of how we design our new algorithm.

\subsection{The Laser Method and Asymptotic Sum Inequality}

In order to design a matrix multiplication algorithm, using Sch{\"{o}}nhage's asymptotic sum inequality~\cite{Schonhage81}, it suffices to give an upper bound on the asymptotic rank of a direct sum of matrix multiplication tensors. However, directly bounding the rank of a matrix multiplication tensor is quite difficult (for example, even determining the rank of the tensor for multiplying $3 \times 3$ matrices is still an open problem today), and so algorithms since the work of Coppersmith and Winograd~\cite{cw90} have used an indirect approach:
\begin{enumerate}
    \item Start with the Coppersmith-Winograd tensor $\CW_q$, which is known to have minimal asymptotic rank.
    \item Take a large Kronecker power $\CW_q^{\otimes n}$, which must therefore also have a small (asymptotic) rank.
    \item \label{step3} Use the laser method to show that $\CW_q^{\otimes n}$ has a \emph{zeroing out}\footnote{If $T$ is a tensor over $X,Y,Z$, a zeroing out of $T$ is a tensor obtained by picking subsets $X' \subseteq X$, $Y' \subseteq Y$, $Z' \subseteq Z$, and restricting $T$ to only those subsets.} to a large direct sum $S$ of matrix multiplication tensors.
    \item Since zeroing out cannot increase the tensor rank, this shows that $S$ has a small (asymptotic) rank, and then we can apply Sch{\"{o}}nhage's asymptotic sum inequality.
\end{enumerate}
The laser method in Step \ref{step3} is a tool introduced by Strassen~\cite{strassenlaser1} for converting Kronecker powers of tensors into direct sums of matrix multiplication tensors. All improvements since Coppersmith and Winograd's algorithm have used this approach~\cite{stothers,virgi12,LeGall32power,AlmanW21,duan2023,VXXZ24}, focusing on improving the use of the laser method to yield a \emph{larger} direct sum $S$ of \emph{larger} matrix multiplication tensors, and this is the step which we improve here as well.
The laser method is a general tool which applies to any tensor $T$ that has been partitioned in a particular way. Let us first introduce some relevant notations.

Suppose $T$ is a tensor over $X,Y,Z$, and we partition $X = \bigsqcup_i X_i$, $Y = \bigsqcup_j Y_j$, and $Z = \bigsqcup_k Z_k$ and define the subtensor $T_{i,j,k}$ as $T$ restricted to $X_i, Y_j, Z_k$. Thus, $T$ is partitioned as $T = \sum_{i,j,k} T_{i,j,k}$. The laser method requires this partition to have additional structure\footnote{Namely, there must be an integer $p$ such that $T_{i,j,k} = 0$ whenever $i+j+k \neq p$. In the case of $\CW_q$, we will use a partition with $p=2$.} which we will not focus on in this overview. Notably, when we take a Kronecker power $T^{\otimes n}$, which is a tensor over $X^n, Y^n, Z^n$, this can also be partitioned as $T^{\otimes n} = \sum_{I,J,K} T_{IJK}$ where $I,J,K$ are vectors of length $n$, and $T_{IJK} \defeq \bigotimes_{\ell=1}^n T_{I_\ell, J_\ell, K_\ell}$ is a tensor over $X_I := \prod_{\ell=1}^n X_{I_\ell}$, $Y_J := \prod_{\ell=1}^n Y_{J_\ell}$, and $Z_K := \prod_{\ell=1}^n Z_{K_\ell}$.

Consider a probability distribution $\alpha$ on the subtensors of $T$, which assigns probability $\alpha_{ijk}$ to subtensor $T_{i,j,k}$. In the Kronecker power $T^{\otimes n}$, we can zero out according to the \emph{marginals} of $\alpha$. For instance, we zero out $X_I \subset X^n$ unless, for all $i$, we have $\frac{1}{n} \cdot |\{ \ell \in [n] \mid I_\ell = i\}| = \sum_{j,k} \alpha_{ijk}$. Assuming $\alpha$ is uniquely determined by its marginals (which is often not the case, as we will discuss more later), this means $T^{\otimes n}$ has been zeroed out so that every remaining $T_{IJK}$ is a copy (up to permutation of indices) of the tensor $B = \bigotimes_{i,j,k} T_{i,j,k}^{\otimes \alpha_{ijk} \cdot n}$.

Summarizing, we have so far zeroed out $T^{\otimes n}$ into a sum of many copies of the tensor $B$ (one copy of $B$ from each $T_{IJK}$). However, this is not a direct sum. Indeed, for example, there can be many remaining $T_{IJK}$ and $T_{I'J'K'}$ with $I = I'$, which both use $X_I$. 

The main result of the laser method is that a further, carefully-chosen zeroing out can result in a \emph{direct} sum of copies of the tensor $B$, where each two remaining $T_{IJK}$ and $T_{I'J'K'}$ have $I \neq I'$, $J \neq J'$, and $K \neq K'$. Moreover, this further zeroing out only removes a small number of $X_I, Y_J, Z_K$, so that a large number of copies of $B$ remain. (The exact number is a combinatorial expression in terms of the marginals of $\alpha$ which we will discuss in more detail in later sections.)

When this method is applied to $\CW_q$, each subtensor is in fact a matrix multiplication tensor, so $B$ is also a matrix multiplication tensor. Thus, the output of the laser method can be directly given to the asymptotic sum inequality.

\subsection{Recursive Applications of the Laser Method to Tensor Powers of \texorpdfstring{\boldmath $\CW_q$}{CWq}}

Coppersmith and Winograd \cite{cw90} then noticed that the laser method can yield improved results when applied to the Kronecker square, $\CW_q^{\otimes 2}$. 

Roughly speaking, the intuition is that the square gives us more freedom to pick the partition that is needed to apply the laser method, since Kronecker squares of the partitions of $\CW_q$ correspond to only a subset of the possible partitions of $\CW_q^{\otimes 2}$.  
Coppersmith and Winograd take advantage of this, and pick a partition of $\CW_q^{\otimes 2}$ which allows the tensors $T_{IJK}$ to become larger. 
In fact, compared to the partition obtained from analyzing  $\CW_q^{\otimes 2n}$, the partition from analyzing  $(\CW_q^{\otimes 2})^{\otimes n}$ is a \emph{coarsening} of the previous partition (where, for instance, each $X_I$ now is a union of possibly multiple previous $X_I$'s). Thus, $T_{IJK}$, which is a tensor over the variable blocks $X_I, Y_J, Z_K$, is now larger. 
If the tensors $T_{IJK}$'s were all matrix multiplication tensors, then the asymptotic sum inequality tends to give better bounds when they are larger. 
However, in this new partition, each subtensor is no longer a matrix multiplication tensor. In other words, the laser method still yields a direct sum of copies of $B = \bigotimes_{i,j,k} T_{i,j,k}^{\otimes \alpha_{ijk} \cdot n}$, but $B$ is not a matrix multiplication tensor, so the asymptotic sum inequality cannot be directly applied. 

Coppersmith and Winograd solved this by applying the laser method \emph{recursively}. They first applied the laser method to each subtensor $T_{i,j,k}$ which is not a matrix multiplication tensor, showing that $T_{i,j,k}^{\otimes m}$ for large $m$ can zero out into direct sums of matrix multiplication tensors. Then, since the tensor $B$ is a Kronecker product of such powers, this can be used to zero out into a direct sum of matrix multiplication tensors.

The next few improvements to $\omega$ \cite{stothers,virgi12,LeGall32power} applied this recursive laser method approach to $\CW_q^{\otimes 2^{\lvl-1}}$ for larger and larger integers $\lvl \ge 1$. Similar to before, there is a natural partition $X = \bigsqcup_i X_i$, $Y = \bigsqcup_j Y_j$, and $Z = \bigsqcup_k Z_k$ of the variables from the base tensor $\CW_q^{\otimes 2^{\lvl-1}}$, and products of $X_i$'s, $Y_j$'s or $Z_k$'s form subsets of variables $X_I$'s, $Y_J$'s, or $Z_K$'s in $(\CW_q^{\otimes 2^{\lvl-1}})^{\otimes n}$. We will call each subset of variables $X_I, Y_J, Z_K$ a \emph{level-$\lvl$ variable block} and we call a subtensor obtained from restricting to some level-$\ell$ variable blocks a \emph{level-$\ell$ subtensor}. Since the level-$\ell$ partition is a coarsening of the level-$(\ell-1)$ partition, each level-$\lvl$ variable block is the union of possibly many level-$(\lvl-1)$ variable blocks induced by the natural partition of $\CW_q^{\otimes 2^{\lvl-2}}$, and consequently each level-$\ell$ subtensor is the sum of many level-$(\ell-1)$ subtensors.
As $\lvl$ increases, the task of applying the laser method to increasingly more intricate subtensors becomes more difficult, and significant algorithmic and computational challenges arise. 

The next improvement~\cite{AlmanW21} focused on the case when the marginals of $\alpha$ do not uniquely determine $\alpha$. In this case, there is a corresponding penalty that reduces the number of copies of $B$ the laser method can achieve, and~\cite{AlmanW21} showed a new approach to reduce this penalty.

\subsection{Asymmetry and Combination Loss}

All improvements to the laser method thus far would apply equally well to any tensor with an appropriate partition. 
The next, recently improved algorithm~\cite{duan2023} focused on improving the recursive use of the laser method specifically when it is applied to the Coppersmith-Winograd tensor $\CW_q$ in the recursive fashion described above.

The key observation, referred to as ``combination loss'' in their paper, is as follows. Above, in order to ensure that the laser method gave a direct sum of copies of $B$, previous works ensured that any two remaining level-$\ell$ subtensors do not share any level-$\ell$ variable blocks $X_I$, $Y_J$ or $Z_K$. For $\ell = 1$, namely when the laser method is applied to the first power of $\CW_q$, this is both necessary and sufficient, because the level-$1$ subtensors are matrix multiplication tensors which use the entire level-1 variable blocks, so any two sharing an $I$, $J$, or $K$ must intersect.

However, for $\ell > 1$, the laser method needs to be applied recursively to every level-$k$ subtensor for $k = \ell,\,\ell\!-\!1,\dots, 1$. This means that by the end of the algorithm, we are potentially only using a small fraction of the variables in the level-$\ell$ subtensor rather than the entire subtensor. So it is possible that even if two level-$\ell$ subtensors $T_{IJK}$ and $T_{I'J'K}$ share the same level-$\ell$ $Z$-variable block $Z_K$ in the beginning, after recursive applications of the laser method, the subtensors of $T_{IJK}$ and $T_{I'J'K}$ which remain after zeroing out into matrix multiplication tensors may use disjoint subsets of $Z_K$, i.e., they may be independent. In this case, we could keep both $T_{IJK}$ and $T_{I'J'K}$ when we apply the laser method to $\CW_q^{\otimes 2^{\ell-1}}$ and still have a direct sum of subtensors, contrary to what was done in the previous algorithms where only one could be kept.

In order to benefit from this observation, one major issue is that when $T_{IJK}$ and $T_{I'J'K}$ are each zeroed out to their desired subtensors via the recursive applications of the laser method, they may each zero out parts of $Z_K$ that are meant to be used by the other. Thus it is unclear how to control and analyze the recursive procedure so that we maximize the number of these subtensors remaining in the end while ensuring that they are indeed independent.

However, for the $\CW_q$ tensor, Duan, Wu, and Zhou in \cite{duan2023} observed that the issue can be partially overcome. In particular, they were able to allow two level-$\ell$ subtensors to share level-$\ell$ variable blocks in the $Z$-dimension as long as they had the guarantee that any two level-$\ell$ subtensors do not share any level-$\ell$ blocks in the $X$- or $Y$-dimension. Due to the special structure of the $\CW_q$ tensor, it turns out that by ensuring independence among level-$\ell$ $X$- and $Y$-blocks, we can obtain sufficient information about the level-$(\ell-1)$ structure of $Z$-blocks in level-$\ell$ subtensors.
Therefore, by introducing asymmetry via treating the $Z$-variables differently from the $X$- and $Y$-variables, \cite{duan2023} obtained a procedure that obtains a collection of level-$\ell$ subtensors that do not share any level-$(\ell-1)$ $Z$-variable blocks, so they are independent. The follow-up work~\cite{VXXZ24} further generalized the techniques presented in \cite{duan2023} and allowed the sharing of level-$(\ell-1)$ $Z$-variable blocks. In fact, their approach only requires that each level-$1$ $Z$-variable block belongs to a unique remaining level-$\ell$ subtensor $T_{IJK}$. This allowed them to obtain an improvement since each subtensor $T_{IJK}$ can use a more ``fine-grained'' subset of $Z$-variables.

A priori, it is perhaps surprising that such an asymmetric approach, applied to the very symmetric Coppersmith-Winograd tensor, leads to an improved algorithm. In previous work, much effort was devoted to ensuring that the algorithm was entirely symmetric with respect to the $X,Y,Z$-dimensions. For instance, many intermediate steps of the laser method require taking the \emph{minimum} of three quantities depending separately on $X,Y,Z$ (like the entropies of the three marginal distributions of $\alpha$), and picking a symmetric distribution prevents losses in these steps. Similarly, in the design of algorithms for \emph{rectangular} matrix multiplication using the laser method~\cite{legallrect,legallrect2,VXXZ24}, one seems to truly suffer a penalty due to the necessary asymmetry. (For a rough comparison, symmetrizing the best rectangular matrix multiplication algorithms leads to a substantially worse square matrix multiplication algorithm.) However, the gains from the approach above ultimately outweigh the losses from asymmetry.

\subsection{Our Contributions: More Asymmetry}

Following from the sequence of improvements obtained by \cite{duan2023,VXXZ24}, the natural end goal for this line of work is to allow different remaining subtensors to share level-$\lvl$ variable blocks in all three dimensions, and only require level-$1$ variable blocks to belong to unique remaining level-$\ell$ subtensors. However, it is unclear how one may achieve this goal since the algorithms of both  \cite{duan2023} and \cite{VXXZ24} crucially use the fact that different remaining level-$\ell$ subtensors use unique level-$\lvl$ $X$- and $Y$-variable blocks. In this work, we make substantial progress toward this goal by allowing different remaining subtensors to share level-$\lvl$ variable blocks in both the $Y$- and $Z$-dimensions.

In this work, we give a new approach in which we \emph{only} require that any two level-$\ell$ subtensors do not share level-$\ell$ $X$-blocks. In particular, two level-$\ell$ subtensors can share the same level-$\ell$ $Y$-variable block or $Z$-variable block, but cannot share the same level-$1$ $Y$-variable block or $Z$-variable block, strengthening ideas in \cite{VXXZ24}. In this way, we are able to keep a larger number of level-$\ell$ subtensors and thus zero out the input $\CW_q^{\otimes 2^{\ell-1}}$ into more matrix multiplication tensors. This is achieved by introducing more asymmetry into our procedure: we now give a different zeroing-out procedure for each of the $X,Y,Z$-dimensions. This is perhaps unexpected, because in the end goal all $X$-, $Y$-, $Z$-dimensions have the same requirement that every level-$1$ variable block belongs to unique remaining level-$\ell$ subtensors. One might expect that a more symmetric implementation of the ideas of prior work would help towards the goal, but we instead use more asymmetry to get around a major limitation of the prior work.

Compared to the prior works, which use constrained information from both the $X$- and $Y$-dimensions to control the $Z$-dimension, we first use just the $X$-dimension to partially control the $Y$-dimension, and then use both the fully-controlled $X$-dimension and partially-controlled $Y$-dimension to control the $Z$-dimension. More specifically, by utilizing more structure specific to the  $\CW_q$ tensor, combined with the requirement that each level-$\ell$ subtensor $T_{IJK}$ uses a unique level-$\ell$ $X$-block, we find that this already constrains $T_{IJK}$ to use a smaller and more structured subset of level-$1$ $Y$-blocks, which is sufficient to perform zero-outs over level-$1$ $Y$-blocks so that each of them is contained in a unique level-$\ell$ subtensor. After obtaining a set of level-$\ell$ subtensors that do not share any level-$\ell$ $X$-blocks or level-$1$ $Y$-blocks, we then have enough information to do zero-outs over level-$1$ $Z$-blocks so that each of them are contained in a unique level-$\ell$ subtensor. See \cref{sec:outline} for a more detailed overview of our algorithm.

Similar to recent prior matrix multiplication algorithms, we need to solve a large, non-convex optimization problem in order to determine parameters for the laser method to achieve the best exponent. 
That said, we can already see that our new approach could potentially lead to an improvement without solving the optimization program to get the final bound for $\omega$. A key parameter that reflects the effectiveness of the laser method is the hash modulus being used in the hashing step standard to all applications of the laser method, which determines the number of remaining level-$\ell$ subtensors we are able to keep in expectation. Minimizing the value of this hash modulus can give an improvement to the bound we are able to obtain for $\omega$. The new hash modulus we use in this work is taken to be (ignoring low-order terms)
\[
\max\left\{\frac{\numtriple}{\numxblock},  
\frac{\numalpha \cdot \pcompY}{\numyblock}, 
\frac{\numalpha \cdot \pcompZ}{\numzblock}\right\},
\]
in comparison with the following value (similarly without low-order terms) used in \cite{VXXZ24}
\[
\max\left\{\frac{\numtriple}{\numxblock},  
\frac{\numtriple}{\numyblock}, 
\frac{\numalpha \cdot \pcompZ}{\numzblock}\right\}. 
\]
The definitions of all the variables in these expressions are not important for the purpose of this overview, and we focus instead on their structure.  
Two out of the three terms in the two bounds are the same, but we have improved the middle term from $\frac{\numtriple}{\numyblock}$ to $\frac{\numalpha \cdot \pcompY}{\numyblock}$, where $\pcompY$ is a variable between $0$ and $1$ which we can pick in our optimization problem. In fact, even in the worst case $\pcompY = 1$, we improve the second term by a factor of $\numtriple/\numalpha$, which is guaranteed to be at least $1$ and is exactly the penalty term studied by \cite{AlmanW21}.  Thus, even if we pick $\pcompY = 1$, we would at least recover the modulus of prior work. That said, since the modulus is defined by taking the maximum over the three values, it is not immediately clear that our approach uses a strictly smaller modulus. 
However, by solving the new constraint program,
one can verify that even in $\CW_q^{\otimes 2}$, the optimal parameters use a smaller value of $\pcompY$, which leads to a better bound on $\omega$. Indeed, our new approach gives a better analysis of $\CW_q^{\otimes 2}$ than prior work, improving from the bound $\omega < 2.374399$ achieved by~\cite{duan2023}\footnote{The subsequent work~\cite{VXXZ24} improves the analysis of higher powers of $\CW_q$ but not the square.} to  $\omega < 2.37432$.

\subsection{Limitations to Our Techniques}

In our approach, we successfully allow all obtained level-$\ell$ subtensors to share level-$k$ variable blocks for any $1<k\leq \ell$ in both the $Y$- and $Z$-dimensions, as long as they are independent with respect to the level-$1$ blocks. However, we still require that none of them share  level-$\ell$ $X$-blocks. In an ideal approach, to truly maximize the number of independent tensors, the subtensors should also be allowed to share level-$\ell$ $X$-blocks as long as they do not share level-$1$ $X$-blocks. 
However, there does not seem to be a way to continue our approach to also allow the $X$ blocks to share level-$\ell$-blocks. Our approach crucially uses the fact that a level-$\ell$ $X$-block uniquely determines the level-$\ell$ subtensor that it belongs in, to control the structure of level-$1$ blocks in level-$\ell$ subtensors. It seems that a truly new idea is needed to be able to control this structure without making the level-$\ell$ subtensors to have unique level-$\ell$ variable blocks in one of the dimensions. In other words, we believe that we have reached the limit of generalizing the techniques introduced in \cite{duan2023}.

Second, as in all recent improvements, determining the final running time of our matrix multiplication algorithm requires solving a non-convex optimization problem to determine the best choice of the probability distribution $\alpha$ to use in the laser method. Since we can no longer use a symmetry between the $X$- and $Y$-dimensions, our new optimization problem has more variables and appears to be considerably harder to solve to high precision than the problems in prior work. It would be particularly exciting to find an approach to improving $\omega$ which does not need powerful optimization software and computing systems.

\section{Preliminaries}
\label{sec:prelim}

We use by now standard notation as in \cite{AlmanW21,duan2023,VXXZ24}. The experienced reader may skip this section, as the definitions and notation are the same as in prior work.

\subsection{Tensors and Tensor Operations}

\paragraph{Tensors.} 
We deal with tensors of order $3$ which can be viewed as either $3$-dimensional arrays or as trilinear polynomials. We consider the latter representation:

A tensor $T$ over variable sets $X = \midBK{x_1, \ldots, x_{|X|}}$, $Y = \midBK{y_1, \ldots, y_{|Y|}}$, $Z = \midBK{z_1, \ldots, z_{|Z|}}$, and field $\F$ is a trilinear form
\[
  T = \sum_{i=1}^{|X|} \sum_{j=1}^{|Y|} \sum_{k=1}^{|Z|} a_{i,j,k} \cdot x_i y_j z_k,
\]
where all $a_{i,j,k}$ are from $\F$.

All tensors in this paper have all their entries $a_{i,j,k} \in \midBK{0, 1}$.
As any field $\F$ has a $0$ and a $1$, all tensors in the paper can be considered to be over any $\F$.

In the following, assume $T$ is a tensor over $X = \midBK{x_1, \ldots, x_{|X|}}$, $Y = \midBK{y_1, \ldots, y_{|Y|}}$, $Z = \midBK{z_1, \ldots, z_{|Z|}}$ and $T'$ is a tensor over $X' = \midBK{x'_1, \ldots, x'_{|X'|}}$, $Y' = \midBK{y'_1, \ldots, y'_{|Y'|}}$, $Z' = \midBK{z'_1, \ldots, z'_{|Z'|}}$, written as
\[
  T = \sum_{i=1}^{|X|} \sum_{j=1}^{|Y|} \sum_{k=1}^{|Z|} a_{i,j,k} \cdot x_i y_j z_k, \qquad
  T' = \sum_{i=1}^{|X'|} \sum_{j=1}^{|Y'|} \sum_{k=1}^{|Z'|} b_{i,j,k} \cdot x'_i y'_j z'_k,
\]
\paragraph{Tensor operations.} 
We define the sum, direct sum, and Kronecker product of two tensors $T$ and $T'$:
\begin{itemize}
\item The \emph{sum} $T + T'$ is only defined when both tensors are over the same variable sets $(X, Y, Z) = (X', Y', Z')$. $T+T'$ is the tensor defined by the sum of the two polynomials defining $T$ and $T'$:
  \[
    T + T' = \sum_{i=1}^{|X|} \sum_{j=1}^{|Y|} \sum_{k=1}^{|Z|} (a_{i,j,k} + b_{i,j,k}) \cdot x_i y_j z_k.
  \]
\item 
The \emph{direct sum} $T \oplus T'$ is the sum $T+T'$ when the variable sets of $T$ and $T'$ are disjoint. In this case, $T'$ and $T'$ are said to be \emph{independent}.

If $T$ and $T'$ share variables, we first relabel their variables so that the variable sets are disjoint and then $T \oplus T'$ is the sum of the two polynomials over new variables.

We write $T^{\oplus n} \defeq \underbrace{T \oplus T \oplus \cdots \oplus T}_{n~\textup{copies}}$ to denote the sum of $n$ independent copies\footnote{This is often written as $n\odot T$.} of $T$.

\item The \emph{Kronecker product} is defined as the tensor
  \[
    T \otimes T' = \sum_{i=1}^{|X|} \sum_{j=1}^{|Y|} \sum_{k=1}^{|Z|} \sum_{i'=1}^{|X'|} \sum_{j'=1}^{|Y'|} \sum_{k'=1}^{|Z'|} a_{i,j,k} \cdot b_{i', j', k'} \cdot (x_i, x'_{i'}) \cdot (y_j, y'_{j'}) \cdot (z_k, z'_{k'})
  \]
  over variable sets $X \times X'$, $Y \times Y'$, and $Z \times Z'$. We write $T^{\otimes n} \defeq \underbrace{T \otimes T \otimes \cdots \otimes T}_{n~\textup{times}}$ to denote the \emph{$n$-th tensor power} of $T$.
\item We say $T$ and $T'$ are \emph{isomorphic}\footnote{Note that in the literature, a more general type of isomorphism is considered, where $T$ and $T'$ can be transformed into one another under arbitrary linear transformations between $X$ and $X'$, $Y$ and $Y'$, and $Z$ and $Z'$, or even more generally under degenerations.}, denoted by $T \equiv T'$, if $|X| = |X'|$, $|Y| = |Y'|$, $|Z| = |Z'|$, and there are permutations $\pi_\textit{X}, \pi_\textit{Y}, \pi_\textit{Z}$ over $[|X|], [|Y|], [|Z|]$ respectively, such that $a_{i,j,k} = b_{\pi_\textit{X}(i), \pi_\textit{Y}(j), \pi_\textit{Z}(k)}$ for all $i, j, k$. In other words, both tensors are equivalent up to a relabeling of the variables. 
\end{itemize}

\subsection{Tensor Rank}
A tensor $T$ over $X, Y, Z$ has rank $1$ if there are some vectors $a,b,c$ of length $|X|,|Y|,|Z|$, respectively, so that 
\[
T=\Biggbk{\sum_{i = 1}^{|X|} a_{i}\cdot x_i}
\Biggbk{\sum_{j = 1}^{|Y|} b_{j}\cdot y_j}
\Biggbk{\sum_{k = 1}^{|Z|} c_{k}\cdot z_k}.
\]

For a tensor $T$ over $X, Y, Z$, the tensor rank $R(T)$ is defined to be the minimum integer $r\ge 0$ such that $T$ can be written as the sum of $r$ rank-$1$ tensors.
The corresponding sum is called the \emph{rank decomposition} of $T$.

Tensor rank satisfies subadditivity and submultiplicativity:

\begin{itemize}
    \item $R(T+T')\le R(T) + R(T')$.
    \item $R(T\oplus T')\le R(T) + R(T')$.
    \item $R(T\otimes T')\le R(T)\cdot R(T').$
\end{itemize}

The third item above implies that for all $m$, $R(T^{\otimes m})\leq R(T)^m$, and for many tensors (e.g., $2\times 2$ matrix multiplication) the inequality is strict. Due to Fekete's lemma, the following
 \emph{asymptotic rank} $\Tilde{R}(T)$ is well-defined for every tensor $T$: 
\[\Tilde{R}(T) \defeq \lim_{m\to \infty}\lpr{R(T^{\otimes m})}^{1/m}.\]
The asymptotic rank is upper bounded by $R(T^{\otimes m})^{1/m}$ for any fixed integer $m > 0$.

\subsection{Degenerations, Restrictions, Zero-outs}
Let $T$ and $T'$ be tensors over a field $\F$, $T$ is over $X, Y, Z$, and $T'$ is over $X', Y', Z'$. 

\paragraph{Degeneration.}
Let $\F[\lambda]$ be the ring of polynomials of the formal variable $\lambda$. We say that $T'$ is a degeneration of $T$, written as $T \unrhd T'$, if there exists $\F[\lambda]$-linear maps 
\begin{align*}
    \phi_\itX &: \Span_{\F[\lambda]}(X)\to \Span_{\F[\lambda]}(X'),\\
    \phi_\itY &: \Span_{\F[\lambda]}(Y)\to \Span_{\F[\lambda]}(Y'),\\
    \phi_\itZ &: \Span_{\F[\lambda]}(Z)\to \Span_{\F[\lambda]}(Z'),
\end{align*}
and $d\in \N$ such that  
\[T' = \lambda^{-d}\lpr{\sum_{i = 1}^{|X|}\sum_{j = 1}^{|Y|}\sum_{k = 1}^{|Z|}a_{i,j,k} \cdot \phi_\itX(x_i) \cdot \phi_\itY(y_j) \cdot \phi_\itZ(z_k)} + O(\lambda).\]

If $T' \unrhd T$, then $\Tilde{R}(T')\le \Tilde{R}(T)$.

\paragraph{Restriction.}
Restriction is a special type of degeneration for the case where the maps $\phi_\itX, \phi_\itY, \phi_\itZ$ are $\F$-linear maps. More specifically, $T'$ is a restriction of $T$ if there exist $\F$-linear maps
\begin{align*}
    \phi_\itX &: \Span_{\F}(X)\to \Span_{\F}(X'),\\
    \phi_\itY &: \Span_{\F}(Y)\to \Span_{\F}(Y'),\\
    \phi_\itZ &: \Span_{\F}(Z)\to \Span_{\F}(Z'),
\end{align*}
such that
\[T' = \sum_{i = 1}^{|X|}\sum_{j = 1}^{|Y|}\sum_{k = 1}^{|Z|}a_{i,j,k} \cdot \phi_\itX(x_i) \cdot \phi_\itY(y_j) \cdot \phi_\itZ(z_k).\]
If $T'$ is a restriction of $T$, $R(T')\le R(T)$, and consequently $\Tilde{R}(T')\le \Tilde{R}(T)$.

\paragraph{Zero-out.}

The laser method utilizes a limited type of restriction called zero-out, namely the maps $\phi_\itX, \phi_\itY, \phi_\itZ$ set some variables to zero. More specifically, for subsets $X'\subseteq X$, $Y'\subseteq Y$, $Z'\subseteq Z$ we define the maps as 
\[\phi_\itX(x_i) = \begin{cases}x_i & \text{If }x_i\in X',\\ 0 & \text{otherwise},\end{cases}\]
and similarly for $\phi_\itY, \phi_\itZ$. The resulting tensor 
\[T' = \sum_{i = 1}^{|X|}\sum_{j = 1}^{|Y|}\sum_{k = 1}^{|Z|}a_{i,j,k} \cdot \phi_\itX(x_i) \cdot \phi_\itY(y_j) \cdot \phi_\itZ(z_k) = \sum_{x_i\in X'}\sum_{y_j\in Y'}\sum_{z_k\in Z'}a_{i,j,k} \cdot x_iy_jz_k\]
is called a zero-out of $T$. Throughout this paper, we use the notation $T' = T\vert_{X',Y',Z'}$ to denote the tensor $T'$ obtained from $T$ by setting the variables in $X\setminus X'$, $Y\setminus Y'$, $Z\setminus Z'$ to zero. In this case, we also call $T'$ the \emph{subtensor} of $T$ over $X', Y', Z'$.

\subsection{Matrix Multiplication Tensors}

For positive integers $a,b,c$, the $a\times b \times c$ matrix multiplication tensor $\ang{a,b,c}$ is a tensor over the variable sets $\{x_{ij}\}_{i\in [a], j\in [b]}, \{y_{jk}\}_{j\in [b],k\in[c]}, \{z_{ki}\}_{i\in [a],k\in [c]}$ defined as the tensor computing the $a\times c$ (transpose of the) product matrix
$\{z_{ki}\}_{i\in [a], k\in [c]}$
of an $a\times b$ matrix $\{x_{ij}\}_{i\in [a], j\in [b]}$ and $b\times c$ matrix $\{y_{jk}\}_{j\in [b], k\in [c]}$. Specifically, 
\[\ang{a,b,c} = \sum_{i\in [a]}\sum_{j\in [b]}\sum_{k\in [c]} x_{ij} y_{jk} z_{ki}.\]
Notice that the coefficient in front of $z_{ki}$ is the $(i,k)$ entry of the product of $x$ with $y$.
It is not hard to check that $\ang{a,b,c}\otimes \ang{d,e,f} \equiv \ang{ad, be, cf}$.

Following the recursive approach introduced by Strassen in \cite{strassen}, for any integer $q \ge 2$, if $R(\ang{q,q,q})\le r$, then one can use the rank decomposition of $\ang{q,q,q}$ to design an arithmetic circuit of size $O(n^{\log_q(r)})$ to multiply two $n \times n$ matrices. This motivates the definition of the \emph{matrix multiplication exponent} $\omega$ as follows:
\[\omega \defeq \inf_{q \in \N, \, q \ge 2} \log_q (R(\ang{q,q,q})).\]
Namely, for every $\eps > 0$, there exists an arithmetic circuit of size $O(n^{\omega + \eps})$ that computes the multiplication of two $n\times n$ matrices. Since $\ang{q,q,q}^{\otimes n} \equiv \ang{q^n, q^n, q^n}$, equivalently $\omega$ can be written in terms of the asymptotic rank of $\ang{q,q,q}$ as 
\[\omega = \log_q (\Tilde{R}(\ang{q,q,q})).\]

We also consider the arithmetic complexity of multiplying rectangular matrices of sizes $n^a\times n^b$ and $n^b\times n^c$ where $a,b,c \in \R_{\ge 0}$. We define $\omega(a,b,c)$ as 
\[\omega(a,b,c) = \log_q\lpr{\Tilde{R}(\ang{q^a,q^b,q^c})}\]
where $q \ge 2$ is a positive integer. This means that for any $\eps > 0$, there exists an arithmetic circuit of size $O(n^{\omega(a,b,c) + \eps})$ that computes the multiplication of an $n^a\times n^b$ matrix with an $n^b\times n^c$ matrix. 

It is known that $\omega(a,b,c)=\omega(a,c,b)=\omega(c,a,b)$ so that the value remains the same for any permutation of the dimensions $a,b,c$.
In this paper, we focus on bounds for $\omega(1,\kappa,1)$ for $\kappa > 0$. 

\subsection{Sch\"onhage's Asymptotic Sum Inequality}

From a bound $r$ on the rank of any $\ang{a,b,c}$ tensor, one can obtain an upper bound $\omega\leq 3\log_{abc}(r)$.
 Sch\"onhage \cite{Schonhage81} extended this fact by showing that one can obtain an upper bound on $\omega$ using any upper bound on the asymptotic rank of a {\em direct sum of matrix multiplication tensors}:

\begin{theorem}[Asymptotic Sum Inequality \cite{Schonhage81}]\label{thm:schonhage-ineq}
For positive integers $r > m$ and $a_i, b_i,c_i$ for $i\in [m]$, if 
\[\Tilde{R}\bk{\bigoplus_{i = 1}^m \ang{a_i,b_i,c_i}}\le r,\]
then $\omega\le 3\tau$ where $\tau\in [2/3,1]$ is the solution to the equation
\[\sum_{i = 1}^m (a_i\cdot b_i\cdot c_i)^\tau = r.\]
\end{theorem}

Analogously, the asymptotic sum inequality can also be used to obtain bounds on the rectangular matrix multiplication as follows.

\begin{theorem}[Asymptotic Sum Inequality for $\omega(a,b,c)$ \cite{Schonhage81,blaser}]\label{thm:schonhage-ineq-rect}
    Let $t, \, q > 0$ be positive integers and $a,b,c \ge 0$ , then
   \[t\cdot q^{\omega(a,b,c)}\le \Tilde{R}\lpr{\bigoplus_{i = 1}^t \ang{q^a, q^b,q^c}}.\]
\end{theorem}

\subsection{The Coppersmith-Winograd Tensor}

For a nonnegative integer $q\ge 0$, the Coppersmith-Winograd tensor $\CW_q$ over the variables $X = \{x_0,\dots, x_{q+1}\}$, $Y = \{y_0,\dots, y_{q+1}\}$, $Z = \{z_0,\dots, z_{q+1}\}$ is defined as
\[\CW_q \defeq x_0y_0z_{q+1} + x_0y_{q+1}z_0 + x_{q+1}y_0z_0 + \sum_{i = 1}^q \lpr{x_0y_iz_i + x_iy_0z_i + x_iy_iz_0}.\]
Coppersmith and Winograd \cite{cw90} showed that $\Tilde{R}(\CW_q) \le q+2$.

Observe that 
\[\CW_q \equiv \ang{1,1,q} + \ang{q,1,1} + \ang{1,q,1} + \ang{1,1,1}+\ang{1,1,1}+\ang{1,1,1},\]
where unfortunately \say{$+$} is a sum and not a direct sum, so that Sch\"{o}nhage's theorem doesn't immediately apply.
The goal of the laser method is to zero out variables in $\CW_q^{\otimes n}$ so that Sch\"{o}nhage's theorem can be applied.

\subsection{Base Leveled Partition of \texorpdfstring{\boldmath $\CW_q$}{CW\_q}}

Consider the $2^{\lvl-1}$-th tensor power of $\CW_q$ for $\lvl \ge 1$. For convenience, denote $T^{(\lvl)} \defeq \CW_q^{\otimes 2^{\lvl-1}}$. 

Coppersmith and Winograd \cite{cw90}
defined a partitioning of the variables of $\CW_q$ which is used in all following works including \cite{virgi12,AlmanW21,LeGall32power,duan2023}. 
We call this the level-$1$ partition. More generally, we
describe the leveled partition of $T^{(\lvl)}$ which also comes from \cite{cw90} and the subsequent works.

\paragraph{Level-1 partition.} For $T^{(1)} = \CW_q$, its variable sets $X^{(1)}, Y^{(1)}, Z^{(1)}$ are partitioned into three parts
\begin{align*}
    X^{(1)} &= X^{(1)}_0 \sqcup X^{(1)}_1 \sqcup X^{(1)}_2 = \{x_0\}\sqcup \{x_1,\dots, x_q\}\sqcup \{x_{q+1}\},\\
    Y^{(1)} &= Y^{(1)}_0 \sqcup Y^{(1)}_1 \sqcup Y^{(1)}_2 = \{y_0\}\sqcup \{y_1,\dots, y_q\}\sqcup \{y_{q+1}\},\\
    Z^{(1)} &= Z^{(1)}_0 \sqcup Z^{(1)}_1 \sqcup Z^{(1)}_2 =  \{z_0\}\sqcup \{z_1,\dots, z_q\}\sqcup \{z_{q+1}\}.
\end{align*}
We  denote the subtensor $T^{(1)}\big\vert_{X^{(1)}_i, Y^{(1)}_j,Z^{(1)}_k}$ by $T_{i,j,k}^{(1)}$ and we call it a \emph{level-$1$ constituent tensor}. Note that $T^{(1)}_{i, j, k}$ is nonzero if and only if $i+j+k = 2$. In particular, we can write $\CW_q$ as a sum of constituent tensors as follows
\[T^{(1)} = \CW_q = \sum_{\substack{i,j,k\ge 0\\ i+j+k = 2}} T_{i,j,k}^{(1)}.\]

\paragraph{Level-\boldmath$\lvl$ partition.} For $T^{(\lvl)}=\CW_q^{\otimes 2^{\lvl-1}}$ with variable sets $X^{(\lvl)}, Y^{(\lvl)}, Z^{(\lvl)}$, the above level-$1$ partition on $T^{(1)}$ induces a natural partition on the variable sets $X^{(\lvl)}, Y^{(\lvl)}, Z^{(\lvl)}$ 
where parts are indexed by $\{0, 1, 2\}$-sequences of length $2^{\lvl-1}$. Specifically, for the $X$ variables we get the partition
\[X^{(\lvl)} = \bigsqcup_{(\hat{i}_1, \hat{i}_2, \ldots, \hat{i}_{2^{\lvl-1}}) \in \{0, 1, 2\}^{2^{\lvl-1}}}X^{(1)}_{\hat{i}_1} \otimes  X^{(1)}_{\hat{i}_2} \otimes \cdots \otimes X^{(1)}_{\hat{i}_{2^{\lvl-1}}}.\]
The partitions for $Y$- and $Z$-variables are analogous.

Prior work starting with \cite{cw90} used partitions obtained from the above induced partitions by merging parts whose sequences indices sum to the same number.

\[
  X^{(\lvl)} = \bigsqcup_{i = 0}^{2^\lvl} X_i^{(\lvl)},
  \qquad \textup{where} \quad
  X_i^{(\lvl)} \defeq
  \bigsqcup_{\substack{(\hat{i}_1, \hat{i}_2, \ldots, \hat{i}_{2^{\lvl-1}}) \in \{0, 1, 2\}^{2^{\lvl-1}}  \\\sum_t \hat{i}_t = i}}X^{(1)}_{\hat{i}_1} \otimes  X^{(1)}_{\hat{i}_2} \otimes \cdots \otimes X^{(1)}_{\hat{i}_{2^{\lvl-1}}}.
\]
This above partition of $T^{(\lvl)}$ is called the \emph{level-$\lvl$ partition}.
One can also view this partition as a coarsening of the level-($\lvl-1$) partition via merging, i.e.,
\[X_i^{(\lvl)} = \bigsqcup_{\substack{0 \le i' \le i \\ 0 \le i', i - i' \le 2^{\lvl - 1}}}X^{(\lvl-1)}_{i'} \otimes X^{(\lvl-1)}_{i-i'}.\]
The variable sets $Y^{(\lvl)}$ and $Z^{(\lvl)}$ are partitioned similarly. 
 
 Under the level-$\lvl$ partition, denote the subtensor $T^{(\lvl)} \big\vert_{X_i^{(\lvl)}, Y_j^{(\lvl)}, Z_k^{(\lvl)}}$ by $T^{(\lvl)}_{i, j, k}$ and call it a \emph{level-$\lvl$ constituent tensor}.
 Call $X^{(\lvl)}_i, Y^{(\lvl)}_j, Z^{(\lvl)}_k$ level-$\lvl$ variable blocks. We omit the superscript $(\lvl)$ when $\lvl$ is clear from context.
 
 Note that $T^{(\lvl)}_{i, j, k}\neq 0$ if and only if $i+j+k = 2^{\lvl}$. Thus:
 
 \[T^{(\lvl)} = \CW_q^{\otimes 2^{\lvl-1}}= \sum_{\substack{i,j,k\ge 0\\ i+j+k = 2^{\lvl}}} T_{i,j,k}^{(\lvl)}.\]

\subsection{Leveled Partition for Large Tensor Powers of \texorpdfstring{\boldmath $\CW_q$}{CW\_q}}

In the laser method, we focus on a tensor power of $\CW_q$ in the form $(T^{(\lvl)})^{\otimes n} = (\CW_q)^{\otimes n\cdot 2^{\lvl-1}}$. We set $N \defeq n\cdot 2^{\lvl-1}$ and note that the leveled partition of $T^{(\lvl)}$ induces a partition on $(T^{\lvl})^{\otimes n}$. We recall some basic terminology and notation with respect to the leveled partition of $(T^{\lvl})^{\otimes n}$.

\paragraph{Level-1 partition of \boldmath$(\CW_q)^{\otimes N}$.}
In level-$1$, we view $(\CW_q)^{\otimes N}$ as the tensor $(T^{(1)})^{\otimes N}$ and consider the partition induced by the level-$1$ partition on $T^{(1)}$. Each level-$1$ $X$-variable block $X_{\hat{I}}$ is indexed by a sequence $\hat{I} = (\hat{I}_1,\dots, \hat{I}_N)$ of length $N$ in $\{0,1,2\}^N$. The variable block $X_{\hat I}$ is defined as 
\[X_{\hat{I}} \defeq X_{\hat{I}_1}^{(1)} \otimes \dots \otimes X_{\hat{I}_N}^{(1)},\]
where $X_{\hat I_t}^{(1)}$ for $t\in [N]$ is the level-$1$ partition of $T^{(1)}$. We call $X_{\hat{I}}$ a \emph{level-$1$ variable block} and $\hat{I}$ its \emph{level-$1$ index sequence}. The level-$1$ $Y$- and $Z$-variable blocks $Y_{\hat{J}}$ and $Z_{\hat{K}}$ are defined similarly for level-$1$ index sequences $\hat{J}, \hat{K}\in \{0,1,2\}^N$. Then notice that $X_{\hat{I}}, Y_{\hat{J}}, Z_{\hat{K}}$ form a nonzero subtensor of $(T^{(1)})^{\otimes N}$ if $\hat{I}_t + \hat{J}_t + \hat{K}_t = 2$ for all $t\in [N]$. Thus we write $(T^{(1)})^{\otimes N}$ as a sum of subtensors
\[(T^{(1)})^{\otimes N} = \sum_{\substack{\hat{I}, \hat{J}, \hat{K}\in \{0,1,2\}^{N}\\ \hat{I}_t + \hat{J}_t + \hat{K}_t = 2\,\ \forall t\in [N]}} (T^{(1)})^{\otimes N} \big\vert_{X_{\hat{I}}, Y_{\hat{J}}, Z_{\hat{K}}}.\]
For convenience, we use $X_{\hat{I}}Y_{\hat{J}}Z_{\hat{K}}$ to denote the subtensor $(T^{(1)})^{\otimes N}\big\vert_{X_{\hat{I}}, Y_{\hat{J}}, Z_{\hat{K}}}$ and we call $X_{\hat{I}}Y_{\hat{J}}Z_{\hat{K}}$ a \emph{level-$1$ triple}.

\paragraph{Level-\boldmath$\lvl$ partition of $(\CW_q)^{\otimes N}$.}
In level-$\lvl$, we view $(\CW_q)^{\otimes N}$ as the tensor $(T^{(\lvl)})^{\otimes n}$ where $n = N / 2^{\lvl - 1}$ and consider the partition induced by the level-$\lvl$ partition on $T^{(\lvl)}$. Each level-$1$ $X$-variable block $X_{I}$ is indexed by a sequence $I\in \{0,1,\dots, 2^{\lvl}\}^n$ of length $n$. The variable block $X_I$ is defined as 
\[X_I \defeq  X_{I_1}^{(\lvl)}\otimes \dots \otimes X_{I_n}^{(\lvl)}\]
where $X_{i}^{(\lvl)}$ ($0 \le i \le 2^{\lvl}$) is the $i$-th part in the level-$\lvl$ partition of $T^{(\lvl)}$. We call $X_{I}$ a \emph{level-$\lvl$ variable block} and $I$ its \emph{level-$\lvl$ index sequence}. The level-$\lvl$ $Y$- and $Z$-variable blocks $Y_{J}$ and $Z_{K}$ are defined similarly for level-$\lvl$ index sequences $J,K\in \{0,1,\dots, 2^{\lvl}\}^n$. Similarly, the level-$\lvl$ variable blocks $X_I, Y_J, Z_K$ form a nonzero subtensor of $(T^{(\lvl)})^{\otimes n}$ when $I_t + J_t + K_t = 2^\lvl$ for all $t\in [n]$. So we can write
\[(T^{(\lvl)})^{\otimes n} = \sum_{\substack{I, J, K \in \{0,1,2^\lvl\}^{n} \\ I_t + J_t + K_t = 2^\lvl\,\ \forall t\in [N]}}(T^{(\lvl)})^{\otimes n}\big\vert_{X_I, Y_J, Z_K}.\]
As in the level-$1$ partition, for convenience, we use the notation $X_I Y_J Z_K$ to denote the subtensor $(T^{(\lvl)})^{\otimes n}\vert_{X_I, Y_J, Z_K}$ and we call such $X_I Y_J Z_K$ a \emph{level-$\lvl$ triple.}

In addition, note that since the level-$\lvl$ partition of $T^{(\lvl)}$ is a coarsening of the partition induced by the level-$1$ partition of $T^{(1)}$, a level-$1$ variable block $X_{\hat{I}}$ is contained in a level-$\lvl$ variable block $X_I$ if the sequence $I' = (I'_1,\dots, I'_n)$ formed by taking $I'_t = \sum_{i = 1}^{2^{\lvl-1}} \hat{I}_{(t-1)\cdot 2^{\lvl-1} + i}$ satisfies $I'_t = I_t$ for all $t\in [n]$. Namely, if taking the sum of consecutive length-$2^{\lvl-1}$ subsequences in $\hat{I}$ yields the sequence $I$, then $X_{\hat{I}}$ is contained in $X_I$. In this case, we use the notation $\hat{I}\in I$ and $X_{\hat{I}}\in X_I$.

\subsection{Distributions and Entropy}
We only consider distributions with a finite support. Let $\alpha$ be a distribution supported on a set $S$, i.e., we have $\alpha(s)\ge 0$ for all $s\in S$ and $\sum_{s\in S}\alpha(s) = 1$. The \emph{entropy} of $\alpha$, denoted as $H(\alpha)$, is defined as\[H(\alpha) \defeq - \sum_{\substack{s\in S\\ \alpha(s) > 0}}\alpha(s)\log \alpha(s), \]
where the $\log$ has base $2$. 
We use the following well-known combinatorial fact.
\begin{lemma}
  Let $\alpha$ be a distribution over the set $[s] = \{1,\dots, s\}$. Let $N > 0$ be a positive integer, then we have
  \[\binom{N}{\alpha(1)N,\dots, \alpha(s)N}= 2^{N(H(\alpha) \pm o(1))}.\]
\end{lemma}

 For two distributions $\alpha$ and $\beta$ over the sets $S$ and $S'$ respectively, we define the joint distribution $\alpha \times \beta$ as the distribution over $S \times S' = \{(s,s') \mid s \in S, \, s' \in S'\}$ such that
 \[(\alpha\times \beta) (s,s') = \alpha(s)\cdot \beta(s').\]
  When $S$ and $S'$ are sets of integer sequences, we will instead define $\alpha \times \beta$ as a distribution over all integer sequences that can be obtained by concatenating one sequence in $S$ and one sequence in $S'$, such that 
$$(\alpha \times \beta)(s \circ s') = \alpha(s) \cdot \beta(s'), $$
where $s \circ s'$ denotes the concatenation of $s$ and $s'$. 

\subsection{Complete Split Distributions}

Motivated by the leveled partition of tensor powers of $\CW_q$, the complete split distribution is defined to characterize the level-$1$ variable blocks contained in level-$\lvl$ variable blocks. (Complete split distributions were first defined and used by \cite{VXXZ24}.)

\begin{definition}[Complete Split Distribution]
    A \emph{complete split distribution} for a level-$\lvl$ constituent tensor $T_{i,j,k}$ with $i+j+k=2^{\lvl}$ is a distribution on all length $2^{\lvl-1}$ sequences $(\hat{i}_1, \hat{i}_2, \ldots, \hat{i}_{2^{\lvl-1}}) \in \{0, 1, 2\}^{2^{\lvl-1}}$. 
\end{definition}

For a level-$1$ index sequence $\hat I \in \{0, 1, 2\}^{2^{\lvl-1} \cdot n}$, we say that it is \emph{consistent} with a complete split distribution $\splres$ if the proportion of any  index sequence $(\hat{i}_1, \hat{i}_2, \ldots, \hat{i}_{2^{\lvl-1}})$ in
\[ \left\{ \bigbk{\hat{I}_{(t-1)\cdot 2^{\lvl - 1}+p}}_{p=1}^{2^{\lvl-1}} \;\middle|\; t \in [n] \right\} \]
equals $\splres(\hat{i}_1, \hat{i}_2, \ldots, \hat{i}_{2^{\lvl-1}})$. Namely, for every $(\hat{i}_1,\dots, \hat{i}_{2^{\lvl-1}})\in \{0,1,2\}^{2^{\lvl-1}}$, we have
\[\labs{\lcr{t \in [n] \mmid \bigbk{\hat{I}_{(t-1)\cdot 2^{\lvl - 1}+p}}_{p=1}^{2^{\lvl - 1}} = (\hat{i}_1,\dots, \hat{i}_{2^{\lvl-1}})}}  = \splres(\hat{i}_1, \hat{i}_2, \ldots, \hat{i}_{2^{\lvl-1}}) \cdot n.\]

Notice that any level-1 index sequence $\hat{I}\in \{0,1,2\}^{2^{\lvl - 1}\cdot n}$ defines a complete split distribution by computing the proportion of each type of length-$2^{\lvl-1}$ consecutive chunks present in $\hat{I}$. More specifically, we have the following definition.

\begin{definition}\label{def:split-hatI}
    Given a level-1 index sequence $\hat{I}\in \{0,1,2\}^{2^{\lvl - 1}\cdot n}$, its complete split distribution over $(\hat{i}_1,\dots, \hat{i}_{2^{\lvl-1}})\in \{0,1,2\}^{2^{\lvl - 1}}$ is defined as
    \[\split\bigbk{\hat{I}}\bigbk{\hat{i}_1,\dots, \hat{i}_{2^{\lvl-1}}} = \frac{1}{n}\cdot \labs{\lcr{t \in [n] \mmid \bigbk{\hat{I}_{(t-1)\cdot 2^{\lvl - 1}+p}}_{p=1}^{2^{\lvl - 1}} = \bigbk{\hat{i}_1,\dots, \hat{i}_{2^{\lvl-1}}}}}.\]
    Given a subset $S\subseteq [n]$, we can define the complete split distribution over $(\hat{i}_1,\dots, \hat{i}_{2^{\lvl-1}})\in \{0,1,2\}^{2^{\lvl - 1}}$ given by $\hat{I}$ restricted to the subset $S$ as 
    \[\split\bigbk{\hat{I}, S}\bigbk{\hat{i}_1,\dots, \hat{i}_{2^{\lvl-1}}} = \frac{1}{|S|}\cdot \labs{\lcr{t \in S \mmid \bigbk{\hat{I}_{(t-1)\cdot 2^{\lvl - 1}+p}}_{p=1}^{2^{\lvl - 1}} = \bigbk{\hat{i}_1,\dots, \hat{i}_{2^{\lvl-1}}}}}.\]
\end{definition}

Given two complete split distributions $\splres_1$ and $\splres_2$ over the length-$2^{\lvl-1}$ index sequences $\{0,1,2\}^{2^{\lvl - 1}}$, the $L_\infty$ distance between $\splres_1$ and $\splres_2$ is defined to be 
\[\norm{\splres_1 - \splres_2}_\infty = \max_{\sigma\in \{0,1,2\}^{2^{\lvl-1}}}|\splres_1(\sigma) - \splres_2(\sigma)|.\]
For any constant $\eps > 0$ and a fixed complete split distribution $\splres$, we say that a level-1 index sequence $\hat{I}\in \{0,1,2\}^{2^{\lvl - 1}\cdot n}$ is consistent with $\splres$ up to $\eps$ error if  $\midnorm{\split(\hat{I}) - \splres}_\infty \le \eps$. When the $\eps$ is clear from context, we say that $\hat{I}$ is \emph{approximately consistent} with $\splres$ if it is consistent with $\splres$ up to $\eps$ error.

\begin{definition}
    For a level-$\lvl$ constituent tensor $T_{i,j,k}$, an integer exponent $N$, a constant $\eps \ge 0$, and three complete split distributions $\splresX, \splresY, \splresZ$ for the $X$-, $Y$-, $Z$-variables respectively, we define
    $$T_{i,j,k}^{\otimes N}[\splresX, \splresY, \splresZ,\eps] \defeq \sum_{\substack{\text{level-}1 \text{ triple } X_{\hat I}Y_{\hat J}Z_{\hat K} \textup{ in } T_{i,j,k}^{\otimes N} \\ \hat I \text{ approximately consistent with } \splresX \\ \hat J \text{ approximately consistent with } \splresY \\ \hat K \text{ approximately consistent with } \splresZ}} X_{\hat I} Y_{\hat J} Z_{\hat K}. $$
    It is a subtensor of $T_{i,j,k}^{\otimes N}$ over all level-1 $X$-, $Y$-, $Z$-variable blocks that are approximately consistent with $\splresX$, $\splresY$, $\splresZ$, respectively. When $\eps = 0$, we will simplify the notation to $T_{i,j,k}^{\otimes N}[\splresX, \splresY, \splresZ]$.
\end{definition}

\subsection{Salem-Spencer Sets}
In the hashing step of the laser method, we make use of the existence of a large subset of $\Z_M$ that avoids $3$-term arithmetic progressions, as given by Salem and Spencer \cite{salemspencer} and improved by Behrend \cite{behrend1946sets}.

\begin{theorem}[\cite{behrend1946sets,elsholtz2024improving}]
\label{thm:salemspencer}
 For every positive integer $M > 0$, there exists a subset $B\subseteq \Z_M$ of size 
 \[|B|\ge M\cdot e^{-O(\sqrt{\log M})} = M^{1-o(1)}\]
 that contains no nontrivial $3$-term arithmetic progressions. Specifically, any $a,b,c\in B$ satisfy $a+b \equiv 2c \pmod M$ if and only if $a = b = c$.
\end{theorem}

\section{Algorithm Outline}
\label{sec:outline}

Our high-level framework is similar to previous works \cite{duan2023, VXXZ24}, and we borrow definitions such as interface tensors, compatibility, and usefulness. 
We begin by defining the notion of interface tensors, which captures the structure of the intermediate tensors we obtain from different stages of our algorithm and pass to the next stages.

\begin{definition}[Interface Tensor]\label{def:interface-tensor}
  For a positive integer $\lvl\ge 1$ and any constant $0\le \eps \le 1$, a level-$\lvl$ $\eps$-\emph{interface tensor} $\T^*$ with parameter list 
  \[\{(n_t, i_t, j_t, k_t, \splresXt, \splresYt, \splresZt)\}_{t \in [s]}\]
  is defined as
  \[
    \T^* \defeq \bigotimes_{t = 1}^{s} T_{i_t, j_t, k_t}^{\otimes n_t}[\splresXt, \splresYt, \splresZt, \eps],
  \]
  where $i_t+j_t+k_t = 2^{\lvl}$ for every $t\in [s]$ (i.e., $T_{i_t, j_t, k_t}$ is a level-$\lvl$ constituent tensor) and  $\splresXt, \splresYt, \splresZt$ are level-$\lvl$ complete split distributions for $X$-, $Y$-, $Z$-variables respectively. We call each $T_{i_t, j_t, k_t}^{\otimes n_t}[\splresXt, \splresYt, \splresZt, \eps]$ a \emph{term} of $\T^*$. When $\eps = 0$, we will simply call $\T^*$ a level-$\lvl$ interface tensor. 
\end{definition}

Our algorithm takes $\bigbk{\CW_q^{\otimes 2^{\lvl^*-1}}}^{\otimes n}$ as input, and applies the global stage of the algorithm described in \cref{sec:global} to degenerate it into independent copies of a  level-$\lvl^*$ $\eps_{\lvl^*}$-interface tensor. For every $\lvl \in [\lvl^*]$, we will have a parameter $\eps_\lvl$ controlling the error margin of level-$\lvl$ interface tensors. Then for $\lvl = \lvl^*, \lvl^* \!-\! 1, \ldots, 2$, we apply the constituent stage of the algorithm described in \cref{sec:constituent} to degenerate a level-$\lvl$ $\eps_\lvl$-interface tensor into a tensor product between independent copies of level-$(\lvl-1)$ $\eps_{\lvl-1}$-interface tensors and a matrix multiplication tensor. Eventually, we obtain a tensor product between independent copies of level-$1$ $\eps_1$-interface tensors and a matrix multiplication tensor, which can then be degenerated into independent copies of matrix multiplication tensors. Thus, the overall algorithm degenerates $\bigbk{\CW_q^{\otimes 2^{\lvl^*-1}}}^{\otimes n}$ into independent copies of matrix multiplication tensors as desired.

\subsection{Algorithm Outline}\label{subsec:outline}

The framework for the constituent stage is similar to the framework for the global stage.
For simplicity, we only present the outline of our global stage algorithm for $\eps = 0$ as it captures the main ideas of our algorithm and illustrates the main differences between our algorithm and \cite{VXXZ24}'s. A high-level comparison between our algorithm and \cite{VXXZ24}'s is provided in \cref{fig:outline-compare}. 

\begin{figure}
    \centering
    \tikzset{>={Latex[width=1.5mm,length=1mm]},
 base/.style = {rectangle, rounded corners, draw=black, minimum width=11cm, minimum height=1cm, text centered},
 regularblock/.style = {base},
 small/.style = {rectangle, rounded corners, draw=black, minimum width=3.25cm, minimum height=1.5cm, text centered},
 big/.style = {rectangle, rounded corners, draw=black, minimum width=11cm, minimum height=2cm, text centered}
}

\begin{tikzpicture}[scale=0.7,every node/.style={scale=0.7}]

\node (input) at (0,0) [regularblock] {\textbf{Input}: $(\CW_q^{\otimes 2^{\ell-1}})^{\otimes n}$};

\node (alpha) at (0,-2) [regularblock] {Zero out level-$\ell$ blocks according to $\alpha$};

\node (hash) at (0,-4) [regularblock] {Asymmetric hashing};

\node[align=left] (Zcomp1) at (-3.875, -7) [small] {$Z$-Compatibility\\ Zero out I};

\node[align=left] (Zcomp2) at (0, -7) [small] {$Z$-Compatibility\\ Zero out II};

\node[align=left] (Zcomp3) at (3.875, -7) [small] {$Z$-Usefulness\\ Zero out};

\node (fix) at (0,-9.75) [regularblock] {Fix holes};

\node (out) at (0,-11.75) [regularblock] {\textbf{Output}: Independent copies of interface tensors};

\node[scale=1.3] at (0,-13) {(a) Global stage algorithm outline in \cite{VXXZ24}};

\draw[->] (input) -- (alpha);
\draw[->] (alpha) -- (hash);
\draw[->] (Zcomp1) -- (Zcomp2);
\draw[->] (Zcomp2) -- (Zcomp3);
\draw[->] (fix) -- (out);

\draw[->] (5.5,-4) -- (5.75, -4) -- (5.75, -5.5) -- (-5.75, -5.5) -- (-5.75, -7) -- (-5.5, -7);
\draw[->] (5.5,-7) -- (5.75, -7) -- (5.75, -8.5) -- (-5.75, -8.5) -- (-5.75,-9.75) -- (-5.5, -9.75);

\node[blue] at (0,-5.25) {Every level-$\ell$ $X$-block and $Y$-block is in a unique level-$\ell$ triple};
\node[blue] at (0,-8.25) {Every level-$1$ $Z$-block is in a unique level-$\ell$ triple};

\coordinate (shift1) at (12.5,0);
\begin{scope}[shift=(shift1)]

\node (input) at (0,0) [regularblock] {\textbf{Input}: $(\CW_q^{\otimes 2^{\ell-1}})^{\otimes n}$};

\node (alpha) at (0,-1.5) [regularblock] {Zero out level-$\ell$ blocks according to $\alpha$};

\node (hash) at (0,-3) [regularblock,fill=green!30] {More asymmetric hashing};

\node[align=left,fill=green!30] (Ycomp1) at (-3.875, -5.25) [small] {$Y$-Compatibility\\ Zero out I};

\node[align=left,fill=green!30] (Ycomp2) at (0, -5.25) [small] {$Y$-Compatibility\\ Zero out II};

\node[align=left,fill=green!30] (Ycomp3) at (3.875, -5.25) [small] {$Y$-Usefulness\\ Zero out};

\node[align=left] (Zcomp1) at (-3.875, -7.75) [small] {$Z$-Compatibility\\ Zero out I};

\node[align=left] (Zcomp2) at (0, -7.75) [small] {$Z$-Compatibility\\ Zero out II};

\node[align=left] (Zcomp3) at (3.875, -7.75) [small] {$Z$-Usefulness\\ Zero out};

\node (fix) at (0,-10.25) [regularblock] {Fix holes};

\node (out) at (0,-11.75) [regularblock] {\textbf{Output}: Independent copies of interface tensors};

\draw[->] (input) -- (alpha);
\draw[->] (alpha) -- (hash);
\draw[->] (Ycomp1) -- (Ycomp2);
\draw[->] (Ycomp2) -- (Ycomp3);
\draw[->] (Zcomp1) -- (Zcomp2);
\draw[->] (Zcomp2) -- (Zcomp3);
\draw[->] (fix) -- (out);

\node[blue] at (0,-4) {Every level-$\ell$ $X$-block is in a unique level-$\ell$ triple};

\draw[->] (5.5,-3) -- (5.75, -3) -- (5.75, -4.25) -- (-5.75, -4.25) -- (-5.75, -5.25) -- (-5.5, -5.25);

\node[blue] at (0,-6.5) {Every level-$1$ $Y$-block is in a unique level-$\ell$ triple};

\draw[->] (5.5,-5.25) -- (5.75, -5.25) -- (5.75, -6.75) -- (-5.75, -6.75) -- (-5.75, -7.75) -- (-5.5, -7.75);

\node[blue] at (0,-9) {Every level-$1$ $Z$-block is in a unique level-$\ell$ triple};

\draw[->] (5.5,-7.75) -- (5.75, -7.75) -- (5.75, -9.25) -- (-5.75, -9.25) -- (-5.75,-10.25) -- (-5.5, -10.25);

\node[scale=1.3] at (0,-13) {(b) Global stage algorithm outline in this work};
\end{scope}
    
\end{tikzpicture}
    \caption{Main differences between the global stage algorithm in \cite{VXXZ24} and this work as outlined in \cref{subsec:outline}. Other technical differences are omitted.}
    \label{fig:outline-compare}
\end{figure}

As mentioned earlier, the goal of the global stage is to take $\bigbk{\CW_q^{\otimes 2^{\lvl-1}}}^{\otimes n}$ as input (for convenience, we will use $\lvl$ to denote $\lvl^*$), and output a set of level-$1$-independent level-$\lvl$ interface tensors. Since our algorithm treats the $X, Y, Z$ dimensions asymmetrically and there are six permutations of the $X, Y, Z$ dimensions, we can essentially apply our algorithm in six different ways. As a result, we will split $\bigbk{\CW_q^{\otimes 2^{\lvl-1}}}^{\otimes n}$ into $\bigotimes_{r \in [6]} \bigbk{\CW_q^{\otimes 2^{\lvl-1}}}^{\otimes A_r \cdot n}$ for some $A_1, \ldots, A_6 \ge 0$  and $\sum_{r = 1}^6 A_r = 1$, where we call each $\bigbk{\CW_q^{\otimes 2^{\lvl-1}}}^{\otimes A_r \cdot n}$ a \emph{region}. We apply our algorithm for each permutation of the $X, Y, Z$ dimensions on one of the regions. We note that in \cite{VXXZ24}, two of the three dimensions are treated symmetrically, so they only need to split the input tensor into $3$ regions instead of $6$. For simplicity, we outline each step of our algorithm on a fixed permutation without splitting $\bigbk{\CW_q^{\otimes 2^{\lvl-1}}}^{\otimes n}$ into regions in the following.

\begin{enumerate}
\item \textbf{Zero-out according to \boldmath$\alpha$.} This is a standard first step present in all previous applications of the laser method. For some distribution $\alpha$ over level-$\lvl$ constituent tensors, let $\alphx, \alphy, \alphz$ be the marginal distributions of $\alpha$ on the $X$-, $Y$-, $Z$-dimensions respectively. We zero out $X$-, $Y$-, $Z$-variable blocks that are not consistent with $\alphx, \alphy, \alphz$ respectively. 

\item \textbf{Asymmetric hashing.} Similar to \cite{duan2023, VXXZ24}, we then apply asymmetric hashing. Unlike previous works, where after the asymmetric hashing each level-$\lvl$ $X$-variable block $X_I$ or $Y$-variable block $Y_J$ is contained in a unique level-$\lvl$ triple $X_I Y_J Z_K$, we only require each level-$\lvl$ $X$-block to be contained in a unique triple. Moreover, the number of triples containing each level-$\lvl$ $Y$-block is at most the number of triples containing each level-$\lvl$ $Z$-block.

\item \textbf{\boldmath$Y$-compatibility zero-outs.} In \cite{duan2023, VXXZ24}, the next step is to do some zeroing outs based on a notion of compatibility defined for level-$1$ $Z$-blocks, provided that every level-$1$ $X$- and $Y$-block is contained in a unique block triple. So in our algorithm, we first need to ensure that each level-$1$ $Y$-block is contained in a unique block triple by performing the following sub-steps. 

Note that our goal is to obtain an interface tensor with  $\left\{\splres_{\itX, i, j, k}, \splres_{\itY, i, j, k}, \splres_{\itZ, i, j, k}\right\}_{i+j+k = 2^\lvl}$ as its complete split distributions. Since each level-$\lvl$ $X$-block is contained in a unique block triple, we are already able to identify the subset of indices $t$ with $X_t = i$, $Y_t = j$, $Z_t = k$ for every $i, j, k$. This means that we can compute the complete split distributions for each level-$1$ $X$-block with respect to these indices, and zero out any level-$1$ $X$-block that is not consistent with $\left\{\splres_{\itX, i, j, k}\right\}_{i+j+k = 2^\lvl}$. 

\begin{enumerate}
    \item \textbf{\boldmath$Y$-compatibility zero-out I.} Then we define a notion of compatibility for a level-$1$ $Y$-block with respect to a level-$\lvl$ block triple containing it. We say a level-$1$ $Y$-block is compatible with a level-$\lvl$ block triple if they satisfy the compatibility conditions. We zero out all level-$1$ $Y$-blocks that are not compatible with any triple. 
    \item \textbf{\boldmath$Y$-compatibility zero-out II.} In this step, we zero out the level-$1$ $Y$-blocks that are compatible with more than one triple. After this step, we can ensure that each level-$1$ $Y$-block is contained in a unique level-$\lvl$ block triple and it is compatible with that triple. 
    \item \textbf{\boldmath$Y$-usefulness zero-out.} At this point, each level-$1$ $Y$-block is contained in a unique triple, so we can zero out a level-$1$ $Y$-block if it is not consistent with $\left\{\splres_{\itY, i, j, k}\right\}_{i+j+k = 2^\lvl}$, similar to what we did to the level-$1$ $X$-blocks. 
\end{enumerate}

\item \textbf{\boldmath$Z$-compatibility zero-outs.} Although after the more asymmetric hashing step, we only have $X$-blocks in a unique triple, notice that after $Y$-compatibility zero-outs, we achieve the property that now every level-$1$ $X$- and $Y$-block is compatible with a unique level-$\ell$ triple. This property is weaker than the property achieved in \cite{duan2023, VXXZ24} after their original asymmetric hashing steps, but this is sufficient for later steps. In particular, in our $Z$-compatibility zero-out, we can perform essentially the same procedure as in \cite{VXXZ24} to achieve the property that every level-$1$ $Z$-block is compatible with a unique level-$\lvl$ triple.

\item  \textbf{Fixing holes.} Now every level-$1$ block is contained in a unique level-$\ell$ triple, but some of the level-$1$ blocks we zeroed out in $Y$-compatibility zero-outs and $Z$-compatibility zero-outs become holes, i.e., the copies of the interface tensors we obtain miss some variables due to the compatibility zero-outs. We use the following result to fix these holes.

\begin{theorem}[{\cite[Corollary 3.2]{VXXZ24}}]\label{thm:fix-holes}
    Let $T$ be a level-$\ell$ interface tensor with parameter list
    \[\{(n_t, i_t, j_t, k_t, \splresXt, \splresYt, \splresZt)\}_{t \in [s]}.\]
    Let $N = 2^{\ell-1} \cdot \sum_{t\in [s]} n_t$. Suppose $T_1, \dots, T_r$ are broken copies of $T$ where $\le \frac{1}{8N}$ fraction of level-1 $X$, $Y$, and $Z$-blocks are holes. If $r \ge 2^{C_1 N / \log N}$ for some large enough constant $C_1>0$, the direct sum $\bigoplus_{i = 1}^r T_i$ can degenerate into an unbroken copy of $T$.
\end{theorem}

\end{enumerate}

\section{Global Stage}\label{sec:global}

For our global stage algorithm, the input is the tensor $\CW_q^{\otimes N}$ where $N = n\cdot 2^{\lvl^*-1}$. The algorithm will output a set of independent copies of level-$\lvl^*$ interface tensors that are degenerated from the input tensor. For convenience, we use $\lvl$ to denote $\lvl^*$ in the rest of this section.

Given a distribution $\alpha$ over the set $\{(i, j, k) \in \mathbb{Z}_{\ge 0}^3 \mid i + j + k = 2^{\lvl}\}$ and level-$\lvl$ complete split distributions $\splres_{\itX, i, j, k}$, $\splres_{\itY, i, j, k}$, $\splres_{\itZ, i, j, k}$ we define the quantities listed in \cref{tab:notation}.

\begin{table}[ht]
    \centering
    {\def\arraystretch{1.5}
    \begin{tabular}{|| c  p{12cm} ||}
    \hline
       Notation  & Definition \\
       \hline\hline
        \multirow{1.7}*{$\alphx, \alphy,\alphz$} & The marginal distributions of $\alpha$ on the $X,Y,Z$-dimension respectively. E.g., for any $i$, $\alphx(i) = \sum_{j, k} \alpha(i, j, k)$; similarly for $\alphy(j), \alphz(k)$.\\

        \hline
        
        \multirow{1.7}*{$P_\alpha$} & The penalty term $P_\alpha \defeq \max_{\alpha' \in D} H(\alpha') - H(\alpha) \ge 0$ where $D$ is the set of distributions with marginals $\alphx, \alphy, \alphz$.\\

        \hline

        $\alpha(i,\+, \+)$ & $\alpha(i, \+, \+) \defeq \sum_{j > 0, k > 0} \alpha(i, j, k)$\\

        $\alpha(\+,j,\+)$ & $\alpha(\+,j, \+) \defeq \sum_{i > 0, k > 0} \alpha(i, j, k)$ \\

        $\alpha(\+,\+,k)$ & $\alpha(\+, \+,k) \defeq \sum_{i > 0, j > 0} \alpha(i, j, k)$ \\

        \hline

        $\splresavg_{\itX, i, \+, \+}$ &  $\splresavg_{\itX, i, \+, \+}\defeq \frac{1}{\alpha(i,\+, \+)} \sum_{j > 0, k > 0} \alpha(i, j, k) \cdot \splres_{\itX, i, j, k}$\\

        $\splresavg_{\itY, \+, j, \+}$ &  $\splresavg_{\itY, \+, j, \+}\defeq \frac{1}{\alpha(\+, j, \+)} \sum_{i > 0, k > 0} \alpha(i, j, k) \cdot \splres_{\itY, i, j, k}$\\

        $\splresavg_{\itZ, \+, \+, k}$ & $\splresavg_{\itZ, \+, \+, k}\defeq \frac{1}{\alpha(\+, \+, k)} \sum_{i > 0, j > 0} \alpha(i, j, k) \cdot \splres_{\itZ, i, j, k}$\\

        \hline

        $\splavg{\itX}$ & $\splavg{\itX} \defeq \sum_{i, j, k} \alpha(i, j, k) \cdot \splres_{\itX, i, j, k}$ \\

        $\splavg{\itY}$ & $\splavg{\itY} \defeq \sum_{i, j, k} \alpha(i, j, k) \cdot \splres_{\itY, i, j, k}$ \\

        $\splavg{\itZ}$ & $\splavg{\itZ} \defeq \sum_{i, j, k} \alpha(i, j, k) \cdot \splres_{\itZ, i, j, k}$ \\

        \hline

        $\splresavg_{\itX, i,*,*}$ & $\splresavg_{\itX, i,*,*} = \frac{1}{\alpha(i,*,*)}\sum_{j+k = 2^\lvl - i} \alpha(i, j, k) \cdot \splres_{\itX,i,j,k}$\\

        $\splresavg_{\itY, *,j,*}$ & $\splresavg_{\itY, *,j,*} = \frac{1}{\alpha(*,j,*)}\sum_{i+k = 2^\lvl - j} \alpha(i, j, k) \cdot \splres_{\itY,i,j,k}$\\

        $\splresavg_{\itZ, *,*,k}$ & $\splresavg_{\itZ, *,*,k} = \frac{1}{\alpha(*,*,k)}\sum_{i+j = 2^\lvl - k} \alpha(i, j, k) \cdot \splres_{\itZ,i,j,k}$\\

        \hline

        $\eta_\itY$ & $\eta_\itY = \sum_{i,j,k \,:\, k =0}\alpha(i,j,k)\cdot H(\splres_{\itY, i,j,k}) + \sum_{j}\alpha(*,j,\+)\cdot H(\splresavg_{\itY,*,j,\+})$\\

        $\lambda_\itZ$ & $\lambda_\itZ = \sum_{i, j, k \,:\, i = 0 \textup{ or } j = 0} \alpha(i, j, k) \cdot H(\splres_{\itZ, i, j, k}) + \sum_k \alpha(\+, \+, k) \cdot H(\splresavg_{\itZ, \+, \+, k})$\\

        \hline
        
    \end{tabular}
    }
    \caption{Table of notations with respect to a distribution $\alpha$ over $\{(i, j, k) \in \mathbb{Z}_{\ge 0}^3 \mid i + j + k = 2^{\lvl}\}$ and level-$\lvl$ complete split distributions $\splres_{\itX, i, j, k}, \splres_{\itY, i, j, k}, \splres_{\itZ, i, j, k}$.}
    \label{tab:notation}
\end{table}

In particular, our {\bf notations} follow the general rules below:

\begin{enumerate}
    \item Given a function $f(i,j,k)$, we may replace any of the input by the symbol $*$ or $\+$. If an input coordinate is a \say{$*$}, then it means the sum over $f$ evaluated at all the inputs that are $\ge 0$ in this input coordinate; if an input coordinate is a \say{$\+$}, then it means the sum over  $f$ evaluated at all the inputs that are $> 0$ in this input coordinate. For example, $f(*, j,k) = \sum_{i\ge 0} f(i,j,k)$ and $f(*,\+,k) = \sum_{i\ge 0, j>0} f(i,j,k)$.

    We will also use $\bar{f}$ together with the symbols $*$ or $\+$ to denote a weighted average. Typically, we use this notation on the complete split distributions $\beta$, and the average is weighted by $\alpha$. For instance, $\splresavg_{\itX, *, \+, k} = \frac{1}{\alpha(*, \+, k)} \sum_{i\ge 0, j>0} \alpha(i, j, k) \cdot \splres_{\itX, i, j, k}$. 
    
    \item Given a family of sets $S_{i,j,k}$, we may replace any of the subscripts by the symbol $*$ or $\+$. If any coordinate of the subscript is a ``$*$'', then it means the union over $S_{i,j,k}$ with subscript $\ge 0$ on this coordinate; if any coordinate of the subscript is a ``$\+$'', then it means the union over $S_{i,j,k}$ with subscript $> 0$ on this coordinate. For example, $S_{*,j,k} = \bigcup_{i\ge 0} S_{i,j,k}$ and $S_{*,\+,k}= \bigcup_{i\ge 0, j>0} S_{i,j,k}$.

\end{enumerate}

In the following proposition, for every $r\in [6]$, we have the distributions $\alpha^{(r)}$, $\splres_{\itX, i, j, k}^{(r)}$, $\splres_{\itY, i, j, k}^{(r)}$, $\splres_{\itZ, i, j, k}^{(r)}$ correspondingly and we use the superscript $(r)$ on variables obtained from $\alpha^{(r)}$, $\splres_{\itX, i, j, k}^{(r)}$, $\splres_{\itY, i, j, k}^{(r)}$, $\splres_{\itZ, i, j, k}^{(r)}$. 

\begin{prop}
\label{prop:global-stage-no-eps}
    $\bigbk{\CW_q^{\otimes 2^{\lvl-1}}}^{\otimes n}$ can be degenerated into a direct sum of $2^{(\sum_{r=1}^6 A_r E_r) n - o(n)}$ copies of a level-$\lvl$ interface tensor with parameter list 
    \[ \left\{\bk{ n \cdot A_r \cdot \alpha^{(r)}(i, j, k), i, j, k, \splres_{\itX, i, j, k}^{(r)}, \splres^{(r)}_{\itY, i, j, k}, \splres^{(r)}_{\itZ, i, j, k} }\right\}_{r \in [6], \, i + j + k = 2^{\lvl}} \]
    where
    \begin{itemize}
        \item $0 \le A_r \le 1$, $\sum_{r=1}^6 A_r = 1$;
        \item $\alpha^{(r)}$ for every $r \in [6]$ is a  distribution over $\{(i, j, k) \in \mathbb{Z}_{\ge 0}^3 \mid i + j + k = 2^{\lvl}\}$;
        \item For every $W \in \{X, Y, Z\}$, $\splres_{\itW, i, j, k}^{(r)}$ for $r \in [6]$, $i + j + k = 2^{\lvl}$ is a level-$\lvl$ complete split distribution;
        \item For each $r \in [6]$, define $\pi_r: \{X, Y, Z\} \to \{X, Y, Z\}$ as the $r$-th permutation in the lexicographic order. Then 
        \[
        E_r \defeq \min\left\{H\Bigbk{\alpha_{\pi_r(X)}^{(r)}} - P_\alpha^{(r)}, H\Bigbk{\splavg{\pi_r(Y)}^{(r)}} - \eta_{\pi_r(Y)}^{(r)},  H\Bigbk{\splavg{\pi_r(Z)}^{(r)}} - \lambda_{\pi_r(Z)}^{(r)}\right\}.
        \]
    \end{itemize}
\end{prop}

The following remark illustrates some simple relationships between the complete split distributions for the $X$, $Y$, $Z$-dimensions that we can assume without loss of generality, which will be useful in our algorithms.

\begin{remark}[\cite{VXXZ24}]
\label{rmk:assumptions_on_complete_split_dist}
Without loss of generality, we can assume that, for every $r, i, j, k$, and every $L \in \{0, 1, 2\}^{2^{\lvl-1}}$, 
\[
\splres^{(r)}_{\itX, i, 0, k}(L) = \splres^{(r)}_{\itZ, i, 0, k}(\vec{2}-L), \quad   
\splres^{(r)}_{\itZ, 0, j, k}(L) = \splres^{(r)}_{\itY, 0, j, k}(\vec{2}-L), \quad
\splres^{(r)}_{\itY, i, j, 0}(L) = \splres^{(r)}_{\itX, i, j, 0}(\vec{2}-L), 
\]
where $\vec{2}$ denotes the length-$(2^{\lvl-1})$ vector whose coordinates are all $2$, 
and 
\[
\splres^{(r)}_{\itX, i, j, k}(L) = 0 \text{ if } \sum_{t} L_t \ne i, \quad   
\splres^{(r)}_{\itY, i, j, k}(L) = 0 \text{ if } \sum_{t} L_t \ne j, \quad   
\splres^{(r)}_{\itZ, i, j, k}(L) = 0 \text{ if } \sum_{t} L_t \ne k,   
\]
because otherwise, the level-$\lvl$ interface tensor will be the zero tensor and the lemma will follow trivially. 
\end{remark}

The following is a corollary of \cref{prop:global-stage-no-eps}. We omit its proof as it is similar to the proof of \cite[Theorem 5.3]{VXXZ24}. 

\begin{theorem}
\label{thm:global-stage-with-eps}
For any $\eps > 0$, $2^{o(n)}$ independent copies of $(\CW_q^{\otimes 2^{\lvl-1}})^{\otimes n}$ can be degenerated into 
    $$2^{(\sum_{r=1}^6 E_r A_r - o_{1/\eps}(1)) n - o(n)}$$
    independent copies of a level-$\lvl$ $\eps$-interface tensor with parameter list 
    \[ \left\{ \bk{ n \cdot A_r \cdot \alpha^{(r)}(i, j, k), i, j, k, \splres_{\itX, i, j, k}^{(r)}, \splres^{(r)}_{\itY, i, j, k}, \splres^{(r)}_{\itZ, i, j, k} } \right\}_{r \in [6], i + j + k = 2^{\lvl}} \]
    where the constraints are the same as those in \cref{prop:global-stage-no-eps}.\footnote{$\oeps(1)$ denotes a function $f(\eps)$ where $f(\eps) \to 0$ as $\eps \to 0$. We also use $\oeps(n)$ to denote $\oeps(1) \cdot n$.}
\end{theorem}

In the remainder of this section, we prove \cref{prop:global-stage-no-eps}.

\subsection{Dividing into Regions}

As in previous works \cite{duan2023, VXXZ24}, we partition the $n$-th tensor power of $\CW_q^{\otimes 2^{\lvl-1}}$ into several ``regions''. For each different region, we use a different permutation for the roles of the $X$, $Y$, $Z$-dimensions. In \cite{duan2023, VXXZ24}, they partition $(\CW_q^{\otimes 2^{\lvl-1}})^{\otimes n}$ into three regions due to the fact that two out of the three dimensions are treated symmetrically in their algorithms. Since our method is more asymmetric, we need to partition $(\CW_q^{\otimes 2^{\lvl-1}})^{\otimes n}$ into six regions instead. More specifically, we consider 
\[\bigbk{\CW_q^{\otimes 2^{\lvl-1}}}^{\otimes n} \equiv \bigotimes_{r=1}^6 \bigbk{\CW_q^{\otimes 2^{\lvl-1}}}^{\otimes A_r\cdot n}\]
for $A_1, \ldots, A_6 \ge 0$ and $A_1 + \cdots + A_6 = 1$, and we denote the $r$-th region by $\T^{(r)}$. 

In the following, we will focus on the analysis for $\T^{(1)}$. The analysis for other regions is similar, but we permute the roles of the $X$, $Y$, $Z$-dimensions in different regions. For simplicity, we will drop the superscript $(1)$. 

\subsection{More Asymmetric Hashing}

This hashing step is standard in all previous literature on applications of the laser method, and we only deviate from previous works at the end of this step. Nevertheless, we repeat the description of the hashing procedure for completeness. 

Let $\alpha$ be a distribution on $\{(i, j, k) \in \mathbb{Z}_{\ge 0}^3\mid i + j + k = 2^{\lvl}\}$, and let $\alphx, \alphy, \alphz$ be the marginal distributions of $\alpha$ on the three dimensions respectively. 

First, in $\T$, we zero out the level-$\ell$ $X$-blocks that are not consistent with $\alphx$, the level-$\ell$ $Y$-blocks that are not consistent with $\alphy$, and the level-$\ell$ $Z$-block that are not consistent with $\alphz$. Let $\numxblock, \numyblock, \numzblock$ be the number of remaining $X$, $Y$, $Z$-blocks respectively, $\numtriple$ be the number of remaining block triples, and $\numalpha$ be the number of remaining block triples consistent with $\alpha$. These values can be approximated in the standard way: 

\begin{claim}
\EquationOnSameLine{\numxblock = 2^{H(\alphx) \cdot A_1 n \pm o(n)}, \quad \numyblock = 2^{H(\alphy) \cdot A_1 n \pm o(n)}, \quad \numzblock = 2^{H(\alphz) \cdot A_1 n \pm o(n)}, \phantom{\textbf{Claim X.X.}}}
\[
  \numalpha = 2^{H(\alpha) \cdot A_1 n \pm o(n)}, \quad \numtriple = 2^{(H(\alpha) + P_\alpha) \cdot A_1 n \pm o(n)}. 
\]
\end{claim}

Let
\begin{equation}
\label{eq:global:initial-M0-bound}
    M_0 \ge 8 \cdot \frac{\numtriple}{\numxblock}
\end{equation}
be an integer yet to be fixed, and let $M \in [M_0, 2M_0]$ be a prime number (its existence is guaranteed by Bertrand's postulate). For random $b_0, \{w_t\}_{t=0}^n \in \Z_M$, we define three hash functions $\hashx, \hashy, \hashz : \{0, \ldots, 2^\lvl\}^n \rightarrow \Z_M$ as:
\begin{align*}
    \hashx(I) &= b_0 + \left(\sum_{t=1}^n w_t \cdot I_t \right) \bmod M,\\
    \hashy(J) &= b_0 + \left(w_0 + \sum_{t=1}^n w_t \cdot J_t \right) \bmod M,\\
    \hashz(K) &= b_0 + \frac{1}{2}\left(w_0+\sum_{t=1}^n w_t \cdot (2^\lvl - K_t) \right) \bmod M.
\end{align*}
Let $B$ be the Salem-Spencer subset of $\Z_M$ of size $M^{1-o(1)}$ from \cref{thm:salemspencer} that contains no three-term arithmetic progressions. Then we zero out all the level-$\ell$ $X$-blocks $X_I$ where $\hashx(I) \not \in B$, all the level-$\ell$ $Y$-blocks $Y_I$ where $\hashy(J) \not \in B$, and all the level-$\ell$ $Z$-blocks $X_K$ where $\hashz(K) \not \in B$. 

Note that by definition of the hash functions, we have  $\hashx(I) + \hashy(J) = 2 \hashz(K)$ for any block triple $X_I Y_J Z_K$ (as $\hashx(I), \hashy(J), \hashz(K) \in Z_M$, this equation holds modulo $M$ if we view them as integers in $\Z$). Therefore, since we only keep blocks whose hash values belong to the Salem-Spencer set $B$, all remaining block triples $X_I Y_J Z_K$ must have $\hashx(I) = \hashy(J) = \hashz(K) = b$ for some $b \in B$. Now for every $b \in B$, if there are two block triples $X_I Y_J Z_K$ and $X_I Y_{J'} Z_{K'}$ that are both hashed to $b$ and share the same $X$-block $X_I$, we zero out $X_I$. 

We call the tensor after this zeroing out $\THash$ and note that every level-$\ell$ $X$-block is in a unique level-$\ell$ block triple in $\THash$. We highlight that this is the start of the main deviation from previous works. In previous works \cite{duan2023, VXXZ24}, after this step all level-$\ell$ $X$-blocks and $Y$-blocks are in unique block triples; in earlier works \cite{laser,virgi12,stothers,LeGall32power,AlmanW21}, all level-$\ell$ blocks are in unique block triples.

We summarize the properties of this hashing procedure and $\THash$ in the following lemma.

\begin{lemma}[Properties of more asymmetric hashing]\label{lem:more-asym-hash}

The above described procedure and its output $\THash$ satisfy the following:

\begin{enumerate}
    \item \textup{(Implicit in \cite{cw90}, see also \cite{duan2023})} For any level-$\lvl$ block triple $X_I Y_J Z_K\in \T$ and every bucket $b\in \Z_M$, we have 
    \[\Pr\Bk{\tall \hashx(I) = \hashy(J) = \hashz(K) = b} = \frac{1}{M^2}.\]
    Furthermore, for any $b \in \Z_M$ we have that any two different block triples $X_I Y_J Z_K,  X_I Y_{J'} Z_{K'} \in \T$ that share the same $X$-block satisfy
    \[\Pr\Bk{\tall \hashx(I) = \hashy(J') = \hashz(K') = b \;\middle\vert\; \hashx(I) = \hashy(J) = \hashz(K) = b} = \frac{1}{M}.\]
    The same holds for different blocks that share the same $Y$-block or $Z$-block.

    \item
    \label{item:lem:more-asym-hash:item2}
    \textup{(Similar to {\cite[Claim 5.6]{VXXZ24}})} For every $b \in B$ and every level-$\lvl$ block triple $X_I Y_J Z_K \in \T$ consistent with $\alpha$, we have
    \[\Pr\Bk{X_IY_JZ_K\in \THash \;\middle\vert\; \hashx(I) = \hashy(J) = \hashz(K) = b}\ge \frac{3}{4}.\]

    \item \textup{({\cite[Claim 5.7]{VXXZ24}})} \[\E[\textup{number of level-$\lvl$ triples in $\THash$}] \ge \numalpha \cdot M_0^{-1-o(1)}.\]
\end{enumerate}
    
\end{lemma}

\subsection{\texorpdfstring{\boldmath$Y$}{Y}-Compatibility Zero-Out}
\label{subsec:global:Y-comp-zero-out}
Let 
\[ S^{(I, J, K)}_{i, j, k} \defeq \{t \in [n] \mid I_t = i, J_t = j, K_t = k\}\]
and 
\[S^{(J)}_{*, j, *} \defeq \{t \in [n] \mid J_t = j\}, \quad  S^{(K)}_{*, *, k} \defeq \{t \in [n] \mid K_t = k\}.\]
We will drop the superscripts if the context is clear.\footnote{If we follow our general rules for notations described at the beginning of \cref{sec:global} strictly, $S_{*,j,*}^{(J)}$ and $S_{*,*,k}^{(K)}$ would be denoted as $S_{*,j,*}^{(I, J, K)}$ and $S_{*,*,k}^{(I, J, K)}$, but notice that the extra superscripts can be dropped as they do not affect the values of $S_{*,j,*}^{(I, J, K)}$ and $S_{*,*,k}^{(I, J, K)}$.
}

Recall that in $\THash$, every level-$\ell$ $X$-block is in a unique block triple, so given a level-$\ell$ block $X_I$, we can uniquely determine the level-$\ell$ block triple $X_I Y_J Z_K$ containing it. For any $i, j, k$, and any level-1 block $X_{\hat I} \in X_I$, if $\split(\hat I, S_{i, j, k}) \ne \splres_{\itX, i, j, k}$, we zero out $X_{\hat I}$. 

The overall goal of this step is to zero out some level-1 $Y$-blocks, so that each remaining level-1 $Y$-block belongs to a unique level-$\ell$ block triple as well. 

\subsubsection{\texorpdfstring{\boldmath$Y$}{Y}-Compatibility Zero-Out I}

Given any level-1 $Y$-block $Y_{\hat J}$, if there exists a $j$ such that $\split(\hat J, S_{*, j, *}) \ne \splresavg_{\itY, *,j,*}$, then we zero out $Y_{\hat J}$. Eventually, the goal is to obtain independent copies of the interface tensor with complete split distributions $\left\{\splres_{\itX, i, j, k}, \splres_{\itY, i, j, k}, \splres_{\itZ, i, j, k}\right\}_{i+j+k = 2^\lvl}$. For any $Y_{\hat J} \in Y_J$, if $Y_{\hat J}$ belongs to such an interface tensor, then we must have $\split(\hat J, S_{*, j, *}) = \splresavg_{\itY, *,j,*}$ for every $j$. This is because if there exist $X_{\hat I} \in X_I$ and $ Z_{\hat K} \in Z_K$ such that $X_{\hat I} Y_{\hat J} Z_{\hat K} \in X_I Y_J Z_K$, where $X_{\hat I} Y_{\hat J} Z_{\hat K}$ is consistent with the complete split distributions and $X_I Y_J Z_K$ is consistent with $\alpha$, then one can verify that $\split(\hat J, S_{*, j, *}) = \splresavg_{\itY, *,j,*}$ for every $j$ regardless of what $X_{\hat I}$ and $Z_{\hat K}$ are. So it does not hurt to zero out these $Y_{\hat J}$ blocks. We call the tensor obtained after this zero-out $\TYComp$. 

Next, we define the notion of compatibility for level-1 $Y$-blocks, which is similar to the definition of compatibility in \cite{VXXZ24}.

\begin{definition}[$Y$-Compatibility]
  \label{def:global:Y-compatibility}
Given a level-$\ell$ triple $X_I Y_J Z_K$ and a level-$1$ $Y$-block $Y_{\hat J}\in Y_J$, we say that $Y_{\hat J}$ is compatible with $X_I Y_J Z_K$ if the following hold:
  \begin{enumerate}
  \item 
    \label{item:global:Y-compatibility1} For every $(i, j, k) \in \mathbb{Z}_{\ge 0}^3$ with $ i + j + k = 2^{\lvl}$ and $k = 0$, $\split(\hat J, S_{i,j,k}) = \splres_{\itY,i,j,k}$.
  \item
    \label{item:global:Y-compatibility2}
    For every index $j \in \{0, 1, \ldots, 2^\lvl\}$,  $\split(\hat J, S_{*,j,*}) = \splresavg_{\itY, *,j,*}$.
  \end{enumerate}
\end{definition}

\begin{claim}
  \label{cl:global:Y-compatible}
  In $\TYComp$, for every level-1 block triple $X_{\hat I}Y_{\hat J}Z_{\hat K}$ and the level-$\lvl$ block triple $X_I Y_J Z_K$ that contains it, $Y_{\hat J}$ is compatible with $X_I Y_J Z_K$.
\end{claim}
\begin{proof}
    It is easy to see that \cref{item:global:Y-compatibility2} is satisfied by our procedure since we will zero out all the level-1 $Y_{\hat{J}}$ block if there exists $j\in \{0, \ldots, 2^\ell\}$ such that $\split(\hat J, S_{*, j, *}) \ne \splresavg_{\itY, *,j,*}$. So it suffices to show \cref{item:global:Y-compatibility1}.

    Note that all the remaining level-1 $X$-blocks $X_{\hat{I}}$ satisfy $\split(\hat{I}, S_{i,j,k}) = \splres_{\itX,i,j,k}$ for all $(i,j,k)\in \Z_{\ge 0}^3$ with $i+j+k = 2^\ell$ due to the first zero-out in \cref{subsec:global:Y-comp-zero-out} on level-$1$ $X$-blocks. Now consider $(i,j,k)\in \Z_{\ge 0}^3$ with $k = 0$ and $i+j+k = 2^\ell$. In a remaining level-$1$ triple $X_{\hat I}Y_{\hat J}Z_{\hat K}$, we must have $\split(\hat{I}, S_{i,j,k}) = \splres_{\itX, i,j,k}$. Then since $k = 0$, we have $K_t = 0$ for all $t\in S_{i,j,k}$, i.e., $(\hat{K}_{(t-1)\cdot 2^{\ell-1}+1},\hat{K}_{(t-1)\cdot 2^{\ell-1}+2}, \dots, \hat{K}_{t\cdot 2^{\ell-1}}) = \vec{0}$.
    Since we have for each $\hat{t} \in \{(t \! - \! 1) \cdot 2^{\ell-1} \! + \! 1, \, \dots \, , \, t \cdot 2^{\ell-1}\}$, $\hat{I}_{\hat t} + \hat{J}_{\hat t} + \hat{K}_{\hat t} = 2$, we have for all $t\in S_{i,j,k}$,
    \[\lpr{\hat{J}_{(t-1)\cdot 2^{\ell-1}+1},\hat{J}_{(t-1)\cdot 2^{\ell-1}+2}, \dots, \hat{J}_{t\cdot 2^{\ell-1}}} = \vec{2} - \lpr{\hat{I}_{(t-1)\cdot 2^{\ell-1}+1},\hat{I}_{(t-1)\cdot 2^{\ell-1}+2}, \dots, \hat{I}_{t\cdot 2^{\ell-1}}}.\]
    So for every length-$2^{\ell-1}$ tuple $L\in \{0,1,2\}^{2^{\ell-1}}$, we must have 
    \[\split(\hat{J},S_{i,j,k})(L) = \split(\hat{I}, S_{i,j,k})(\vec{2}-L) = \splres_{\itX,i,j,k}(\vec{2} - L).\]
    By \cref{rmk:assumptions_on_complete_split_dist}, we have that $\split(\hat{J},S_{i,j,k})(L) = \splres_{\itX,i,j,k}(\vec{2} - L) = \splres_{\itY,i,j,k}(L)$ as desired.
\end{proof}

\subsubsection{\texorpdfstring{\boldmath$Y$}{Y}-Compatibility Zero-Out II: Unique Triple}

After the previous step, it is guaranteed that every remaining level-1 $Y$-block $Y_{\hat J}$ is compatible with all the level-$\lvl$ block triples containing it by \cref{cl:global:Y-compatible}. In order to ensure that each remaining $Y_{\hat J}$ belongs to a unique triple after further zero-outs, it suffices to guarantee that $Y_{\hat J}$ belongs to a unique triple that it is compatible with. Therefore, given any $Y_{\hat J}$ that is compatible with more than one triple, we zero it out. 

After this step, each remaining $Y_{\hat J}$ belongs to a unique triple, and additionally $Y_{\hat J}$ is compatible with this triple. 

\subsubsection{\texorpdfstring{\boldmath$Y$}{Y}-Usefulness Zero-Out}

At this point, each remaining level-1 block $Y_{\hat J}$ belongs to a unique triple $X_I Y_J Z_K$. So now we can check whether $\split(\hat J, S_{i, j, k}) = \splres_{\itY, i, j, k}$ for all $i, j, k$, and zero out $Y_{\hat J}$ if there exists some $i,j,k$ such that the equality does not hold. 

For convenience, we define the following notion of usefulness. 

\begin{definition}[$Y$-Usefulness]
    Given a level-$\lvl$ triple $X_I Y_J Z_K$ and a level-1 block $Y_{\hat J}\in Y_J$, we say $Y_{\hat J}$ is \emph{useful} for $X_I Y_J Z_K$ if $\split(\hat J, S_{i, j, k}) = \splres_{\itY, i, j, k}$ for every $i, j, k$. 
\end{definition}

Using the above definition, this step is equivalent to zeroing out every level-1 block $Y_{\hat J}$ that is not useful for the unique triple containing it. We call the remaining tensor after this step $\TYUseful$. 

\subsection{\texorpdfstring{\boldmath$Z$}{Z}-Compatibility Zero-Out}

After $Y$-compatibility zero-outs, every level-$1$ $Y$-block $Y_{\hat{J}}\in Y_J$ must be in a unique level-$\lvl$ triple $X_IY_JZ_K$ and the same holds for every level-$1$ $X$-block $X_{\hat I}$ as well. The goal of this step is to zero out some level-$1$ $Z$-blocks $Z_{\hat{K}}$ so that every level-$1$ $Z$-block is also in a unique level-$\ell$ triple. We note that this step is similar to the compatibility zero-out and usefulness zero-out steps in~\cite{VXXZ24} and \cref{def:global:Z-compatibility,def:global:Z-useful} are identical to the definition of compatibility and usefulness respectively in \cite{VXXZ24}.

\subsubsection{\texorpdfstring{\boldmath$Z$}{Z}-Compatibility Zero-Out I}

Given any level-1 block $Z_{\hat K} \in Z_K$, we zero out $Z_{\hat K}$ if there exists a $k$ such that $\split(\hat K, S_{*,*,k}) \ne \splresavg_{\itZ, *,*,k}$. We call the obtained tensor $\TZComp$.

Next, we are ready to define compatibility for level-$1$ $Z_{\hat K}$ blocks.

\begin{definition}[$Z$-Compatibility]
  \label{def:global:Z-compatibility}
  Given a level-$\ell$ triple $X_I Y_J Z_K$ and a level-$1$ $Z$-block $Z_{\hat K}\in Z_K$, we say that $Z_{\hat K}$ is \emph{compatible} with $X_I Y_J Z_K$ if the following hold:
  \begin{enumerate}
  \item 
    \label{item:global:Z-compatibility1} For every $(i, j, k) \in \mathbb{Z}_{\ge 0}^3$ with $i + j + k = 2^{\lvl}$, $i = 0$ or $j = 0$, there is $\split(\hat K, S_{i,j,k}) = \splres_{\itZ,i,j,k}$.
  \item
    \label{item:global:Z-compatibility2}
    For every index $k \in \{0, 1, \ldots, 2^\lvl\}$, $\split(\hat K, S_{*,*,k}) = \splresavg_{\itZ, *,*,k}$.
  \end{enumerate}
\end{definition}

\begin{claim}[{\cite[Claim 5.9]{VXXZ24}}]
\label{cl:global:Z-compatible}
  In $\TZComp$, for every level-1 block triple $X_{\hat I}Y_{\hat J}Z_{\hat K}$ and the level-$\lvl$ block triple $X_I Y_J Z_K$ that contains it, $Z_{\hat K}$ is compatible with $X_I Y_J Z_K$.
\end{claim}

\subsubsection{\texorpdfstring{\boldmath$Z$}{Z}-Compatibility Zero-Out II: Unique Triple}

Now in $\TZComp$, every level-$1$ $Z$-block $Z_{\hat K}$ is compatible with the level-$\lvl$ triple $X_I Y_J Z_K$ containing it, but there can be multiple level-$\lvl$ triples containing the level-$1$ block $Z_{\hat K}$. So in this step, we zero out $Z_{\hat K}$ if it is compatible with more than one level-$\lvl$ triple. Then the remaining tensor satisfies the property that every level-$1$ $Z_{\hat K}$ is contained in a unique level-$\lvl$ triple, and $Z_{\hat K}$ is compatible with that triple.

\subsubsection{\texorpdfstring{\boldmath$Z$}{Z}-Usefulness Zero-Out}

For convenience, we also define the notion of usefulness for $Z$-blocks. 

\begin{definition}[$Z$-Usefulness]\label{def:global:Z-useful}
Given a level-$\lvl$ triple $X_I Y_J Z_K$ and a level-1 block $Z_{\hat K} \in Z_K$, we say that $Z_{\hat K}$ is \emph{useful} for $X_I Y_J Z_K$ if $\split(\hat K, S_{i,j,k}) = \splres_{\itZ,i,j,k}$ for every $i, j, k$. 
\end{definition}

We zero out every level-1 block $Z_{\hat K}$ that is not useful for the unique triple containing it and call the remaining tensor $\TZUseful$.

\subsection{Fixing Holes}
\label{sec:global:fix-holes}

At this stage, the subtensor of the remaining tensor $\TZUseful$ over $X_I Y_J Z_K$ is a subtensor of the following tensor 
\[ \mathcal{T}^* = \bigotimes_{i+j+k = 2^\lvl} T_{i,j,k}^{\otimes A_1 \cdot \alpha(i, j, k) \cdot n} [\splres_{\itX, i, j, k}, \splres_{\itY, i, j, k}, \splres_{\itZ, i, j, k}], \]
i.e., it is the level-$\lvl$ interface tensor with parameter list 
\[ \left\{ \bk{ A_1 \cdot \alpha(i, j, k) \cdot n, i, j, k, \splres_{\itX, i, j, k}, \splres_{\itY, i, j, k}, \splres_{\itZ, i, j, k} } \right\}_{i + j + k = 2^{\lvl}}. \]
More precisely:

\begin{claim}
  \label{cl:global:after-unique-triple-zero-out}
  For any level-$\lvl$ block triple $X_I Y_J Z_K$ contained in $\THash$, the subtensor of $\TZUseful$ restricted to the level-$\ell$ blocks $X_I, Y_J, Z_K$ is a subtensor of $\T^*$, where the missing variables in this subtensor are exactly those level-1 blocks $Y_{\hat J}$ that are compatible with multiple level-$\lvl$ triples in $\TYComp$ and level-1 blocks $Z_{\hat K}$ that are compatible with multiple level-$\lvl$ triples in $\TZComp$.
\end{claim}

\begin{proof}

We first show that for any level-$\ell$ block triple $X_IY_JZ_K$ in $\THash$, $\TZUseful|_{X_IY_JZ_K}$ is a subtensor of $\T^*$. Note that we already have 
\[\THash \big|_{X_IY_JZ_K} \equiv \bigotimes_{i+j+k = 2^\ell} T_{i,j,k}^{\otimes A_1\cdot \alpha(i,j,k)\cdot n},\]
so it suffices to show that all the level-1 triples $X_{\hat I} Y_{\hat J} Z_{\hat K}$ remaining in $\TZUseful$ satisfy the following for all $i,j,k$:
\begin{enumerate}
    \item $\split(\hat{I}, S_{i,j,k}) = \splres_{\itX, i,j,k}$, \label{item:global:after-unique:X}
    \item $\split(\hat{J}, S_{i,j,k}) = \splres_{\itY, i,j,k}$, \label{item:global:after-unique:Y}
    \item $\split(\hat{K}, S_{i,j,k}) = \splres_{\itZ, i,j,k}$. \label{item:global:after-unique:Z}
\end{enumerate}
Note that \cref{item:global:after-unique:X} is automatically satisfied due to the step in $Y$-compatibility zero-out, since we have zeroed out any level-1 $X_{\hat I}\in X_{I}$ with $i,j,k$ such that $\split(\hat{I}, S_{i,j,k}) \ne \splres_{\itX,i,j,k}$. \cref{item:global:after-unique:Y} is satisfied due to $Y$-usefulness zero-out and \cref{item:global:after-unique:Z} is satisfied due to $Z$-usefulness zero-out. Thus, $\TZUseful|_{X_IY_JZ_K}$ is a subtensor of $T^*$.

Now we analyze what the missing $Y$-variables are. Note that for level-1 $Y_{\hat J}$ blocks, we have enforced the following conditions:
\begin{enumerate}[label=(\arabic*)]
\item  In $Y$-compatibility zero-out I, we enforce that $\split(\hat{J}, S_{*,j,*}) = \splresavg_{\itY,*,j,*}$.
\item  In $Y$-compatibility zero-out II, we enforce that each $\hat{J}$ is compatible with a unique level-$\ell$ triple.
\item In $Y$-usefulness zero-out, we enforce that $\split(\hat{J}, S_{i,j,k}) = \splres_{i,j,k}$ for every $i,j,k$.
\end{enumerate}
We claim that condition (3) is strictly stronger than condition (1), because by definition, (3) implies
\begin{align*}
    \Split(\hat{J}, S_{*,j,*}) 
    &= \frac{1}{\sum_{i,k}\alpha(i,j,k)}\sum_{i+k = 2^\ell - j} \alpha(i,j,k)\cdot \Split(\hat{J}, S_{i,j,k})\\
    &= \frac{1}{\sum_{i,k}\alpha(i,j,k)}\sum_{i+k = 2^\ell - j} \alpha(i,j,k)\cdot \splres_{\itY, i,j,k}\\
    &= \splresavg_{\itY,*,j,*}.
\end{align*}
Also, enforcing condition (3) does not create missing $Y$-variables, as it is enforcing the necessary complete split distribution condition in the definition of $\T^*$. Thus, the missing variables are exactly due to enforcing condition (2), i.e., the level-1 $Y_{\hat J}$ blocks that are compatible with more than one level-$\ell$ triple. Similarly, the missing $Z$ variables are exactly those level-1 $Z_{\hat K}$ blocks that are compatible with more than one level-$\ell$ triple.
\end{proof}

Furthermore, observe that $\TZUseful\vert_{X_I Y_J Z_K}$ is level-1-independent for different block triples $X_I Y_J Z_K$: The $X$-blocks are level-$\lvl$-independent after the zero-out in the asymmetric hashing step; for $Y$ and $Z$-blocks, by \cref{cl:global:Y-compatible,cl:global:Z-compatible}, each level-1 block is compatible with all the level-$\lvl$ triples containing it, and we zeroed out all level-1 $Y$ and $Z$-blocks that are compatible with multiple level-$\lvl$ triples, so the remaining level-1 $Y$ and $Z$-blocks belong to unique level-$\lvl$ triples. As a result, 
\[
  \TZUseful \, = \bigoplus_{X_I Y_J Z_K \; \textup{remaining}} \TZUseful \big\vert_{X_I Y_J Z_K}
  \numberthis \label{eq:T'''_as_direct_sum}
\]
is a direct sum of broken copies of $\T^*$.

To fix the holes, we need to first bound the fraction of holes in the broken copies of $\T^*$ contained in $\TZUseful$. Similar to the analysis in \cite{VXXZ24}, we introduce the notion of \emph{typicalness} for level-$1$ $Y$-blocks and $Z$-blocks and define the values $\pcompY$ and $\pcompZ$ respectively. Previously, the notion of typicalness and $\pcomp$ were only defined with respect to level-$1$ $Z_{\hat K}$ blocks. This is because previously only level-1 $Z_{\hat K}$ blocks could become holes in the remaining tensor. However, in our case both level-$1$ $Y_{\hat{J}}$ blocks and $Z_{\hat K}$ blocks can become holes, so we need to define similar notions for $Y_{\hat J}$ blocks accordingly. 

\begin{definition}[$Y$-Typicalness]
  A level-1 $Y$-block $Y_{\hat{J}}$ in a level-$\lvl$ $Y$-block $Y_J$ is \emph{typical} if $\split(\hat J, S_{*,j,*}) = \splresavg_{\itY,*,j,*}$ for every $j$. When $Y_J$ is consistent with $\alphy$, this condition is equivalent to $\split(\hat J, [A_1 n]) = \splavg{\itY}$.
\end{definition}

\begin{definition}[$Z$-Typicalness]
  A level-1 $Z$-block $Z_{\hat{K}}$ in a level-$\lvl$ $Z$-block $Z_K$ is \emph{typical} if $\split(\hat K, S_{*,*,k}) = \splresavg_{\itZ,*,*,k}$ for every $k$. When $Z_K$ is consistent with $\alphz$, this condition is equivalent to $\split(\hat K, [A_1 n]) = \splavg{\itZ}$.
\end{definition}

Then we can define the values $\pcompY$ and $\pcompZ$ as the probability of a typical level-$1$ $Y_{\hat J}$-block (resp.\ $Z_{\hat K}$-block) being compatible with a random level-$\ell$ triple $X_IY_JZ_K$ where $\hat J\in J$ (resp.\ $\hat K\in K$).

\begin{definition}[$\pcompY$]
For a fixed $Y_J$ and a fixed typical $Y_{\hat{J}} \in Y_J$, $\pcompY$ is the probability of $Y_{\hat J}$ being compatible with a uniformly random block triple $X_I Y_J Z_K$ consistent with $\alpha$.
\end{definition}

\begin{definition}[$\pcompZ$]
For a fixed $Z_K$ and a fixed typical $Z_{\hat{K}} \in Z_K$, $\pcompZ$ is the probability of $Z_{\hat K}$ being compatible with a uniformly random block triple $X_I Y_J Z_K$ consistent with $\alpha$.
\end{definition}

For any two different typical level-$1$ blocks $Y_{\hat J}\in Y_J$ and $Y_{\hat J '}\in Y_{J'}$, their level-$1$ complete split distributions are the same. It is not difficult to see that $\pcompY$ defined using $Y_{\hat J}$ is the same as $\pcompY$ defined using $Y_{\hat J'}$. This implies that the definition of $\pcompY$ is not dependent on the choice of the level-$1$ block $Y_{\hat J}$, so it is well-defined. The same holds for $\pcompZ$ due to the same reasoning.

\begin{claim}
The value of $\pcompY$ is 
\[
2^{(\eta_\itY - H(\splresavg_{\itY, *, *, *}) + H(\alpha_\itY)) \cdot A_1 n \pm o(n) },
\]
where we recall 
\[\eta_\itY = \sum_{i,j,k \,:\, k = 0} \alpha(i,j,k) \cdot H(\splres_{\itY, i,j,k}) + \sum_{j}\alpha(*,j,\+)\cdot H(\splresavg_{\itY,*,j,\+}).\]
\end{claim}

\begin{proof}
We define the following two quantities:
\begin{enumerate}

    \item[($P$)] The number of tuples $(I, J, K, \hat{J})$ where $X_I Y_J Z_K$ is consistent with $\alpha$, $Y_{\hat J} \in Y_J$, and $Y_{\hat J}$ is typical.

    \item[($Q$)] The number of tuples $(I, J, K, \hat{J})$ where $X_I Y_J Z_K$ is consistent with $\alpha$, $Y_{\hat J} \in Y_J$, and $Y_{\hat J}$ is typical, and additionally $Y_{\hat J}$ is compatible with the triple $X_I Y_J Z_K$.

\end{enumerate}
Notice that by definition and by symmetry of different choices of $Y_{\hat J}$ and $Y_J$, $\pcompY = Q/P$.

The quantity $P$ is simple to calculate, as it is the number of block triples $X_I Y_J Z_K$ consistent with $\alpha$ (this quantity is $2^{H(\alpha) \cdot A_1 n \pm o(n)}$) times the number of typical $Y_{\hat J}$ contained in every $Y_J$. The total number of typical $Y_{\hat J}$ that is contained in some $Y_J$ consistent with $\alpha_\itY$ is $2^{H(\splresavg_{\itY, *,*,*}) \cdot A_1 n \pm o(n)}$. These $Y_{\hat J}$ are evenly distributed among all $Y_J$ consistent with $\alpha_\itY$, so each $Y_J$ contains $2^{(H(\splresavg_{\itY, *,*,*}) - H(\alpha_\itY)) \cdot A_1 n \pm o(n)}$ typical $Y_{\hat{J}}$. Overall, 
$$P = 2^{(H(\alpha) + H(\splresavg_{\itY, *,*,*}) - H(\alpha_\itY)) \cdot A_1 n \pm o(n)}. $$

Next, we consider how to compute $Q$. Similar to the alternative definition of compatibility in the proof of Claim 5.14 in \cite{VXXZ24}, we have the following claim:
\begin{claim}
    Given a level-$\ell$ triple $X_IY_JZ_K$ that is consistent with $\alpha$, a level-$1$ block $Y_{\hat J}$ with $\hat J\in J$ is compatible with $X_IY_JZ_K$ if and only if the following holds:

    \begin{enumerate}
        \item For every $(i,j,k)\in \Z_{\ge 0}^3$ satisfying $i+j+k = 2^\ell$ and $k = 0$, $\split(\hat{J}, S_{i,j,k}) = \splres_{\itY,i,j,k}$.

        \item For $j\in \{0,\dots, 2^\ell\}$, $\split(\hat J, S_{*,j,\+}) = \splresavg_{\itY,*,j,\+}$ where $S_{*,j,\+} = \bigcup_{i \ge 0, k>0} S_{i,j,k}$.
    \end{enumerate}
\end{claim}
We omit the proof, since it is similar to the proof of the equivalence between two definitions of compatibility in the proof of Claim 5.14 in \cite{VXXZ24}.

Fix some $X_I Y_J Z_K$ that is consistent with $\alpha$, $\{S_{i, j, 0}\}_{i, j} \cup \{S_{*, j, \+}\}_{j}$ is a partition of $[A_1 n]$, and we can consider the possibilities of $\hat{J}$ on each of these parts. For $S_{i, j, 0}$, the number of possibilities of $\hat{J}$ is $2^{H(\splres_{\itY, i, j, 0}) \cdot \alpha(i, j, 0) \cdot A_1 n \pm o(n)}$, as we must have $\split(\hat{J}, S_{i,j,0}) = \splres_{\itY,i,j,0}$. For $S_{*, j, \+}$, the number of possibilities of $\hat{J}$ is $2^{H(\splresavg_{\itY, *, j, \+}) \cdot \alpha(*, j, \+) \cdot A_1 n \pm o(n)}$, as we must have $\split(\hat J, S_{*,j,\+}) = \splresavg_{\itY,*,j,\+}$. Clearly, by combining each set of possibilities from each part, the resulting $\hat{J}$ is typical. The total number of possible $\hat{J}$ for a fixed triple $X_I Y_J Z_K$ is thus
\[
2^{(\sum_{i + j = 2^\l} H(\splres_{\itY, i, j, 0}) \cdot \alpha(i, j, 0) + \sum_j H(\splresavg_{\itY, *, j, \+}) \cdot \alpha(*, j, \+) ) A_1 n \pm o(n)} = 2^{\eta_\itY A_1 n \pm o(n)}. 
\]
Therefore, 
\[
Q = 2^{\left(H(\alpha) + \eta_\itY\right) A_1 n\pm o(n)}.
\]
Finally, 
\[
\pcompY = Q / P = 2^{(\eta_\itY - H(\splresavg_{\itY, *, *, *}) + H(\alpha_\itY)) \cdot A_1 n \pm o(n) }.
\qedhere
\]
\end{proof}

\begin{claim}[{\cite[Claim 5.14]{VXXZ24}}]
  The value of $\pcompZ$ is
  $$2^{\left(\lambda_\itZ - H(\splavg{\itZ}) + H(\alphz)\right) A_1 \cdot n \pm o(n)}, $$
  where we recall that 
  \[
    \lambda_\itZ = \sum_{i, j, k \,:\, i = 0 \textup{ or } j = 0} \alpha(i, j, k) \cdot H(\splres_{\itZ, i, j, k}) + \sum_k \alpha(\+, \+, k) \cdot H(\splresavg_{\itZ, \+, \+, k}).
  \]
\end{claim}

\begin{claim}[Essentially {\cite[Claim 5.16]{VXXZ24}}]
  \label{cl:global:prob-of-holes}
  For every $b \in B$, every level-$\lvl$ block triple $X_I Y_J Z_K$ consistent with $\alpha$, and every typical $Z_{\hat{K}} \in Z_K$, the probability that $Z_{\hat{K}}$ is compatible with multiple triples in $\TZComp$ is at most
  \[ \frac{\numalpha \cdot \pcompZ}{\numzblock \cdot M_0}, \]
  conditioned on $\hashx(I) = \hashy(J) = \hashz(K) = b$.

  Similarly,  for every $b \in B$, every level-$\lvl$ block triple $X_I Y_J Z_K$ consistent with $\alpha$, and every typical $Y_{\hat{J}} \in Y_J$, the probability that $Y_{\hat{J}}$ is compatible with multiple triples in $\TYComp$ is at most
  \[ \frac{\numalpha \cdot \pcompY}{\numyblock \cdot M_0}, \]
  conditioned on $\hashx(I) = \hashy(J) = \hashz(K) = b$.
\end{claim}

Initially, in \cref{eq:global:initial-M0-bound} we required $M_0 \ge 8\cdot \frac{\numtriple}{\numxblock}$. Now we finalize all constraints on $M_0$, and set 
\begin{align} \label{eq:hashmodulus}
M_0 & = \max\left\{8\cdot \frac{\numtriple}{\numxblock}, \;
\frac{\numalpha \cdot \pcompY}{\numyblock} \cdot 80 N, \;
\frac{\numalpha \cdot \pcompZ}{\numzblock} \cdot 80 N\right\} \\ 
& = 2^{\max\{H(\alpha)+P_\alpha - H(\alpha_\itX), \; \eta_\itY - H(\splresavg_{\itY, *, *, *}) + H(\alpha), \; \lambda_\itZ - H(\splresavg_{\itZ, *, *, *}) + H(\alpha) \} \cdot A_1 n \pm o(n)}.
\end{align}
For every $b \in B$ and every level-$\lvl$ block triple $X_I Y_J Z_K$ consistent with $\alpha$ that is hashed to bucket $b$ in asymmetric hashing, we consider the probability that it remains in $\TZUseful$ and $\TZUseful \vert_{X_I Y_J Z_K}$ is a copy of $\T^*$ with a small number of holes. 

First, by \cref{item:lem:more-asym-hash:item2} in \cref{lem:more-asym-hash}, the level-$\lvl$ block $X_I Y_J Z_K$ remains with probability at least $\frac{3}{4}$. By \cref{cl:global:prob-of-holes}, the expected fractions of holes for level-1 $Z$-blocks is at most $\frac{\numalpha \cdot \pcompZ}{\numzblock \cdot M_0} \le \frac{1}{80 N}$. By Markov's inequality, this fraction exceeds $\frac{1}{8N}$ with probability $\le 1/10$. The same applies to the fraction of holes of $Y$-variables. Therefore, by union bound, with constant probability, the level-$\lvl$ block $X_I Y_J Z_K$ remains and $\TZUseful \vert_{X_I Y_J Z_K}$ is a broken copy of $\T^*$ whose fraction of holes is $\le \frac{1}{8N}$ in all three dimensions. The expected number of such $X_I Y_J Z_K$ is $\numalpha \cdot M^{-1 - o(1)}$, so in expectation, we obtain $\numalpha \cdot M^{-1 - o(1)}$ independent broken copies of $\T^*$ with $\le \frac{1}{8N}$ fraction of holes, and by \cref{thm:fix-holes}, we can degenerate them into $\numalpha \cdot M^{-1 - o(1)}$ unbroken copies of $\T^*$. 

\subsection{Summary}

In conclusion, the above algorithm degenerates 
$\bigbk{\CW_q^{\otimes 2^{\lvl - 1}}}^{\otimes A_1 \cdot n}$ into 
\[\numalpha \cdot M_0^{-1-o(1)} \ge 2^{A_1 n \cdot \min\BK{H(\alphx^{(1)}) - P_\alpha^{(1)}, \; H(\splavg{\itY}^{(1)}) - \eta_\itY^{(1)}, \;  H(\splavg{\itZ}^{(1)}) - \lambda_\itZ^{(1)}} - o(n)}\]
independent copies of a level-$\lvl$ interface tensor $\T^*$ with parameter list
\[ \left\{ \bk{ n \cdot A_1 \cdot \alpha^{(1)}(i, j, k), \, i, j, k, \, \splres^{(1)}_{\itX, i, j, k}, \splres^{(1)}_{\itY, i, j, k}, \splres^{(1)}_{\itZ, i, j, k} } \right\}_{i + j + k = 2^{\lvl}}. \]

In region $r\in [6]$, let $\pi_r: \{X,Y,Z\}\to \{X,Y,Z\}$ be the $r$-th permutation in the lexicographic order and we perform the same procedure with $\pi_r(X)$-blocks in place of  $X$-blocks, $\pi_r(Y)$-blocks in place of $Y$-blocks, and $\pi_r(Z)$-blocks in place of $Z$-blocks. Note that we have described the procedure in the first region which corresponds to the identity permutation. In the end, we take the tensor product over the output tensor of the algorithm over all $6$ regions.

\section{Constituent Stage}
\label{sec:constituent-tensor}
\label{sec:constituent}

The framework for our constituent stage algorithm is almost identical to the framework for our global stage algorithm, but for the sake of rigor and notations, we present these two stages separately. 

In the level-$\lvl$ constituent stage, the algorithm takes as input a collection of level-$\lvl$ $\eps$-interface tensors, and degenerates them into the tensor product between a matrix multiplication tensor (obtained from the parts of the interface tensor where $i_t = 0$, $j_t = 0$, or $k_t = 0$) and a collection of independent copies of level-$(\lvl-1)$ $\eps'$-interface tensors for some $\eps' > 0$ (obtained from the other parts). 

More specifically, given an $\eps$-interface tensor with parameters
\[\{(n_t, i_t, j_t, k_t, \splresXt, \splresYt, \splresZt)\}_{t \in [s]},\]
we use the notation $n := \sum_t n_t$ and $N := 2^{\lvl-1} \cdot n$. 

\begin{remark}
\label{rmk:constituent:assumptions_on_complete_split_dist}
We can assume without loss of generality that the parameters additionally satisfy the following:
\begin{itemize}
    \item For every $t \in [s]$ with $j_t = 0$, and every $L \in \{0, 1, 2\}^{2^{\lvl-1}}$, 
    \[
    \splresXt(L) = \splresZt(\vec{2} - L).
    \]
  The same holds between $\splresXt$ and $\splresYt$ when $k_t = 0$, and between $\splresYt$ and $\splresZt$ when $i_t = 0$.
  \item For every $t \in [s]$ and every $L \in \{0, 1, 2\}^{2^{\lvl-1}}$, $\splresXt(L) = 0$ if $\sum_q L_q \ne i_t$. The same holds for $\splresYt, \splresZt$ with respect to $j_t, k_t$ correspondingly. 
\end{itemize}
\end{remark}

To obtain matrix multiplication tensors from the parts of the interface tensors where $i_t = 0$ or $j_t = 0$ or $k_t = 0$, we invoke the following theorem. 

\begin{theorem}[\cite{VXXZ24}]
  \label{thm:consituent_MM_terms}
  If $k_t=0$, then
  \[ T_{i_t,j_t,k_t}^{\otimes n_t}[\splresXt, \splresYt, \splresZt, \eps] \equiv \angbk{1, M, 1}, \]
  where
  \[ M = 2^{n_t (H(\splresXt) \pm o_{1/\eps}(1)) \pm o(n)} \cdot q^{n_t\sum_{(\hat{i}_1, \hat{i}_2, \ldots, \hat{i}_{2^{\lvl-1}})} \splresXt(\hat{i}_1, \hat{i}_2, \ldots, \hat{i}_{2^{\lvl-1}}) \sum_{p=1}^{2^{\lvl-1}} \ind[\hat{i}_p = 1]}. \]
  Similar results hold for the case when $i_t = 0$ or $j_t = 0$.
\end{theorem}

We use \cref{thm:consituent_MM_terms} to degenerate the terms $t \in [s]$ where $i_t = 0$ or $j_t = 0$ or $k_t = 0$ into matrix multiplication tensors, and we can take their product to obtain a single matrix multiplication tensor. So now we are left with the remaining terms $t\in [s]$ where $i_t\ne 0, j_t\ne 0, k_t\ne 0$. Without loss of generality, we can reorder the terms so that the remaining terms are the first $s' \le s$ terms. 

For every $t \in [s']$, we have a triple of level-$\lvl$ complete split distributions $(\splresXt, \splresYt, \splresZt)$ associated with the $n_t$-th tensor power of the constituent tensor $T_{i_t,j_t,k_t}$. We define a distribution $\splonelevelXt$ on $\{0, \ldots, 2^{\lvl-1}\}^2$ as follows: for $l_\itX, r_\itX \in \{0, \ldots, 2^{\lvl-1}\}$,
\[
  \splonelevelXt(l_\itX, r_\itX) \; \defeq \sum_{\substack{(\hat{i}_1, \hat{i}_2, \ldots, \hat{i}_{2^{\lvl - 1}})\,: \\  \hat{i}_1 + \cdots + \hat{i}_{2^{\lvl - 2}} = l_\itX, \\ \hat{i}_{2^{\lvl-2}+1} + \cdots + \hat{i}_{2^{\lvl - 1}} = r_\itX}} \splresXt(\hat{i}_1, \hat{i}_2, \ldots, \hat{i}_{2^{\lvl-1}}).
\]
The distribution $\splonelevelXt$ specifies how a level-$\lvl$ index sequence $i_t$ splits into two level-$(\lvl - 1)$ index sequences. We define $\splonelevelYt$ and $\splonelevelZt$ similarly.

Let $\alpha_t$ be a distribution on possible combinations of $(l_\itX, l_\itY, l_\itZ)\in \{0, \ldots, 2^{\lvl-1}\}^3$ such that the marginals of $\alpha_t$ are consistent with $\splonelevelXt(l_\itX, i-l_\itX)$, $\splonelevelYt(l_\itY, j-l_\itY)$, $\splonelevelZt(l_\itZ, k-l_\itZ)$. Moreover, let $\splresXt[t, i', j', k']$, $\splresYt[t, i', j', k']$, $\splresZt[t, i', j', k']$ be level-$(\lvl-1)$ complete split distributions. Given $t\in [s']$, we define the following quantities in \cref{tab:notation-constituent}.

\begin{table}[ht]
    \centering
    {\def\arraystretch{1.5}
    \begin{tabular}{|| c  p{12cm} ||}
    \hline
       Notation & Definition \\
       \hline\hline

       \multirow{2.5}*{$P_{\alpha, t}$} & The penalty term $P_{\alpha, t} \defeq \max_{\alpha_t' \in D} H(\alpha_t') - H(\alpha_t) \ge 0$ where $D$ is the set of distributions whose marginal distributions on the three dimensions are consistent with $\splonelevelXt(l_\itX, i-l_\itX)$, $\splonelevelYt(l_\itY, j-l_\itY)$, and $\splonelevelZt(l_\itZ, k-l_\itZ)$, respectively. \\ 
        
        \hline

        $\alpha_t(i', \+, \+)$ & $\alpha_t(i', \+, \+) \defeq \sum_{j' > 0, k' > 0} \alpha_t(i', j', k')$ \\ 

        $\alpha_t(\+, j', \+)$ & $\alpha_t(\+, j', \+) \defeq \sum_{i' > 0, k' > 0} \alpha_t(i', j', k')$ \\

        $\alpha_t(\+, \+, k')$ & $\alpha_t(\+, \+, k') \defeq \sum_{i' > 0, j' > 0} \alpha_t(i', j', k')$ \\ 
        
        \hline

        $\alpha_t(i', \<, \<)$ & $\alpha_t(i', \<, \<) \defeq \sum_{j' < j_t, k' < k_t} \alpha_t(i', j', k')$ \\

        $\alpha_t(\<, j', \<)$ & $\alpha_t(\<, j', \<) \defeq \sum_{i' < i_t, k' < k_t} \alpha_t(i', j', k')$ \\

        $\alpha_t(\<, \<, k')$ & $\alpha_t(\<, \<, k') \defeq \sum_{i' < i_t , j' < j_t} \alpha_t(i', j', k')$ \\

        \hline

        $\splresavg_{\itX, t, i', \+, \+}$ & $\splresavg_{\itX, t, i', \+, \+} \defeq \frac{1}{\alpha_t(i',\+, \+) + \alpha_t(i_t \!-\! i', \<, \<)} \sum_{j' > 0, k' > 0} \bigbk{\alpha_t(i', j', k') + \alpha_t(i_t \!-\! i', j_t \!-\! j', k_t \!-\! k')} \cdot \splres_{\itX, t, i', j', k'}$ \\
        
        $\splresavg_{\itY, t, \+, j', \+}$ & $\splresavg_{\itY, t, \+, j', \+} \defeq \frac{1}{\alpha_t(\+, j', \+) + \alpha_t(\<, j_t \!-\! j', \<)} \sum_{i' > 0, k' > 0} \bigbk{\alpha_t(i', j', k') + \alpha_t(i_t \!-\! i', j_t \!-\! j', k_t \!-\! k')} \cdot \splres_{\itY, t, i', j', k'}$ \\

        $\splresavg_{\itZ, t, \+, \+, k'}$ & $\splresavg_{\itZ, t, \+, \+, k'} \defeq \frac{1}{\alpha_t(\+, \+, k') + \alpha_t(\<, \<, k_t \!-\! k')} \sum_{i' > 0, j' > 0} \bigbk{\alpha_t(i', j', k') + \alpha_t(i_t \!-\! i', j_t \!-\! j', k_t \!-\! k')} \cdot \splres_{\itZ, t, i', j', k'}$ \\

        \hline

        $\eta_{\itY, t}$ & $\eta_{\itY, t} := \sum_{i', j'}  \bigbk{\alpha_t(i', j', 0) + \alpha_t(i_t \!-\! i', j_t \!-\! j', k_t \!-\! k')} \cdot H(\splres_{\itY, t, i', j', 0}) + 
        \sum_{j'} \bigbk{\alpha_t(*, j', \+) + \alpha_t(*, j_t \!-\! j', \<)} \cdot H(\splresavg_{\itY, t, *, j', \+})$ \\

        $\lambda_{\itZ, t}$ & 
        $\lambda_{\itZ,t} \defeq \sum_{i', j', k' \,:\, i' = 0 \textup{ or } j' = 0} \bigbk{\alpha_t(i', j', k') + \alpha_t(i_t \!-\! i', j_t \!-\! j', k_t \!-\! k')} \cdot H(\splres_{\itZ, t, i', j', k'}) + \sum_{k'} \bigbk{\alpha_t(\+, \+, k') + \alpha_t(\<, \<, k_t \!-\! k')} \cdot H(\splresavg_{\itZ, t, \+, \+, k'})$\\

        \hline
        
    \end{tabular}
    }
    \caption{Table of notations with respect to distributions $\{\alpha_t\}_{t \in [s']}$ over all possible combinations of $(l_\itX, l_\itY, l_\itZ)$ such that the marginals of $\alpha_t$ are consistent with $\splonelevelXt(l_\itX, i-l_\itX)$, $\splonelevelYt(l_\itY, j-l_\itY)$, $\splonelevelZt(l_\itZ, k-l_\itZ)$, and level-$(\lvl-1)$ complete split distributions $\splresXt[t, i', j', k']$, $\splresYt[t, i', j', k']$, $\splresZt[t, i', j', k']$.}
    \label{tab:notation-constituent}
\end{table}

Our notations satisfy the following general rules:

\begin{enumerate}
    \item Given $t\in [s']$, when we refer to values in the the distribution $\alpha_t(i',j',k')$, we may replace any of the input by the symbol $*$, $\+$, or $\<$. If an input coordinate is a ``$*$'', then it means the sum over $\alpha$ evaluated at all the inputs in this input coordinate; if an input coordinate is a ``$\+$'', then it means the sum over $\alpha$ evaluated at all the inputs that are $> 0$ in this input coordinate; if an input coordinate is a ``$\<$'', then it means the sum over $\alpha$ evaluated at all the inputs that are $<$ the $t$-th coordinate of the corresponding index sequence $i,j,k$ in this input coordinate. See examples in \cref{tab:notation-constituent}.

    \item Similar to before, for a given $t \in [s']$, we use $\splresavg_{\itX, t, i', j', k'}$ with $i',j',k'$ replaced by either $*$ or $\+$ to denote the weighted average of the corresponding values, in which each $\splres_{\itX, t, i', j', k'}$ is weighted by $\alpha_t(i', j', k') + \alpha_t(i_t \!-\! i', j_t \!-\! j', k_t \!-\! k')$ (normalized by the total weight). When an input coordinate is replaced by ``$*$'', it means the weighted average is taken over all inputs in this input coordinate; when an input coordinate is replaced by ``$\+$'', it means the weighted average is taken over all inputs that are $> 0$ in this input coordinate. See examples in \cref{tab:notation-constituent}.

    \item For some $t \in [s']$ and a given family of sets $S_{t, i, j, k}$, we may replace any of the subscripts $i, j, k$ by the symbol $*$ or $\+$. If any coordinate of the subscript is a ``$*$'', then it represents the union over $S_{t, i, j, k}$ with subscript $\ge 0$ on this coordinate; if any coordinate of the subscript is a ``$\+$'', then it represents the union over $S_{t, i,j,k}$ with subscript $> 0$ on this coordinate. For example, $S_{t, *, j, k} = \bigcup_{i\ge 0} S_{t, i,j,k}$, and $S_{t, *,\+,k}= \bigcup_{i\ge 0, j>0} S_{t, i,j,k}$.
\end{enumerate}

In the following proposition, we will use the above definitions for different $t \in [s']$ and $r \in [6]$. We will use $t$ in the subscripts and $(r)$ in the superscripts on variables to specify that they are computed from values of $\alpha_t^{(r)}$, $\splres_{\itX, t}^{(r)}$, $\splres_{\itY, t}^{(r)}$, $\splres_{\itZ, t}^{(r)}$, $\bigBK{\splresXt[t, i', j', k']^{(r)}}_{i',j',k'}$, $\bigBK{\splresYt[t, i', j', k']^{(r)}}_{i',j',k'}$, $\bigBK{\splresZt[t, i', j', k']^{(r)}}_{i',j',k'}$. 

\begin{prop}
  \label{prop:constituent-stage-no-eps}
  For $\eps > 0$, an $s'$-term level-$\lvl$ $\eps$-interface tensor with parameters
  \[\{(n_t, i_t, j_t, k_t, \splresXt, \splresYt, \splresZt)\}_{t \in [s']}\]
  for $i_t, j_t, k_t > 0 \; \forall \; t \in [s']$ can be degenerated into
  \[ 2^{(\sum_{r=1}^6 E_r) - o(n) - \oeps(n)} \]
  independent copies of a level-$(\lvl-1)$ interface tensor with parameter list
  \[ \left\{ \bk{ n_t \cdot A_{t,r} \cdot \bigbk{\alpha_{t}^{(r)}(i', j', k')+\alpha_{t}^{(r)}(i_t \!-\! i', j_t \!-\! j', k_t \!-\! k')}, \, i', j', k', \, \splresXt[t, i', j', k']^{(r)}, \splresYt[t, i', j', k']^{(r)}, \splresZt[t, i', j', k']^{(r)} }\right\}\]
  with $t \in [s']$, $r \in [6]$, $i' + j' + k' = 2^{\lvl-1}$, $0 \le i' \le i_t$, $0 \le j' \le j_t$, $0 \le k' \le k_t$, such that
  \begin{itemize}
  \item $0 \le A_{t,r} \le 1$ for every $t \in [s']$, $r \in [6]$, and $\sum_{r=1}^6 A_{t, r}=1$ for every $t \in [s']$;
  \item For every $t$, and for every $W \in \{X, Y, Z\}$, $\sum_{r=1}^6 A_{t,r} \splres_{\itW, t}^{(r)} = \splres_{\itW, t}$ (the $\splres_{\itW, t}^{(r)}$'s are intermediate variables that will be used later); 
  \item For every $W \in \{X, Y, Z\}$, $r \in [6]$, and $i' + j' + k' = 2^{\lvl - 1}$, $\splres_{\itW, t, i', j', k'}^{(r)}$ is a level-$(\lvl-1)$ complete split distribution;
  \item For every $W \in \{X, Y, Z\}$, $t \in [s']$ and $r \in [6]$,
    \[ \splres_{\itW, t}^{(r)} = \sum_{i', j', k'} \alpha_{t}^{(r)}(i', j', k') \cdot \left(\splres_{\itW, t, i', j', k'}^{(r)} \times \splres_{\itW, t,  i_t - i', j_t - j', k_t - k'}^{(r)}\right); \]
  \item For each $r \in [6]$, define $\pi_r: \{X, Y, Z\} \to \{X, Y, Z\}$ as the $r$-th permutation in the lexicographic order. Then
  \begin{align*}
      E_r \defeq \min\Bigg\{
      & \sum_{t \in [s']} A_{t,r} \cdot n_t \cdot \left( H(\gamma_{\pi_r(X), t}^{(r)}) - P_{\alpha, t}^{(r)}\right), \\
      & \sum_{t \in [s']} A_{t,r} \cdot n_t \cdot \left( H(\splres_{\pi_r(Y), t}^{(r)}) - \eta_{\pi_r(Y), t}^{(r)}\right), \\
    & \sum_{t \in [s']} A_{t,r} \cdot n_t \cdot \left(H(\splres_{\pi_r(Z), t}^{(r)}) - \lambda_{\pi_r(Z), t}^{(r)}\right) \Bigg\}. \\
  \end{align*}
  \end{itemize}
\end{prop}

\begin{theorem}
  \label{thm:constituent-stage-with-eps}
  For $\eps > 0$, we can degenerate $2^{o(n)}$ independent copies of $s'$-term level-$\lvl$ $3\eps$-interface tensor with parameters
  \[\{(n_t, i_t, j_t, k_t, \splresXt, \splresYt, \splresZt)\}_{t \in [s']}\]
  where $i_t, j_t, k_t > 0\ \forall\ t \in [s']$ into
  \[2^{(\sum_{r=1}^6 E_r) - o(n) - \oeps(n)}\]
  independent copies of a level-$(\lvl-1)$ $\eps$-interface tensor with parameter list
  \[ \left\{ \bk{ n_t \cdot A_{t,r} \cdot \bigbk{\alpha_{t}^{(r)}(i', j', k')+\alpha_{t}^{(r)}(i_t-i', j_t-j', k_t-k')}, i', j', k', \splresXt[t, i', j', k']^{(r)}, \splresYt[t, i', j', k']^{(r)}, \splresZt[t, i', j', k']^{(r)} } \right\}\]
  for $t \in [s']$, $r \in [6]$, $i' + j' + k' = 2^{\lvl-1}$, $0 \le i' \le i_t$, $0 \le j' \le j_t$, $0 \le k' \le k_t$ satisfying the same properties as in \cref{prop:constituent-stage-no-eps}.
\end{theorem}

The proof of \cref{thm:constituent-stage-with-eps} assuming \cref{prop:constituent-stage-no-eps} is the same as the proof of \cite[Theorem 6.3]{VXXZ24}, so we omit the proof here. We prove \cref{prop:constituent-stage-no-eps} in the remainder of this section. 

\subsection{Dividing into Regions}

For every $t \in [s']$, we divide the $t$-th term to $6$ regions. More specifically, we pick $A_{t, r} \ge 0$ for $r \in [6]$ such that $\sum_r A_{t, r} = 1$, where $A_{t,r}$ denotes the proportion of the $r$-th region inside the $t$-th term. For each region $r$, we pick complete split distributions $\splresXt^{(r)}, \splresYt^{(r)}, \splresZt^{(r)}$ for the $X$, $Y$, $Z$-dimensions respectively, so that  
\[
\sum_{r=1}^6 \splresXt^{(r)} A_{t,r}= \splresXt, \quad \sum_{r=1}^6 \splresYt^{(r)} A_{t,r}= \splresYt, 
\quad 
\sum_{r=1}^6 \splresZt^{(r)} A_{t,r}= \splresZt. 
\]
These complete split distributions also need to satisfy conditions in \cref{rmk:constituent:assumptions_on_complete_split_dist}.

Then, we only keep level-$1$ blocks that are consistent with these complete split distributions. More specifically, we keep a level-$1$ $X$-block only if for all $t \in [s']$, $r \in [6]$, its corresponding portion in the $r$-th region of the $t$-th part is $\eps$-approximate consistent with $\splresXt^{(r)}$. We similarly zero out level-$1$ $Y$ and $Z$-blocks. 

The following claim shows the structure of the remaining tensor. We omit its proof as it is simple and similar to \cite[Claim 6.4]{VXXZ24}. 

\begin{claim}
  \label{cl:dividing_into_region_zero_out}
  After the above zero-out, we obtain a tensor that is isomorphic to
  \[
    \bigotimes_{r=1}^6 \bigotimes_{t = 1}^{s'} T_{i_t, j_t, k_t}^{\otimes A_{t,r} n_t}[\splresXt^{(r)}, \splresYt^{(r)}, \splresZt^{(r)}, \eps].
  \]
\end{claim}

In the remainder of this section, we will focus on the first region of the tensor, which we denote as 
\[ \T^{(1)} \defeq \bigotimes_{t = 1}^{s'} T_{i_t, j_t, k_t}^{\otimes A_{t,1} n_t}[\splresXt^{(1)}, \splresYt^{(1)}, \splresZt^{(1)}, \eps]. \]
We will omit the superscript $(1)$ from now on. 

\subsection{More Asymmetric Hashing}

Next, we apply more asymmetric hashing. Recall that for each $t\in [s']$, $\alpha_t$ is a distribution over the set $\{(i', j', k') \in \mathbb{Z}_{\ge 0}^3 \mid i' + j' + k' = 2^{\lvl-1}\}$ with marginal distributions on the $X$, $Y$, $Z$-dimensions equal to $\splonelevelXt(i', i_t \!-\! i')$, $\splonelevelYt(j', j_t \!-\! j')$, $\splonelevelZt(k', k_t \!-\! k')$, respectively. We zero out $X$, $Y$, $Z$-blocks that are not consistent with the distributions $\{\gamma_{\itX, t}\}_t, \{\gamma_{\itY, t}\}_t, \{\gamma_{\itZ, t}\}_t$. That is, if there exists some $t \in [s']$ where the $t$-th part of the level-$(\lvl - 1)$ index sequence of some $X$-block is not consistent with $\gamma_{\itX, t}$, we zero out this $X$-block. We zero out $Y$ and $Z$-blocks similarly with respect to $\gamma_{\itY,t}$ and $\gamma_{\itZ,t}$. 

We define the following quantities:
\begin{itemize}
    \item $\numxblock, \numyblock, \numzblock$: the number of remaining level-$(\lvl-1)$ $X$, $Y$, $Z$-blocks, respectively.
    \item  $\numalpha$: the number of remaining block triples  consistent with $\{\alpha_t\}_{t \in [s']}$.
    \item $\numtriple$: the number of remaining level-$(\lvl-1)$ block triples.
\end{itemize}
These quantities can be approximated as in the following claim:

\begin{claim}\label{claim:constituent-hash-quantities}
\EquationOnSameLine{\numxblock = 2^{\sum_{t} H(\splonelevelXt) \cdot A_{t,1} n_t \pm o(n)}, \hfill\; \numyblock = 2^{\sum_{t} H(\splonelevelYt) \cdot A_{t,1} n_t \pm o(n)}, \hfill\; \numzblock = 2^{\sum_{t} H(\splonelevelZt) \cdot A_{t,1} n_t \pm o(n)},}
\[\numalpha = 2^{\sum_t H(\alpha_t) \cdot A_{t,1} n_t \pm o(n)}, \quad \numtriple = 2^{\sum_t (H(\alpha_t) + P_{\alpha, t}) \cdot A_{t,1} n_t \pm o(n)}.\]
\end{claim}

Now we apply the standard hashing procedure first used in \cite{cw90}. Let $M \in [M_0, 2M_0]$ be a prime number for some integer $M_0$  satisfying
\[M_0 \ge 8\cdot \frac{\numtriple}{\numxblock}.\]

We pick independent elements $b_0, \{w_p\}_{p=0}^{2n} \in \Z_M$ uniformly at random, and use the hash functions $\hashx, \hashy, \hashz : \{0, \ldots, 2^{\lvl-1}\}^{2n} \rightarrow \Z_M$ defined as:
\begin{align*}
    \hashx(I) &= b_0 + \left(\sum_{p=1}^{2n} w_p \cdot I_p \right) \bmod M,\\
    \hashy(J) &= b_0 + \left(w_0 + \sum_{p=1}^{2n} w_p \cdot J_p \right) \bmod M,\\
    \hashz(K) &= b_0 + \frac{1}{2}\left(w_0+\sum_{p=1}^{2n} w_p \cdot (2^{\lvl-1} - K_p) \right) \bmod M.
\end{align*}

Then, let $B\subseteq \Z_M$ be the Salem-Spencer set without any $3$-term arithmetic progressions from \cref{thm:salemspencer}, where $|B| = M^{1-o(1)}$. We zero out all the level-$(\lvl-1)$ blocks $X_I$ with $\hashx(I) \notin B$, $Y_J$ with $\hashy(J) \notin B$, and $Z_K$ with $\hashz(K) \notin B$. Similar to before, now all the remaining block triples $X_I Y_J Z_K$ must have $\hashx(I) = \hashy(J) = \hashz(K) = b$ for some $b \in B$. 

If there are two remaining level-$(\lvl-1)$ triples $X_I Y_J Z_K$ and $X_I Y_{J'}Z_{K'}$ sharing the same $X$-block that are hashed to the same hash value $b \in B$, we zero out $X_I$. Then for every level-$(\lvl-1)$ block $X_I$, we check whether the unique triple containing it is consistent with $\{\alpha_t\}_{t \in [s']}$; if not, we zero out $X_I$. So every level-$(\ell-1)$ $X$-block is contained in a unique level-$(\lvl-1)$ triple $X_I Y_J Z_K$ that is consistent with $\{\alpha_t\}_{t \in [s']}$. We call the tensor after this step $\THash$.

The procedure we described above satisfies the following properties.

\begin{lemma}[Properties of more asymmetric hashing]\label{lem:more-asym-hash-constituent}

The above described procedure and its output $\THash$ satisfy the following:

\begin{enumerate}
    \item 
    \label{item:lem:more-asym-hash:constituent:item1}
    \textup{(Implicit in \cite{cw90}, see also \cite{duan2023})} For any level-$(\lvl-1)$ block triple $X_I Y_J Z_K\in \T$ and every bucket $b\in \{0,\dots, M-1\}$, we have 
    \[\Pr\Bk{\tall \hashx(I) = \hashy(J) = \hashz(K) = b} = \frac{1}{M^2}.\]
    Furthermore, for any $b \in \{0, \ldots, M - 1\}$, we have that any two different block triples $X_I Y_J Z_K, \allowbreak X_I Y_{J'} Z_{K'} \in \T$ that share the same $X$-block satisfy
    \[\Pr\Bk{\tall \hashx(I) = \hashy(J') = \hashz(K') = b \;\middle\vert\; \hashx(I) = \hashy(J) = \hashz(K) = b} = \frac{1}{M}.\]
    The same holds for different blocks that share the same $Y$-block or $Z$-block.

    \item
    \label{item:lem:more-asym-hash:constituent:item2}
    \textup{(Similar to {\cite[Claim 5.6]{VXXZ24}})} For every $b \in B$ and every level-$(\lvl-1)$ block triple $X_I Y_J Z_K \in \T$ consistent with $\alpha$, we have
    \[\Pr\Bk{X_IY_JZ_K\in \THash \;\middle\vert\; \hashx(I) = \hashy(J) = \hashz(K) = b}\ge \frac{3}{4}.\]

    \item 
    \label{item:lem:more-asym-hash:constituent:item3}
    \textup{({\cite[Claim 5.7]{VXXZ24}})} \[\E[\textup{number of level-$(\lvl-1)$ triples in $\THash$}] \ge \numalpha \cdot M_0^{-1-o(1)}.\]
\end{enumerate}
\end{lemma}

\subsection{\texorpdfstring{\boldmath$Y$}{Y}-Compatibility Zero-Out}
Let
\[ S^{(I, J, K)}_{t, i', j', k'} \defeq \{p \textup{ is in the $t$-th part} \mid I_p = i', J_p = j', K_p = k'\}, \]
and
\[S^{(J)}_{t, *, j', *} := \{p \textup{ is in the $t$-th part} \mid J_p = j'\}, \quad S^{(K)}_{t, *, *, k'} := \{p \textup{ is in the $t$-th part} \mid K_p = k'\}. \]
If clear from the context, we will drop the superscript $(I, J, K)$, $(J)$, or $(K)$.\footnote{If we follow our general rules for notation described at the beginning of \cref{sec:constituent-tensor} strictly, $S_{t,*,j',*}^{(J)}$ and $S_{t,*,*,k'}^{(K)}$ would be denoted as $S_{t,*,j',*}^{(I, J, K)}$ and $S_{t,*,*,k'}^{(I, J, K)}$, but notice that the extra superscripts can be dropped as they do not affect the values of $S_{t,*,j',*}^{(I, J, K)}$ and $S_{t,*,*,k'}^{(I, J, K)}$.
}

Recall that in $\THash$, every $X$-block $X_I$ is in a unique block triple. Thus, given $X_I$, we can uniquely determine the block triple $X_I Y_J Z_K$ containing it. 
So we can zero out a level-$1$ block $X_{\hat{I}} \in X_I$ if there exist $t, i', j', k'$ such that $\split(\hat I, S_{t, i',j',k'}) \ne \splres_{\itX, t, i', j', k'}$.

The goal of the $Y$-compatibility zero-out step is to ensure each level-$1$ Y-block belongs to a unique block triple. 

\subsubsection{\texorpdfstring{\boldmath$Y$}{Y}-Compatibility Zero-Out I}

For every level-1 $Y$-block $Y_{\hat J}$, if there is some $j'$ where $\split(\hat J, S_{t, *, j', *}) \ne \splresavg_{\itY, t, *,j',*}$, then we zero out $Y_{\hat J}$. We call the tensor after this zeroing out $\TYComp$. 

Next, we define the notion of compatibility for level-1 $Y$-blocks.

\begin{definition}[$Y$-Compatibility]
  \label{def:constituent:Y-compatibility}
  Given a level-$(\lvl-1)$ block triple $X_IY_J Z_K$ and a level-$1$ block $Y_{\hat J} \in Y_J$, we say $Y_{\hat J}$ is compatible with $X_IY_J Z_K$ if
  \begin{enumerate}
  \item 
    \label{item:constituent:Y-compatibility1} 
    
    For every $t \in [s']$ and every $(i', j', k') \in \Z^3_{\ge 0} \cap ([0, i_t] \times [0, j_t] \times [0, k_t])$ with $i' + j' + k'=2^{\lvl - 1}$ and $k' = 0$, $\split(\hat J, S_{t, i',j',k'}) = \splres_{\itY,t, i',j',k'}$.
  \item
    \label{item:constituent:Y-compatibility2}
    For every $t \in [s']$ and $j' \in \{0, 1, \ldots, \min\{2^{\lvl-1}, j_t\}\}$,  $\split(\hat J, S_{t, *,j',*}) = \splresavg_{\itY,t,  *,j',*}$.
  \end{enumerate}
\end{definition}

\begin{claim}
  \label{cl:constituent:Y-compatible}
  In $\TYComp$, for every level-1 block triple $X_{\hat I}Y_{\hat J}Z_{\hat K}$ and the level-$(\lvl-1)$ block triple $X_I Y_J Z_K$ that contains it, $Y_{\hat J}$ is compatible with $X_I Y_J Z_K$.
\end{claim}

The proof is similar to the proof of \cref{cl:global:Y-compatible}. 

\subsubsection{\texorpdfstring{\boldmath$Y$}{Y}-Compatibility Zero-Out II: Unique Triple}

After the previous step, every remaining level-1 $Y$-block $Y_{\hat J}$ is compatible with all the level-$(\lvl-1)$ block triples containing it. In this step, if $Y_{\hat J}$ is compatible with more than one level-$(\lvl-1)$ block triple, we zero it out. At this point, every level-1 $Y$-block belongs to a unique level-$(\lvl-1)$ block triple with which the block is compatible.

\subsubsection{\texorpdfstring{\boldmath$Y$}{Y}-Usefulness Zero-Out}

After the previous step, each remaining level-1 block $Y_{\hat J}$ will belong to a unique block triple $X_I Y_J Z_K$, so given $Y_{\hat J}$, $S_{t, i', j', k'}$ is well-defined for every $t, i', j', k'$. Hence, we can zero out $Y_{\hat J}$ such that there exist some $t, i', j', k'$ where $\split(\hat J, S_{t, i', j', k'}) \ne \splres_{\itY, t, i', j', k'}$.

Similar to before, we define the following notion of usefulness. 

\begin{definition}[$Y$-Usefulness]
For a level-$(\lvl-1)$ block triple $X_I Y_J Z_K$ and a level-1 block $Y_{\hat J} \in Y_J$, we say $Y_{\hat J}$ is useful for $X_I Y_J Z_K$ if $\split(\hat J, S_{t, i', j', k'}) = \splres_{\itY, t, i', j', k'}$ for every $t, i', j', k'$. 
\end{definition}

Using the above definition, the zero-out in this step can be equivalently described as follows: We zero out all level-1 block $Y_{\hat J}$ that is not useful for the unique triple containing it. 

After this step, we call the tensor $\TYUseful$. 

\subsection{\texorpdfstring{\boldmath$Z$}{Z}-Compatibility Zero-Out}

Similar to before, now we perform compatibility zero-outs for $Z$-blocks. 

\subsubsection{\texorpdfstring{\boldmath$Z$}{Z}-Compatibility Zero-Out I}

For every level-1 block $Z_{\hat K} \in Z_K$, we zero out $Z_{\hat K}$ if there exist $t, k'$ where $\split(\hat K, S_{t, *,*,k'}) \ne \splresavg_{\itZ, t, *,*,k'}$.
We call the remaining tensor $\TZComp$.

Now we are ready to define the notion of compatibility for level-1 $Z$-blocks which remains identical to the definition in \cite{VXXZ24}. 

\begin{definition}[$Z$-Compatibility]
  \label{def:constituent:compatibility}

  For a level-$(\lvl-1)$ block triple $X_IY_J Z_K$ and a level-$1$ block $Z_{\hat{K}} \in Z_K$, we say $Z_{\hat K}$ is compatible with $X_IY_J Z_K$ if
  \begin{enumerate}
  \item 
    \label{item:constituent:compatibility1} For every $t$ and every $(i', j', k') \in \mathbb{Z}_{\ge 0}^3 \cap [0, i_t] \times [0, j_t] \times [0, k_t]$ with $ i' + j' + k' = 2^{\lvl-1}$, $i' = 0 \text{ or } j' = 0$, there is $\split(\hat K, S_{t, i',j',k'}) = \splres_{\itZ,t, i',j',k'}$.
  \item
    \label{item:constituent:compatibility2}
    For every $t$ and every index $k' \in \{0, 1, \ldots, \min\{2^{\lvl-1}, k_t\}\}$,  $\split(\hat K, S_{t, *,*,k'}) = \splresavg_{\itZ,t,  *,*,k'}$.
  \end{enumerate}
\end{definition}

\begin{claim}
  \label{cl:constituent:compatible}
  In $\TZComp$, for every remaining level-1 block triple $X_{\hat I}Y_{\hat J}Z_{\hat K}$ and the level-$(\lvl-1)$ block triple $X_I Y_J Z_K$ that contains it, $Z_{\hat K}$ is compatible with $X_I Y_J Z_K$.
\end{claim}

The proof of this claim is the same as the proof of Claim 6.10 in \cite{VXXZ24}. 

\subsubsection{\texorpdfstring{\boldmath$Z$}{Z}-Compatibility Zero-Out II: Unique Triple}

In this step, we zero out every level-$1$ $Z$-block $Z_{\hat K}$ that is compatible with more than one level-$(\lvl-1)$ block triple and they become holes. After this step, each remaining level-$1$ $Z$-block $Z_{\hat{K}}\in Z_K$ belongs to a unique level-$(\lvl-1)$ triple $X_I Y_J Z_K$ containing it, and the block is compatible with that triple.

\subsubsection{\texorpdfstring{\boldmath$Z$}{Z}-Usefulness Zero-Out}

Next, we further zero out some level-1 $Z$-blocks using the following definition of usefulness.

\begin{definition}[$Z$-Usefulness]
For a level-$(\lvl-1)$ block triple $X_I Y_J Z_K$ and a level-1 block $Z_{\hat K} \in Z_K$, we say $Z_{\hat K}$ is useful for $X_I Y_J Z_K$ if for every $t, i', j', k'$, we have $\split(\hat K, S_{t, i',j',k'}) = \splres_{\itZ,t, i',j',k'}$.
\end{definition}

For each $Z_{\hat K}$, it appears in a unique triple $X_I Y_J Z_K$ by the previous zero-out. If $Z_{\hat K}$ is not useful for this triple, we zero out $Z_{\hat K}$. We call the result tensor $\TZUseful$.

\subsection{Fixing Holes}

Eventually, we want to ensure that the subtensor of $\TZUseful$ restricted to each level-$(\lvl-1)$ block triple $X_I Y_J Z_K$ is isomorphic to 
\[ \mathcal{T}^* = \bigotimes_{t \in [s']} \; \bigotimes_{i'+j'+k' = 2^{\lvl-1}} T_{i',j',k'}^{\otimes A_{t,1} \cdot (\alpha_t(i', j', k')+\alpha_t(i_t-i', j_t-j', k_t-k')) \cdot n_t} [\splres_{\itX, t, i', j', k'}, \splres_{\itY, t, i', j', k'}, \splres_{\itZ, t, i', j', k'}]. \]

Similar to before, there are some holes caused by our degeneration process. The three types of holes are the following: 
\begin{itemize}
    \item The input of the constituent stage does not include all level-$1$ blocks, as we enforce some $\eps$-approximate complete split distribution on the input. As analyzed in the previous work \cite{VXXZ24}, the fraction of this type of holes can be upper bounded by $1/n^2$ in all three dimensions. The high-level intuition is that, if we take a random level-$1$ block from $\mathcal{T}^*$, it will likely be $\eps$-approximate consistent with the complete split distribution in the input. Thus, most level-$1$ blocks from $\mathcal{T}^*$ should appear in the input, which implies the fraction of this type of holes is small.
    \item The holes on $Y$-blocks caused by $Y$-compatibility zero-outs. These holes were not considered before, so the main focus of the analysis will be this case. 
    \item The holes on $Z$-blocks caused by $Z$-compatibility zero-outs. The analysis of these holes is similar to the previous work \cite{VXXZ24}, so we will omit most proofs of this case in the following. 
\end{itemize}

Next, we focus on holes on level-1 blocks caused by compatibility zero-outs. These are the fractions of (typical) $Y_{\hat{J}}$ and  $Z_{\hat{K}}$ that are compatible with multiple level-$(\lvl-1)$ triples. Similar to before, we define typicalness, $\pcompY$, and $\pcompZ$. 

\begin{definition}[$Y$-Typicalness]
  A level-1 $Y$-block $Y_{\hat{J}}$ in some level-$(\lvl-1)$ $Y$-block $Y_J$ is \emph{typical} if \\
  $\Vert\split(\hat J, S_{t, *,*,*}) - \splresavg_{\itY,t, *,*,*}\Vert_\infty \le \eps$ for every $t \in [s']$. 
\end{definition}

\begin{definition}[$Z$-Typicalness]
  A level-1 $Z$-block $Z_{\hat{K}}$ in some level-$(\lvl-1)$ $Z$-block $Z_K$ is \emph{typical} if \\
  $\Vert\split(\hat K, S_{t, *,*,*}) - \splresavg_{\itZ,t, *,*,*}\Vert_\infty \le \eps$ for every $t \in [s']$. 
\end{definition}

\begin{definition}[$\pcompY$]
  For fixed $Y_{\hat{J}}$ and $Y_J$ where $Y_{\hat{J}} \in Y_J$ and $\hat J$ has level-$\lvl$ complete split distributions $\midBK{\xiYt}_{t \in [s']}$, we define $\pcompY^*(\midBK{\xiYt}_{t \in [s']})$ as the probability that a uniformly random level-$(\lvl-1)$ block triple $X_I Y_J Z_K$ consistent with $\{\alpha_t\}_{t \in [s']}$ is compatible with $Y_{\hat{J}}$. We further define 
  \[\pcompY \, \defeq \max_{\substack{\midBK{\xiYt}_{t \in [s']} \,: \\ \norm{\xiYt - \splresYt}_\infty \le \eps \; \forall t}} \pcompY^*(\midBK{\xiYt}_{t \in [s']}).\]
\end{definition}

\begin{definition}[$\pcompZ$]
  For fixed $Z_{\hat{K}}$ and $Z_K$ where $Z_{\hat{K}} \in Z_K$ and $\hat K$ has level-$\lvl$ complete split distributions $\midBK{\xiZt}_{t \in [s']}$, we define $\pcompZ^*(\midBK{\xiZt}_{t \in [s']})$ as the probability that a uniformly random level-$(\lvl-1)$ block triple $X_I Y_J Z_K$ consistent with $\{\alpha_t\}_{t \in [s']}$ is compatible with $Z_{\hat{K}}$. We further define 
  \[\pcompZ \, \defeq \max_{\substack{\midBK{\xiZt}_{t \in [s']} \,: \\ \norm{\xiZt - \splresZt}_\infty \le \eps \; \forall t}} \pcompZ^*(\midBK{\xiZt}_{t \in [s']}).\]
\end{definition}

Similar to before, $\pcompY^*(\midBK{\xiYt}_{t \in [s']})$ does not depend on the choice of $Y_{\hat J}$ and $Y_J$ by symmetry, so it is well-defined, and thus is $\pcompY$. It is also the case for $\pcompZ^*(\midBK{\xiZt}_{t \in [s']})$ and $\pcompZ$.

Next, we upper bound the value of $\pcompY$. 

\begin{claim}
  The value of $\pcompY^*(\midBK{\xiYt}_{t \in [s']})$ is at most
  \[
    2^{\sum_{t \in [s']} \left( \eta_{\itY,t} - H(\xiYt) + H(\gamma_{\itY, t}) \right) A_{t,1} \cdot n_t \pm o(n)}.
  \]
  Furthermore,
  \[
    \pcompY \,\le\, 2^{\sum_{t \in [s']} \left( \eta_{\itY,t} - H(\splresYt) + H(\gamma_{\itY, t}) + \oeps(1) \right) A_{t,1} \cdot n_t + o(n)}.
    \numberthis \label{eq:constituent:pcomp_upper_bound}
  \]
\end{claim}

\begin{proof}
We define the following two quantities:
\begin{enumerate}
    \item[($P$)] The number of tuples $(I, J, K, \hat{J})$ where $X_I Y_J Z_K$ is consistent with $\{\alpha_t\}_{t \in [s']}$, $Y_{\hat J} \in Y_J$, and $\hat Y$ has complete split distributions $\midBK{\xiYt}_t$.
    \item[($Q$)] The number of tuples $(I, J, K, \hat{J})$ where $X_I Y_J Z_K$ is consistent with $\{\alpha_t\}_{t \in [s']}$, $Y_{\hat J} \in Y_J$, $\hat Y$ has complete split distributions $\midBK{\xiYt}_t$, and additionally $Y_{\hat J}$ is compatible with the triple $X_I Y_J Z_K$.

\end{enumerate}
Notice that by definition and by symmetry $\pcompY(\midBK{\xiYt}_{t \in [s']}) = Q/P$. 

We first compute the simpler quantity $P$, which is the number of level-$(\lvl-1)$ block triples $X_I Y_J Z_K$ consistent with $\{\alpha_t\}_{t \in [s']}$ (this quantity is $\numalpha = 2^{\sum_t (H(\alpha_t) \cdot A_{t, 1} \cdot n_t) \pm o(n)}$) times the number of $Y_{\hat J} \in Y_J$ consistent with $\midBK{\xiYt}_t$ (this quantity is $2^{\sum_t (H(\xiYt)-H(\splonelevelYt)) \cdot A_{t, 1} \cdot n_t \pm o(n)}$). Overall, 
\[
P = 2^{\sum_t (H(\alpha_t) + H(\xiYt)-H(\splonelevelYt)) \cdot A_{t, 1} \cdot n_t \pm o(n)}. 
\]

Next, we consider how to compute $Q$. In fact, we will show an upper bound on $Q$ (as we only need an upper bound on $\pcompY(\midBK{\xiYt}_{t \in [s']})$) by dropping the condition that ``$\hat J$ has complete split distributions $\midBK{\xiYt}_t$''.

Similar to before, we have the following claim (whose proof we omit):
\begin{claim}
    Given a level-$(\lvl-1)$ triple $X_IY_JZ_K$ that is consistent with $\{\alpha_t\}_{t \in [s']}$, a level-$1$ block $Y_{\hat J}$ with $Y_{\hat J}\in Y_J$ is compatible with $X_IY_JZ_K$ if and only if the following hold:
    \begin{enumerate}
    \item For every $t \in [s']$ and every $(i', j', k') \in \Z^3_{\ge 0} \cap ([0, i_t] \times [0, j_t] \times [0, k_t])$ with $i' + j' + k' = 2^{\lvl - 1}$ and $k' = 0$, there is $\split(\hat J, S_{t, i',j',k'}) = \splres_{\itY,t, i',j',k'}$.
    \item For every $t \in [s']$ and $j' \in \{0, \ldots, j_t\}$,  $\split(\hat J, S_{t, *, j', \+}) = \splresavg_{\itY, t, *, j', \+}$, where $S_{t, *, j', \+} = \bigcup_{i' \ge 0, k' > 0} S_{t, i', j', k'}$.
    \end{enumerate}
\end{claim}

The benefit of the above claim is that $\{S_{t, i',j',0}\}_{t, i', j'} \cup \{S_{t, *, j', \+}\}_{t, j'}$ forms a partition of all the indices in the first region of the $t$-th part of the interface tensor, after we fix some $X_I Y_J Z_K$ that is consistent with $\{\alpha_t\}_{t \in [s']}$. Thus, we can compute the possibilities of $\hat{J}$ on each of these subsets. 

For $S_{t, i',j',0}$, the number of possibilities of $\hat{J}$ is \[
2^{H(\splres_{\itY, t, i', j', 0}) \cdot (\alpha_t(i', j', 0) + \alpha_t(i_t - i', j_t - j', k_t - k')) \cdot A_{t, 1} \cdot n_t \pm o(n)}.
\]
For $S_{t, *, j', \+}$, the number of possibilities of $\hat{J}$ is \[
2^{H(\splresavg_{\itY, t, *, j', \+}) \cdot (\alpha_t(*, j', \+) + \alpha_t(*, j_t - j', \<)) \cdot A_{t, 1} \cdot n_t \pm o(n)}.
\]
Thus, The total number of possible $\hat{J}$ for a fixed triple $X_I Y_J Z_K$ is 
\begin{align*}
& 2^{\sum_t \left(\sum_{i', j'} H(\splres_{\itY, t, i', j', 0}) \cdot (\alpha_t(i', j', 0) + \alpha_t(i_t - i', j_t - j', k_t - k')) 
+ 
\sum_{j'} H(\splresavg_{\itY, t, *, j', \+}) \cdot (\alpha_t(*, j', \+) + \alpha_t(*, j_t - j', \<))
\right) \cdot A_{t, 1} \cdot n_t \pm o(n)} \\
={} & 2^{\sum_t \eta_{\itY, t} \cdot A_{t, 1} \cdot n_t \pm o(n)}.  
\end{align*}
Multiplying the above with the number of triples $X_I Y_J Z_K$, we get 
\[
Q \le 2^{\sum_t (H(\alpha_t) + \eta_{\itY, t}) \cdot A_{t, 1} \cdot n_t \pm o(n)}. 
\]
Finally, 
\[
\pcompY^*(\midBK{\xiYt}_{t \in [s']}) = Q / P \le 2^{\sum_t (\eta_{\itY, t} - H(\xiYt) + H(\splonelevelYt)) \cdot A_{t, 1} \cdot n_t \pm o(n)}
\]

The bound \eqref{eq:constituent:pcomp_upper_bound} on $\pcompY$ follows because the $L_\infty$ distance between $\midBK{\xiYt}_t$ and $\midBK{\splresYt}_t$ is at most $\eps$. 
\end{proof}

The following upper bound of $\pcompZ$ is the same as that in \cite{VXXZ24}, so we omit its proof.

\begin{claim}[{\cite[Claim 6.13]{VXXZ24}}]
  The value of $\pcompZ^*(\midBK{\xiZt}_{t \in [s']})$ is at most
  \[
    2^{\sum_{t \in [s']} \left( \lambda_{\itZ,t} - H(\xiZt) + H(\gamma_{\itZ, t}) \right) A_{t,1} \cdot n_t \pm o(n)}.
  \]
  Furthermore,
  \[
    \pcompZ \,\le\, 2^{\sum_{t \in [s']} \left( \lambda_{\itZ,t} - H(\splresZt) + H(\gamma_{\itZ, t}) + \oeps(1) \right) A_{t,1} \cdot n_t + o(n)}.
    \numberthis \label{eq:constituent:pcompZ_upper_bound}
  \]
\end{claim}

The proof of the following claim is essentially the same as that of \cref{cl:global:prob-of-holes}.
\begin{claim}
  \label{cl:constituent:prob-of-holes}
  For every $b \in B$, every level-$(\lvl-1)$ block triple $X_I Y_J Z_K$ consistent with $\{\alpha_t\}_{t \in [s']}$, and every typical $Z_{\hat{K}}\in Z_K$, the probability that $Z_{\hat{K}}$ is compatible with multiple level-$(\lvl-1)$ block triples in $\TZComp$ is at most
  \[ \frac{\numalpha \cdot \pcompZ}{\numzblock \cdot M_0}, \]
  conditioned on $\hashx(I) = \hashy(J) = \hashz(K) = b$.

  Similarly, for every $b \in B$, every level-$(\lvl-1)$ block triple $X_I Y_J Z_K$ consistent with $\{\alpha_t\}_{t \in [s']}$, and every typical $Y_{\hat{J}}\in Y_J$, the probability that $Y_{\hat{J}}$ is compatible with multiple level-$(\lvl-1)$ block triples in $\TYComp$ is at most
  \[ \frac{\numalpha \cdot \pcompY}{\numyblock \cdot M_0}, \]
  conditioned on $\hashx(I) = \hashy(J) = \hashz(K) = b$.
\end{claim}

Recall that we first require $M_0$ to be $\ge 8 \cdot \frac{\numtriple}{\numxblock}$. Now, we finalize our choice of $M_0$ as:
\begin{gather*}
M_0 = \max\left\{\frac{8\numtriple}{\numxblock}, \; \frac{\numalpha \cdot \pcompY}{\numyblock} \cdot n^2, \; \frac{\numalpha \cdot \pcompZ}{\numzblock} \cdot n^2\right\} \\
\le 2^{\max\BK{\sum_t (H(\alpha_t) - P_{\alpha, t} - H(\gamma_{\itX, t})) A_{t,1} \cdot n_t, \;\; \sum_t (H(\alpha_t) + \eta_{\itY, t} - H(\splresYt)) A_{t,1} \cdot n_t, \;\; \sum_t (H(\alpha_t) + \lambda_{\itZ,t} - H(\splresZt)) A_{t,1} \cdot n_t} + o(n)}.
\end{gather*}

We consider the fraction of holes in $Y$-variables caused by $Y$-compatibility zero-outs. By \cref{cl:constituent:prob-of-holes} and by the upper bound on $M_0$, the probability that each typical $Y_{\hat{J}}$ is compatible with multiple triples is at most $\frac{1}{n^2}$; the same also holds for typical level-$1$ $Z$-blocks. As discussed earlier at the beginning of this subsection, there is another type of holes caused by the input complete split distributions, whose fraction is $O(1/n^2)$. Overall, we expect to get $\numalpha \cdot M^{-1-o(1)}$ copies of $\T^*$ whose fraction of holes is $O(1/n^2)$. By \cref{thm:fix-holes}, we can degenerate them into $\numalpha \cdot M^{-1-o(1)}$ independent copies of unbroken $\T^*$. 

\subsection{Summary}

In conclusion, the above algorithm degenerates $\bigotimes_{t = 1}^{s'} T_{i_t, j_t, k_t}^{\otimes A_{t,1} n_t}\bigBk{\splresXt^{(1)}, \splresYt^{(1)}, \splresZt^{(1)}, \eps}$ into 
\begin{align*}
& \numalpha \cdot M_0^{-1-o(1)}\\
\ge{} & 2^{\min\left\{\sum_{t \in [s']} A_{t,1} \cdot n_t \cdot \left( H(\splonelevelXt[t]^{(1)}) - P_{\alpha, t}^{(1)}\right), \;\; \sum_{t \in [s']} A_{t,1} \cdot n_t \cdot \left(H(\splres_{\itY, t}^{(1)}) - \eta_{\itY, t}^{(1)}\right), \;\; \sum_{t \in [s']} A_{t,1} \cdot n_t \cdot \left(H(\splres_{\itZ, t}^{(1)}) - \lambda_{\itZ, t}^{(1)}\right)\right\} - \oeps(n) - o(n)}
\end{align*}
independent copies of a level-$(\lvl-1)$ interface tensor $\T^*$ with parameter list
\[
  \begin{aligned}
    \Big\{ \Big(
    A_{t,1} \cdot n_t \cdot \bigbk{\alpha^{(1)}_t(i', j', k')+\alpha^{(1)}_t(i_t \!-\! i', j_t \!-\! j', k_t \!-\! k')}, & \\
    i', \; j', \; k', \; \splres^{(1)}_{\itX, t, i', j', k'}, \; \splres^{(1)}_{\itY, t, i', j', k'}, \; \splres^{(1)}_{\itZ, t, i', j', k'} & \Big) \Big\} _{t \in [s'], \, i' + j' + k' = 2^{\lvl-1}}.
  \end{aligned}
\]

The above algorithm was described for the first region; the algorithm for other regions is identical except that we permute the roles of the $X$, $Y$, $Z$-dimensions. In the end, we take the tensor product over the output tensor of the algorithm over all $6$ regions.

\newcommand{\numgrow}{V}
\newcommand{\numgrowl}[1][\l]{\numgrow_{#1}}
\newcommand{\interl}[1][\l]{\T_{#1}}

\section{Numerical Result}
\label{sec:numerical}

The way we combine theorems from previous sections to obtain numerical bounds on $\omega(1, \kappa, 1)$ is similar to previous works (e.g., \cite{VXXZ24}). Let $\l^* > 0$ be an integer denoting the highest level we hope to analyze, and let $N = 2^{\l^* - 1} \cdot n$. We repeatedly apply \cref{thm:global-stage-with-eps,thm:consituent_MM_terms,thm:constituent-stage-with-eps} to degenerate a direct sum of $2^{o(n)}$ copies of $\CW_q^{\otimes N} \equiv \bigbk{\CW_q^{\otimes 2^{\l^* - 1}}}^{\otimes n}$ into a direct sum of matrix multiplication tensors $\angbk{a, a^{\kappa}, a}$, shown in \cref{alg:framework}.

\begin{figure}[ht]
  \centering
  \begin{tcolorbox}
    \captionof{algocf}{Procedure of degeneration (similar to \cite{VXXZ24})}{\label{alg:framework}}
    Let $\eps > 0$ be a fixed constant and $\l^* > 0$ be an integer.
    \begin{enumerate}
    \item Degenerate a direct sum of $2^{o(n)}$ independent copies of $\bigbk{\CW_q^{\otimes 2^{\l^* - 1}}}^{\otimes n}$ into $\numgrowl[\l^*]$ independent copies of a level-$\l^*$ $(\eps \cdot 3^{\l^*})$-interface tensor $\interl[\l^*]$, where the number of copies $\numgrowl[\l^*]$ and the parameter list of $\interl[\l^*]$ are given by \cref{thm:global-stage-with-eps} and \cref{prop:global-stage-no-eps}.
    \item For each $\l = \l^*, \ldots, 2$:
      \begin{itemize}
      \item Degenerate \emph{every} $2^{o(n)}$ copies of the level-$\l$ $(\eps \cdot 3^{\l})$-interface tensor $\interl[\l]$ into a direct sum of $\numgrowl[\l - 1]$ copies of the Kronecker product of a level-$(\l - 1)$ $(\eps \cdot 3^{\l - 1})$-interface tensor $\interl[\l - 1]$ and a matrix multiplication tensor $\angbk{a_\l, b_\l, c_\l}$. The number of copies $\numgrowl[\l - 1]$, the parameter list of $\interl[\l - 1]$, and the size of the matrix multiplication tensor $\angbk{a_\l, b_\l, c_\l}$ are all given by \cref{thm:constituent-stage-with-eps} and \cref{prop:constituent-stage-no-eps}.
      \end{itemize}
    \item Degenerate the level-$1$ $3\eps$-interface tensor $\interl[1]$ into a matrix multiplication tensor $\angbk{a_1, b_1, c_1}$ according to \cref{thm:consituent_MM_terms}.
    \item The above steps will produce a direct sum of $\numgrow \defeq \prod_{\l = 1}^{\l^*} \numgrowl[\l]$ copies of $\angbk{A, B, C} \equiv \bigotimes_{\l = 1}^{\l^*} \angbk{a_\l, b_\l, c_\l}$.
    \end{enumerate}
    Letting $n \to \infty$ and applying Sch{\"o}nhage's asymptotic sum inequality 
    (\cref{thm:schonhage-ineq-rect}) on the above degeneration result in a bound on $\omega(1, \kappa, 1)$ which depends on $\eps$. Then, we let $\eps \to 0$ to obtain the bound $\omega(1, \kappa, 1) \le \omega'$, as long as
    \[
      \lim_{\eps \to 0} \lim_{n \to \infty} \numgrow^{1/n} \cdot \min\bigBK{A, B^{1/\kappa}, C}^{\omega'/n} \ge (q + 2)^{2^{\l^* - 1}}.
      \numberthis \label{eq:asymptotic_sum_inequality_in_algorithm}
    \]
  \end{tcolorbox}
  \vspace{-1.5em}
\end{figure}

\bigskip

Every step of degeneration in \cref{alg:framework} takes a set of parameters, which includes the distribution $\alpha$ over constituent tensors, the fractions of tensor powers $A_1, \ldots, A_6$ assigned to six regions, etc. Given an assignment to these parameters, we can calculate
\[
  \lim_{\eps \to 0} \lim_{n \to \infty} \numgrowl^{1/n}, \quad \lim_{\eps \to 0} \lim_{n \to \infty} a_\l^{1/n}, \quad \lim_{\eps \to 0} \lim_{n \to \infty} b_\l^{1/n}, \quad \lim_{\eps \to 0} \lim_{n \to \infty} c_\l^{1/n}
\]
by \cref{thm:global-stage-with-eps,thm:consituent_MM_terms,thm:constituent-stage-with-eps}. One can substitute them into \eqref{eq:asymptotic_sum_inequality_in_algorithm} to verify a claimed bound on $\omega(1, \kappa, 1)$.

\paragraph{Optimization.}

Finding the set of parameters that leads to the best bound for $\omega(1, \kappa, 1)$ can be formulated as an optimization problem:
\[
  \begin{array}{cl}
    \textup{minimize} & \qquad \omega' \\
    \textup{subject to} & \textup{all constraints in \cref{thm:global-stage-with-eps,thm:consituent_MM_terms,thm:constituent-stage-with-eps}} \\
                      & \displaystyle \lim_{\eps \to 0} \lim_{n \to \infty} \numgrow^{1/n} \cdot \min\bigBK{A, B^{1/\kappa}, C}^{\omega'/n} \ge (q + 2)^{2^{\l^* - 1}}.
  \end{array}
  \numberthis \label{eq:optimization_problem}
\]
To optimize \eqref{eq:optimization_problem}, we utilize the software package SNOPT \cite{SNOPT}, which employs a \emph{sequential quadratic programming (SQP)} algorithm. Although SNOPT does not guarantee finding the optimal solution of the optimization problem, any feasible solution that provides a reasonably good bound on $\omega(1, \kappa, 1)$ is acceptable.

\cite[Section 8]{VXXZ24} mentioned several tricks in the optimization, including setting specific initial points, using the exponential form of Lagrange multiplier constraints, and transforming minimum functions into linear inequalities. We use all of them in the same way. See \cite{VXXZ24} for details.

\paragraph{Numerical results.}

We wrote a MATLAB \cite{MATLAB2022} program which uses SNOPT \cite{SNOPT} to solve the optimization problem \eqref{eq:optimization_problem}. By executing the program with various values of $\kappa$, we derived upper bounds for $\omega(1, \kappa, 1)$, as listed in \cref{table:result} (except for $\kappa = 1$, where the table states the improved bound from \cite{alphaevolve-note}; see below). All these bounds were derived by analyzing $\CW_5^{\otimes 4}$, the fourth power\footnote{We have no doubt that analyzing the eighth power and beyond would yield further improvements. However, the optimization solvers are already very slow in solving our program for the 4th power, so solving the 8th power program seems infeasible with the current solver SNOPT.} of the CW tensor with $q = 5$. In particular, we showed that $\omega \le 2.371339$ and $\mu \le 0.527500$.\footnote{Unfortunately, we could not obtain a bound on $\alpha$ via SNOPT, as it crashed due to potential numerical instabilities: the number of parameters in the optimization problem is significantly larger than in previous works, hence the numerical errors accumulated more and crashed a subroutine in SNOPT.}
The code and the sets of parameters are available at \url{https://osf.io/mw5ak/?view_only=5769f03789354793b61e11aac4dd85dd}.

Subsequently, \cite{alphaevolve-note} apply modern optimization and machine learning techniques to the optimization problem \eqref{eq:optimization_problem}, finding better solutions and scaling the analysis up to the eighth power of the CW tensor; the resulting parameters, substituted into the same analysis, yield $\omega < 2.371177$.

\bibliographystyle{alpha}
\bibliography{ref}
\end{document}